\documentclass{elsarticle}
\usepackage{a4wide,amsmath,amssymb,graphics,epsfig}
 \allowdisplaybreaks[2]
\newcommand{\hatA}[0]{\hat{A}}
\newcommand{\mons}[0]{{merinthons\ }}
\newcommand{\monsp}[0]{{merinthons}}

\newcommand{\mon}[0]{{merinthon\ }}
\newcommand{\monp}[0]{{merinthon}}

\newcommand{\Mon}[0]{{Merinthon\ }}

\newcommand{\spa}[0]{\phantom{spa}}
\newcommand{\tr}[0]{\textnormal{tr}}

\newtheorem{theorem}{Theorem}[section]
\newtheorem{lemma}[theorem]{Lemma}

\newenvironment{proof}[1][Proof]{\begin{trivlist}
\item[\hskip \labelsep {\bfseries #1}]}{\qed\end{trivlist}}

\newenvironment{remark}[1][]{\begin{list}{\textbf{Remark}}{\setlength{\leftmargin}{1cm}}
\item  #1}{\end{list}}

\usepackage{grffile}
\begin{document}

\begin{frontmatter}
\title{The Static Quark Potential from the Gauge Independent Abelian Decomposition}

\author[a]{Nigel Cundy}
\author[b,c]{Y. M. Cho}
\author[a]{Weonjong Lee}
\author[a]{Jaehoon Leem}
\address[a]{    Lattice Gauge Theory Research Center, FPRD, and CTP, Department of Physics \&
    Astronomy,\\ Seoul National University, Seoul, 151-747, South Korea}
\address[b]{Administration Building 310-4, Konkuk University,
Seoul 143-701, Korea}
\address[c]{ Department of Physics \&
    Astronomy,\\ Seoul National University, Seoul, 151-747, South Korea}


\begin{abstract}
We investigate the relationship between colour confinement and the gauge independent Cho-Duan-Ge Abelian decomposition.
 The decomposition is defined in terms of a colour field $n$; the principle novelty of our study is that we have used a unique definition of this field in terms of the eigenvectors of the Wilson Loop. This allows us to establish an equivalence between the path ordered integral of the non-Abelian gauge fields with an integral over an Abelian restricted gauge field which is tractable both theoretically and numerically in lattice QCD. We circumvent path ordering without requiring an additional path integral. By using Stokes' theorem, we can compute the Wilson Loop in terms of a surface integral over a restricted field strength, and show that the restricted field strength may be dominated by certain structures, which occur when one of the quantities parametrising the colour field $n$ winds itself around a non-analyticity in the colour field. If they exist, these structures will lead to a area law scaling for the Wilson Loop and provide a mechanism for quark confinement. Unlike most studies of confinement using the Abelian decomposition, we do not rely on a dual-Meissner effect to create the inter-quark potential.

We search for these structures in quenched lattice QCD. We perform the Abelian decomposition, and compare the electric and magnetic fields with the patterns expected theoretically. We find that the restricted field strength is dominated by objects which may be  peaks a single lattice spacing in size or extended string-like lines of electromagnetic flux. The objects are not isolated monopoles, as they generate electric fields in addition to magnetic fields, and the fields are not spherically symmetric, but may be either caused by a monopole/anti-monopole condensate, some other types of topological objects or a combination of these. Removing these peaks removes the area law scaling of the string tension, suggesting that they are responsible for confinement.
\end{abstract}
\begin{keyword}
Quantum chromodynamics \sep Lattice gauge theory \sep Confinement of Quarks
\PACS  12.38.-t \sep 12.38.Aw \sep 11.15.Ha
\end{keyword}
\end{frontmatter}
\section{Introduction}
An enduring problem in QCD is to find the mechanism which causes quark confinement, which is known to be non-perturbative in its origin. Although several models have been proposed -- for example, center vortices~\cite{Vortices}, and a dual Meissner effect due to magnetic monopoles~\cite{Monopoles1,Monopoles2,Mandelstam:1976,thooft:1976} -- there has not yet been a convincing demonstration that any of them are correct.
In this initial study, we investigate the Cho-Duan-Ge (CDG) Abelian decomposition (sometimes referred to as the Cho-Faddeev-Niemi decomposition)~\cite{Cho:1980,Cho:1981,F-N:98,Shabanov:1999,Duan:1979}, and concentrate on one of the signals of confinement, a linear static quark potential. Unlike Dirac and 't Hooft (Maximum Abelian Gauge) monopoles, the CDG decomposition allows for monopole solutions while respecting the gauge symmetry and does not require a singular gauge field or an additional Higgs field. The decomposition is constructed from a colour field, $n$, which may be built from a SU($N_C$) matrix $\theta$, where SU($N_C$) is the gauge group of QCD. Each $n$ defines a different decomposition, and an important question is the best way of choosing this field.
Recent work~\cite{Kondo:2008su,Shibata:2009af,Kondo:2010pt,Kondo:2005eq} has demonstrated that the magnetic part of the field strength dominates the confining string, using a decomposition constructed from one particular choice of $\theta \in
SU(N_C)/U(N_C-1)$; however in this case only one
of the possible $N_C-1$ types of monopole is visible.

 Here we consider a different choice of $\theta \in SU(N_C)/(U(1))^{N_C-1}$. Our initial goal was to investigate whether monopoles apparent in this construction may also lead to confinement\footnote{A different choice of a $SU(N_C)/(U(1))^{N_C-1}$ Abelian decomposition was described in~\cite{Shibata:2007pi},  but without a discussion of the relationship to quark confinement.}. Concentrating on the Wilson Loop,
an observable used to measure the string tension, we show that the path ordering may be removed by diagonalising the gauge links along the Wilson Loop by a $SU(N_C)/(U(1))^{N_C-1}$ field $\theta$.
This, in principle, allows us to use Stokes' theorem to express the Wilson Loop in terms of a surface integral over an Abelian restricted gauge field strength tensor.
 We use the CDG decomposition, constructed from a colour field $n^j = \theta\lambda^j\theta^\dagger$ for a diagonal Gell-Mann matrix $\lambda^j$, to find a consistent construction of this restricted field across space-time, and not just for those gauge links along the Wilson loop where it was originally defined.
Our relationship for the string tension in terms of this restricted field is exact: we do not require any approximations or additional path integrals.
We search for topological structures in this field strength, and find that discontinuities in the colour field related to the winding of one of the parameters describing it around a particular location may lead to large `electric' and `magnetic' fields (deriving the terms from an obvious analogy to electromagnetism) in both the restricted field strength and the part of the field strength derived from the colour field; in particular the component of the field strength in the plane of the Wilson Loop may exhibit structures, decaying with a power of the distance in all directions, which would lead to an area law scaling of the Wilson Loop.
We label these structures \monsp. We also examine the other components of the electromagnetic field strength and find that if certain simplifying assumptions are reliable they will be dominated either by similar peaks or lines of electromagnetic flux.
We do not show that this winding occurs in practice, but if it does it provides an explanation for the confining potential. We also discuss how string breaking might arise in this picture at large distances.

The second part of this article checks whether these structures are present using SU(3) quenched lattice QCD. We have generated $16^332$ and $20^340$ gauge field configurations, and used them to test this theory. We have concentrated on gauge invariant observables to avoid certain ambiguities with the gauge fixing: the restricted field strength is gauge invariant, but our choice of $\theta$ field is not, making it harder to directly observe \mons in the underlying field. We perform the Abelian decomposition and calculate $\theta$, $n$, the field strength tensor and the portion of the field strength tensor due to the structures in the $\theta$ field. Our goal in this initial study is to show that the Electromagnetic field contains the expected structures, and that the \mons can account for the string tension. Our simulation is limited because although in principle we should recalculate $\theta$ and use a different Abelian decomposition for each Wilson Loop, we do not do so due to limited computer resources; nor have we
considered the deconfinement transition, different representations of the gauge group, flux tubes of the confining string, or the relationship between this picture of confinement and the center vortex or dual Meissner models. These would have to be considered in later works.

Preliminary results were presented in~\cite{Cundy:2012ee}, ~\cite{Cundy:2013pfa} and~\cite{Cundy:2013xsa}.
In section \ref{sec:2} we discuss the Abelian decomposition and its relation to the Wilson Loop and thus quark potential. We discuss the differentiability of the colour field in section \ref{sec:2b}, and then how confinement may arise in this picture in section \ref{sec:2c}. After a brief note on string breaking in section \ref{sec:3}, we present numerical evidence in section \ref{sec:4} and conclude in section \ref{sec:5}. There are four appendices outlining the details of some of the calculations and our numerical methods.

\section{The Abelian Decomposition and Stokes' theorem}\label{sec:2}

The Wilson Loop in an SU($N_C$) gauge theory is defined as
\begin{align}
W_L[C_s,A] = & \frac{1}{N_C} \tr \left(W[C_s,A]\right) & W[C_s,A] = \mathcal{P}[e^{-ig\oint_{C_s} dx_\mu A_\mu(x)}]\label{eq:1}
\end{align}
for a closed curve $C_s$ of length $L$ which starts and finishes at a position $s$, where $\mathcal{P}$ represents path ordering and the gauge field, $A_\mu$, can be written in terms of the Gell-Mann matrices, $\lambda^a$, as $\frac{1}{2}A_\mu^a \lambda^a$. We will use the summation convention that the superscripts $a,b,\ldots$ on a Gell-Mann matrix implies that it should be summed over all values of $a$, $\lambda^a A^a \equiv \sum_{a = 1}^{N_C^2-1} \lambda^aA^a$, while the indices $j,k,\ldots$ are restricted only to the diagonal Gell-Mann matrices, so, in the standard representation, $A^j\lambda^j \equiv \sum_{j = 3,8,\ldots, N_C^2 - 1}\lambda^jA^j$. We shall often leave the gauge field dependence of $W$ and $W_L$ implicit.

The Wilson Loop when $C_s$ is an $R\times T$ rectangle, with spatial extent $R$ and temporal extent $T$, can be used to measure the confining static quark potential, $V(R)$,~\cite{Wilson:1974}
\begin{gather}
V(R)=-\lim_{T\rightarrow\infty} \log(\langle W_L[C_s]\rangle)/T,
\end{gather}
where $\langle \ldots \rangle$ denotes the vacuum expectation value. It is expected, and observed in lattice simulations, that for intermediate distances the confining potential is linear in $R$, so $V(R) \sim \rho R + c$, where $\rho$ is the string tension, and $c$ is a constant. At very small distances, the potential is expected to be Coulomb, while in the presence of fermion loops at very large distances the string is broken and the potential becomes independent of $R$~\cite{Bali:2005fu}. The main focus of this work is on the intermediate regime, so we expect that the expectation value of the Wilson Loop will scale with the spatial and temporal extents of the Wilson Loop as
\begin{gather}
\langle W_L[C_s]\rangle \sim e^{-\rho R T}.\label{eq:al1}
\end{gather}
This is known as the area law scaling of the Wilson Loop. As discussed in \ref{app:B}, this is satisfied, if on each individual configuration, the Wilson Loop scales as
\begin{gather}
W_L[C_s] \sim e^{i F}, \label{eq:al2}
\end{gather}
where $F$ is randomly distributed from configuration to configuration (according to one of a certain set of distributions which includes those discussed in this work) with a mean value proportional to the area contained within the curve $C_s$. The eventual goal (and this work is intended as a step towards that goal) is to demonstrate from first principles that equation (\ref{eq:al2}) is satisfied in pure gauge QCD (and later full QCD), and thus that the quarks are linearly confined.

A difficulty with evaluating equation (\ref{eq:1}) is the path ordering. If the fields $A_\mu(x)$ at different $x$ and $\mu$ commuted with each other, then we could ignore the path ordering, and use Stokes' theorem to convert the line integral to a surface integral. The problem would then reduce to showing that there was some flux flowing through the surface so that the surface integral was proportional to the area (perhaps by counting lines of flux). However, $A_\mu$ are non-Abelian fields: we cannot immediately do this. Our approach is then first to construct an Abelian field $\hat{A}_\mu(x)$ so that $W_L[C_s,\hat{A}] = W_L[C_s,A]$, and the calculated string tension does not depend on which of the two fields we use. We may then remove the path ordering, apply Stokes' theorem to replace the line integral with a surface integral over some field strength, and then show that the surface integral is proportional to the area enclosed within the loop.

First we shall define what is meant by path ordering. We split $C_s$ into infinitesimal segments of length $\delta \sigma$, and define the gauge link as $U_\sigma \in SU(N_C) = \mathcal{P}[e^{-ig\int_{\sigma}^{\sigma + \delta \sigma} A_\sigma d\sigma}] \sim e^{-ig \delta \sigma A_\sigma}$. $0\le\sigma\le L$ represents the position along the curve and we write $A_\sigma \equiv A_{\mu(\sigma)}(x(\sigma))$. We have assumed and will require throughout this work that the gauge field, $A$, is differentiable. This limits us to only a certain subset of gauges, and once we have found a suitable gauge we are restricted to continuous gauge transformations, i.e. only those gauge transformations which can be built up from repeatedly applying infinitesimal gauge transformations, $A_\mu \rightarrow A_\mu + g^{-1}\partial_\mu \alpha + i [\alpha,A_\mu]$, where $\alpha \equiv \alpha^a \lambda^a$ and $\partial_\mu \alpha$ are both infinitesimal. We also neglect the effects of the corners of the Wilson Loop; the discontinuity in $A_\sigma$ at the corner can, for example, be avoided by using a rounded corner. The path ordered integral over gauge fields is defined as the ordered product of the gauge links around the curve in the limit $\delta\sigma \rightarrow 0$.

$W[C_s]$ can be written in this lattice representation as
\begin{gather}
W[C_s] = \lim_{\delta \sigma \rightarrow 0} \prod_{\sigma = 0,\delta \sigma,2\delta\sigma,\ldots}^{L-\delta\sigma} U_\sigma.
\end{gather}
We wish to now replace $U_\sigma$ by an equivalent Abelian field.

Previous work~\cite{DP,PhysRevD.62.074009,Kondo:2008su} has achieved this by inserting the identity operator, as parametrised below, between each pair of gauge links,
\begin{gather}
I_{DP} = \int d\theta\; \theta^\dagger e_{N_C} e^\dagger_{N_C} \theta ,\label{eq:kondoidentity}
\end{gather}
where $\theta \in SU(N)$, $e_{N_C}$ is a unit vector in colour space, and $d\theta$ is a suitable Haar measure for $\theta$. Proof that this is an identity operator is given in~\cite{Kondo:2008su}. This allows one to replace $U_\sigma$ with $\hat{U}_\sigma = e^{-i \delta \sigma e_{N_C}^\dagger( \theta A_\sigma \theta^\dagger + i \theta \partial_\sigma \theta^\dagger)e_{N_C}} + O(\delta\sigma^2)$. The exponent is Abelian, so there is no need for path ordering. In effect, this replaces the path ordering with a path integral over $\theta$. It can be shown that the exponent is related to the CDG Abelian decomposition, which allows for certain monopole solutions which reside in the field $\theta$, and these monopoles may lead to confinement via a dual Meissner effect.

We believe that this approach contains a number of difficulties which have to be resolved before it can be used practically:
\begin{enumerate}
\item It introduces an extra gauge degree of freedom to the system -- contained within the variable $\theta$;
\item It is necessary to evaluate an additional path integral, making calculations harder;
\item In gauge groups larger than SU(2), only a small subset of the possible CDG monopoles are visible to the restricted field (i.e. only the monopoles contained within the $n^{N_C^2-1}$ colour direction. $n$ will be defined in equation (\ref{eq:defeq1}));
\item There is no obvious relation between the physical gauge field and the unphysical $\theta$ field which contains the monopoles and leads to confinement. We expect topological features within the gauge field to contribute to confinement; and this is only possible if the monopoles in the $\theta$ field are related to the monopoles in the gauge field. However, given that we integrate over every possible $\theta$, it is difficult to see how the gauge field itself can affect the confining potential if both the CDG monopole picture and this model are correct;
\item It is possible to introduce a gauge transformation which leaves $\theta$ constant across space time, destroying any monopole structure;
\item The $\theta$ field at each location along the curve is independent (as each needs to be integrated separately), and therefore $e_{N_C}^\dagger( \theta A_\sigma \theta^\dagger + i \theta \partial_\sigma \theta^\dagger)e_{N_C}$ will not be a smooth function of the position for a general configuration of $\theta$. This makes the conversion $\lim_{\delta \sigma \rightarrow 0} \sum {\delta\sigma} \rightarrow \int d \sigma$ problematic -- this is required in the conversion from a product of gauge links to the exponential of a line integral of $A_\sigma$ around the Wilson Loop -- and there are also difficulties defining $\partial_\sigma \theta$. Some regularisation procedure would be required to ensure that $\theta$ is smooth. Normally, in a path integral in Euclidean space time, discontinuities if the fields are suppressed by the differential terms in the action, so we can rightly assume that the fields are differentiable. However, in this case $\theta$ does not contribute to the action, so there is no similar suppression of discontinuities (see ~\cite{PhysRevD.62.025019} for a more detailed criticism along these lines).
\end{enumerate}
In~\cite{Kondo:2008su}, many of these challenges are resolved by fixing the additional degrees of freedom by a procedure analogous to gauge fixing in QCD. This, in effect, selects one particular $\theta$ field and uses that for the subsequent analysis (an alternative way of expressing the approach of~\cite{Kondo:2008su} is to say that their procedure fixes various components of the gauge field $A_\mu$; but this breaks gauge invariance). However, this breaks the identity in equation (\ref{eq:kondoidentity}) and thus the relationship between the Wilson Loop of the restricted gauge field $\hat{U}$ and the Wilson Loop of the original gauge field $U$. In practice, to pursue a program along these lines would require sampling numerous different choices of $\theta$-fixing.

In view of these challenges, we propose that different approach should be taken when trying to express the Wilson Loop in terms of an Abelian field.

We introduce a field $\theta_\sigma \equiv \theta(x(\sigma))$, which, for the moment, we shall take to be an element of U($N_C$), at each location along $C_s$ and insert the identity operator $\theta_\sigma \theta_\sigma^\dagger$ between each of the gauge links. $\theta$ is chosen so that $\theta^\dagger_\sigma U_\sigma \theta_{\sigma + \delta\sigma}$ is diagonal. There are $L/\delta \sigma$ gauge links along the path, and we introduce $L/\delta \sigma$ $\theta$ fields, so there is no obvious reason why the system cannot be solved. In fact, it is easy to construct a solution: it is easy to show (see \ref{app:A.2}) that $\theta_s$ contains the eigenvectors of $W[C_s]$: $W[C_s]\theta_s = \theta_s e^{i\sum_{\lambda^j \text{ diagonal}}\alpha^j \lambda^j}$, for some real $\alpha^j$.

 As the phases of the eigenvectors are arbitrary, this definition only determines $\theta$ up to a $(U(1))^{N_C}$ transformation $\theta \rightarrow \theta \chi$. $\chi$ makes no difference to any physical observable, but for practical purposes it is useful to select the phases and ordering of the eigenvectors by some arbitrary \textit{fixing condition} to give a unique choice of $\theta \in SU(N_C)/(U(1))^{N_C-1}$. We define $SU(N_C)/(U(1))^{N_C-1}$  by considering the following parametrisation of a $U(N)$ matrix:
 \begin{gather}
  \left(\begin{array}{cccc}
                  \cos a_1&i\sin a_1 e^{ic_1}&0&\hdots\\
                  i \sin a_1 e^{-ic_1}&\cos a_1&0&\hdots\\
                  0&0&1&\hdots\\
                  \vdots&\vdots&\vdots&\ddots
                 \end{array}\right)  \left(\begin{array}{cccc}
                  \cos a_2&0&i\sin a_2 e^{ic_2}&\hdots\\
                  0&1&0&\hdots\\
                  i \sin a_2 e^{-ic_2}&0&\cos a_2&\hdots\\
                  \vdots&\vdots&\vdots&\ddots
                 \end{array}\right)\ldots e^{i\left(d_0+ \sum_{\lambda^j \text{diagonal}} d_j \lambda^j\right)}\nonumber
 \end{gather}
 $a_i$, $c_i$, $d_0$ and $d_j$ are real parameters, and there are $N_C(N_C-1)/2$ Givens matrices (i.e. one for each of the possible ways of embedding a $2\times 2$ matrix into a $N_C \times N_C$ matrix) parametrised by one particular $a$ and $c$. An $SU(N_C)/(U(1))^{N_C-1}$ matrix is parametrised in the same way, but without the final $(U(1))^{N_C}$ term (i.e. by setting $d_j$ and $d_0$ to some arbitrary fixed value, most conveniently $d_j = d_0= 0$).

Under a gauge transformation, $U_\sigma \rightarrow \Lambda_\sigma U_{\sigma}\Lambda_{\sigma + \delta\sigma}^\dagger$ for $\Lambda = e^{i\alpha^a \lambda^a}\in SU(N_C)$, $\theta \rightarrow \Lambda \theta \chi$,
where the $(U(1))^{N_C-1}$  factor $\chi$ depends on the fixing condition (in this case $\Lambda\theta \in SU(N)$ so there is no contribution to $\chi$ from $d_0$). This follows from the definition of $\theta$ as containing the eigenvectors of the Wilson Loop. The Wilson Loop transforms under a gauge transformation as $W[C_s,U] \rightarrow \Lambda_s W[C_s,U] \Lambda_s^\dagger$, so the operator which diagonalises it transforms according to $\theta_s \rightarrow \Lambda_s \theta_s$; although we also need to reselect the $U(1)^{N_C-1}$ factor so that the fixing condition remains satisfied. With $\theta^\dagger_\sigma U_\sigma \theta_{\sigma + \delta\sigma} = e^{i\sum_{\lambda^j \text{ diagonal}} \delta \sigma u^j \lambda^j}$ for real $u$,
\begin{gather}
\theta^\dagger_s W[C_s]\theta_s = e^{i \sum_{\lambda^j \text{ diagonal}} \lambda^j \oint_{C_s} d\sigma u^j_\sigma},\label{eq:evth}
\end{gather}
removing the non-Abelian structure and the path ordering without introducing an additional path integral.

Our goal is to apply Stokes' theorem to convert this line integral into a surface integral, and this requires extending the definition of $\theta$ and ${u}^j$ across the surface bounded by $C_s$. In practice, we construct these fields across all of space time. To generalise $\theta$, we construct nested curves, $C^i$, in the same plane as $C_s$ and then stack these rectangles on top of each other in the other dimensions, so that every location in Euclidean space-time is contained within one and only one curve. We then define $\theta$ so it diagonalises the gauge Links (and only the gauge links) which contribute to $W[C^i,U]$.

We cannot naively extend $u^j$ across all of space-time, because its definition requires that all the gauge links $U$ are diagonalised by $\theta$, not just those that contribute to the Wilson loop, and in general this cannot be satisfied. Instead, we replace the gauge links $U$ with a field $\hat{U}$, defined in a consistent way so that it is both diagonalised by $\theta$ across all of space time, and equal to $U$ along the path $C_s$.  The first of these conditions means that
\begin{align}
[\lambda^j,\theta_x^\dagger \hat{U}_{\mu,x}\theta_{x + \hat{\mu}\delta \sigma}] = &0, \\
\intertext{for each diagonal $\lambda^j$, which can be re-written in the form}
\hat{U}_{\mu,x} n^j_{x+\delta\sigma\hat{\mu} }\hat{U}^\dagger_{\mu,x} - n^j_x= &0&n_x^j \equiv &\theta_x \lambda^j \theta^\dagger_x.\label{eq:defeq1}
\end{align}
This condition is satisfied across all of space-time and for all directions $\mu$. Note that $n^j$ is independent of the choice of $\chi$. As we shall see later, the $\theta$-dependence of the restricted field strength $F_{\mu\nu}[\hatA]$ only appears within $n^j$, and objects contributing to the restricted field strength drive confinement. This is the justification of our earlier statement that the choice of $\chi$ does not affect the physical observable, which is the restricted field strength. To give this $\hat{U}$ field a physical meaning we need to relate it to the gauge field $U$, and we do so via a second field $\hat{X}$ defined according to
\begin{gather}
 \hat{X}_\mu = U_\mu(x) \hat{U}^\dagger_\mu.\label{eq:9star}
\end{gather}
 For later convenience (equation (\ref{eq:9})), we restrict $\hat{X}_\mu$ by imposing the condition
\begin{align}
\tr [n^j_x(\hat{X}^\dagger_{\mu,x} - \hat{X}_{\mu,x})] = &0
\label{eq:defeq2}.
\end{align}
 Under a gauge transformation $n$ transforms as $n_x \rightarrow \Lambda_x n_x \Lambda^\dagger_x$ (which follows from the transformation rule for $\theta$) and the requirement that equations (\ref{eq:defeq1}) and (\ref{eq:defeq2}) are satisfied in every gauge leads to the transformation rules $\hat{U}_\mu(x) \rightarrow \Lambda_x \hat{U}_{\mu,x} \Lambda^\dagger_{x+\hat{\mu}\delta\sigma}$ and $\hat{X}_{\mu,x} \rightarrow \Lambda_x \hat{X}_{\mu,x} \Lambda^\dagger_x$.
Equations (\ref{eq:defeq1}) and (\ref{eq:defeq2}) are the lattice versions of the defining equations of the gauge independent CDG decomposition~\cite{Cho:1980,Cho:1981,F-N:98,Shabanov:1999,Duan:1979}, which in the continuum is described by\footnote{The original discoverers of this decomposition prefer to call it a gauge independent decomposition rather than gauge invariant decomposition. In these earlier models, the choice of $\theta$ was left arbitrary. The procedure was to fix to some arbitrary gauge, and then select some $\theta$. Since the decomposition only depends on the choice of $\theta$, and proceeds regardless of which gauge was originally chosen, the decomposition was referred to as gauge independent to distinguish it from other approaches which required fixing to particular gauges, such as the decomposition based on the Maximum Abelian Gauge. Our work differs in philosophy from the original approach because we do not need to perform any gauge fixing to extract our observables, since our key quantities (such as the field strength and Wilson Loop) are gauge covariant and thus the corresponding observables are gauge invariant. However, our particular choice of $\theta$ is gauge dependent, and thus the quantities which we use to parametrise it can only be examined after gauge fixing to some arbitrary gauge. The authors of ~\cite{Cho:1980,Cho:1981,F-N:98,Shabanov:1999,Duan:1979} gauge fixed to an arbitrary gauge and then performed an Abelian decomposition; we decompose and then (if required, which isn't the case for any physical observable) gauge fix.}
\begin{align}
A_\mu = &\hat{A}_\mu + X_\mu\label{eq:11}\\
D_\mu[\hat{A}] n^j = &0\label{eq:contde1} \\
\tr(n^j X) =&0\label{eq:contde2}\\
D_\mu[\hat{A}] \alpha \equiv& \partial_\mu \alpha - i g [\hat{A}_\mu,\alpha]\label{eq:14}\\
\hat{A}_\mu =& \sum_j\left[\frac{1}{2}n^j\tr(n^j A_\mu) + \frac{i}{4g}  [n^j,\partial_\mu n^j]\right],\label{eq:hatA1}
\end{align}
with
\begin{align}
\hat{U} \sim & e^{-i \delta\sigma g \hat{A}}& \hat{X} \sim & e^{i \delta \sigma X}.
\end{align}
Equation (\ref{eq:11}) is the naive continuum limit of equation (\ref{eq:9star}); equations (\ref{eq:contde1}) and (\ref{eq:14}) are the naive continuum limit of equation (\ref{eq:defeq1}), and equation (\ref{eq:contde2}) is the continuum limit of equation (\ref{eq:defeq2}). Proof that equation (\ref{eq:hatA1}) solves the continuum decomposition has previously been given in (for example)~\cite{Cho:1980,Cho:1981,Kondo:2008su}, and is repeated in  \ref{app:A}.

The condition (\ref{eq:defeq2}) may then be interpreted as requiring that $\tr(n^j \hatA) = \tr (n^jA)$, the component of the gauge field $A$ parallel to $n$ is fully contained within $\hatA$. In the continuum, the solution for $\hat{A}$ and $X$ is unique, but on the lattice we found that there were sometimes several distinct solutions to equations (\ref{eq:defeq1}) and (\ref{eq:defeq2}). In this case, we choose the solution which had the largest value of $\tr( \hat{X})$, a condition which is both gauge invariant and satisfied along $C_s$ where $\hat{U} = U$ and thus $\hat{X}=1$.

The continuum defining equations ensure that the field strength $\hat{F}_{\mu\nu}$ associated with $\hat{A}$, defined by
\begin{gather}
[D_\mu[\hat{A}],D_\nu[\hatA]] \alpha = -ig [\hat{F}_{\mu\nu}[\hat{A}],\alpha]\label{eq:f7}
\end{gather}
for any field $\alpha$ in the adjoint representation of the gauge group, satisfies $\hat{F}_{\mu\nu}[\hatA] = \beta_{\mu\nu}^j n^j$ for some real scalars $\beta^j$. The proof of this follows by substituting each of the $n^j$ fields in turn in place of $\alpha$, using equation (\ref{eq:contde1}) to show that $[\hat{F},n^j] = 0$ for all the $n^j$, and noting that the only objects which commute with each of the $n^j$ are the other $n^j$ fields.

We express the restricted field as $\hat{U}_{\mu,x} \equiv \theta_{x}e^{i \lambda^j \delta\sigma\hat{u}^j_{\mu,x}} \theta^\dagger_{x+\hat{\mu}\delta\sigma}$ for real $\hat{u}$, and since $\hat{U} = U$ along the curve $C_s$, we see that $W[C_s,U] = W[C_s,\hat{U}]=\theta_s W[C_s,\theta^\dagger \hat{U}\theta]\theta^\dagger_s = \theta_s e^{i \lambda^j \oint_{C_s} \hat{u}^j_\sigma d\sigma}\theta^\dagger_s$.
 Applying Stokes' theorem to the Abelian field $\theta^\dagger_{x} \hat{U}_{\mu,x} \theta_{x+\hat{\mu}\delta\sigma}$ gives, if $\hat{u}$ is differentiable,
\begin{align}
\theta^\dagger_s W[C_s]\theta_s =& e^{i \lambda^j \int_{x \in \Sigma} d\Sigma_{\mu\nu} \hat{F}^j_{\mu\nu}},\label{eq:7b}\\
\hat{F}^j_{\mu\nu} =& \partial_\mu \hat{u}^j_\nu - \partial_\nu \hat{u}^j_\mu,\label{eq:7}
\end{align}
where $\hat{F}^j$ (like $\hat{u}$) is gauge invariant, $\Sigma$ the (planar) surface bound by the curve $C_s$, and $d\Sigma$ an element of area on that surface. Note that $\hat{u}$ does depend on the fixing condition, although $\hat{F}$ is independent of it.

 Equation (\ref{eq:7}) can be easily proved by considering for smooth $\theta$ and $\hat{u}$
\begin{gather}
\hat{U}_\sigma = e^{i \delta\sigma \hat{u}_\sigma^j n^j  + \delta \sigma \theta \partial_\sigma \theta^\dagger},
\end{gather}
which implies that
\begin{gather}
g\hat{A}_\sigma = - \hat{u}^j_\sigma n^j  + i \theta\partial_\sigma \theta^\dagger,\label{eq:hatA}
\end{gather}
and
\begin{align}
g\hat{F}_{\mu\nu}[\hatA] =& g\left(\partial_\mu \hat{A}_\nu - \partial_\nu\hatA_\mu - i [\hatA_\mu,\hatA_\nu]\right)\nonumber\\
=&-\hat{u}_\nu^j \partial_\mu n^j - n^j \partial_\mu \hat{u}^j_\nu + \hat{u}^j_\mu \partial_\nu n^j + n^j \partial_\nu \hat{u}^j_\mu - \hat{u}^j_\mu[n^j,\theta\partial_\nu \theta^\dagger] + \hat{u}^j_\nu[n^j,\theta\partial_\mu \theta^\dagger] \nonumber\\
&- i \theta\partial_\mu \theta^\dagger \theta \partial_\nu \theta^\dagger + i \theta\partial_\nu \theta^\dagger \theta \partial_\mu \theta^\dagger + i [\theta\partial_\mu \theta^\dagger,\theta \partial_\nu \theta^\dagger]\nonumber\\
= & n^j \partial_\nu \hat{u}^j_\mu - n^j \partial_\mu \hat{u}^j_\nu,\label{eq:f11}
\end{align}
with $\hat{F}_{\mu\nu}^j[\hatA] = \frac{1}{2} \tr(n^j \hat{F}_{\mu\nu}[\hatA])$.
We have used the result
\begin{gather}
\partial_\mu n^j = [n^j,\theta\partial_\mu \theta^\dagger],
\end{gather}
which follows from the definition $n^j = \theta \lambda^j \theta^\dagger$.

We can now consider the Wilson Line around an infinitesimal plaquette $p$, which for a smooth $\hat{u}$ field gives
\begin{gather}
W[p,\hat{U}]=\hat{U}_{x,\mu} \hat{U}_{x+\delta\sigma \hat{\mu},\nu} \hat{U}^\dagger_{x+\delta\sigma\hat{\nu},\mu} \hat{U}^\dagger_{x,\nu} = e^{i \hat{F}_{\mu\nu}},
\end{gather}
which leads to
\begin{gather}
W[p,\theta^\dagger \hat{U}\theta] = e^{i \lambda^j (\partial_\mu\hat{u}^j_\nu -\partial_\nu \hat{u}^j_\mu) },
\end{gather}
and building the integral over the surface bounded by $C_s$ from the product of integrals over these small plaquettes, using that the exponent is Abelian, gives equations (\ref{eq:7b}) and (\ref{eq:7}).

Proof that the definitions of $\hat{F}$ given in equations (\ref{eq:f7}) and (\ref{eq:f11}) are equivalent is given in  \ref{app:A}, along with a proof that the definition of $\hat{A}_{\sigma}$ in equation (\ref{eq:hatA}) is equivalent to that given in equation (\ref{eq:hatA1}).

Finally,  as alluded to earlier, we note that $\hat{F}_{\mu\nu}[\hatA] $ can be written in the forms
\begin{align}
\hat{F}_{\mu\nu}[\hatA] =& \frac{1}{2}n^j (\partial_\mu \tr n^j A_\nu - \partial_\nu \tr n^j A_\mu) - \frac{i}{2g} n^j\tr(n^j[\theta\partial_\mu \theta^\dagger,\theta\partial_\nu \theta^\dagger])\\
=&\frac{1}{2}n^j (\partial_\mu \tr n^j A_\nu - \partial_\nu \tr n^j A_\mu) + \frac{i}{8g} n^j\tr(n^j[\partial_\mu n^k,\partial_\nu n^k]),\label{eq:26}
\end{align}
as is proved in  \ref{app:A}, equations (\ref{eq:F1}) and (\ref{eq:F2}). These functions solely depend on $n$, and $\theta$ only indirectly through $n$, and since $n$ is independent of $\chi$ and $\hat{F}$ is the physical observable we want to study, we conclude that the choice of $\chi$ will not affect any of the physical observables we need.

Equation (\ref{eq:7b}) is only valid if $\hat{u}$ is differentiable. Equation (\ref{eq:7b}) is also similar to what we see in QED, which is, of course, not confining. In analogy to QED, we may expect the contribution of those portions of space time where $\hat{u}$ is continuous to have little contribution to the string tension.  However, we must also add to this equation the effects of discontinuities in $\hat{u}$. We do so by only extending the area integral over those areas where $\hat{u}$ is continuous, and add additional line integrals around the areas where it is discontinuous. The linear string tension will, at least in part, arise from these discontinuities. Since $\hat{u}^j$ is built from the gauge field, and we are working on a gauge where this is assumed to be differentiable, and $\theta$, we must therefore consider whether $\theta$ is differentiable.

\section{Differentiation of eigenvectors}\label{sec:2b}

\subsection{SU(2)}

We wish to investigate how $\theta$ varies over space, which requires evaluating the derivative of the eigenvectors of the Wilson Loop operator. Suppose that an operator $M$ has two eigenvectors (which is the case for SU(2)), $\psi_1$ and $\psi_2$, with eigenvalues $\lambda_1$ and $\lambda_2$ and after a small change in $M \rightarrow M + \delta M$, the eigenvectors change to $\psi'_1$ and $\psi'_2$, with eigenvalues $\lambda_1 + \delta \lambda_1$ and $\lambda_2 + \delta \lambda_2$. For normal $M$ (so that the eigenvectors are orthogonal, left and right eigenvectors the same, and $\delta M$ is (anti-)Hermitian or an anti-Hermitian matrix multiplied by $M$ if $M$ is unitary),  we may parametrise the new eigenvectors as
\begin{align}
\psi_1' = & \psi_1 \cos(\beta) + \psi_2 e^{i\gamma}\sin\beta\nonumber\\
 \psi_2' = & \psi_2 \cos(\beta) - \psi_1 e^{i\gamma}\sin\beta,
\end{align}
and the eigenvalue equations read
\begin{align}
M \psi_1 = & \lambda_1 \psi_1\nonumber\\
M \psi_2 = & \lambda_2 \psi_2\nonumber\\
(M+\delta M)(\psi_1 \cos(\beta) + \psi_2 e^{i\gamma}\sin\beta) = & (\lambda_1 + \delta \lambda_1)(\psi_1 \cos(\beta) + \psi_2 e^{i\gamma}\sin\beta)\nonumber\\
(M+\delta M)(\psi_2 \cos(\beta) - \psi_1 e^{i\gamma}\sin\beta) = & (\lambda_2 + \delta \lambda_2)(\psi_2 \cos(\beta) - \psi_1 e^{i\gamma}\sin\beta),
\end{align}
which, after applying $\psi_1^\dagger$ and $\psi_2^\dagger$ to these equations gives
\begin{align}
\delta \lambda_1 = &\psi_1^\dagger\delta M \psi_1 + \psi_1^\dagger \delta M \psi_2 e^{i\gamma} \tan \beta\nonumber\\
\delta \lambda_2 = &\psi_2^\dagger\delta M \psi_2 - \psi_2^\dagger \delta M \psi_1 e^{i\gamma} \tan \beta\nonumber\\
\tan 2\beta = & 2 \frac{\sqrt{\psi_2^\dagger \delta M \psi_1 \psi_1^\dagger \delta M \psi_2}}{\lambda_1 - \lambda_2 + \psi_1^\dagger \delta M \psi_1 - \psi_2^\dagger \delta M \psi_2}\nonumber\\
e^{i\gamma} = & \sqrt{\frac{\psi_2^\dagger \delta M \psi_1}{\psi_1^\dagger \delta M \psi_2}}.
\end{align}
At $|\lambda_1 - \lambda_2| \gg |\psi_1^\dagger \delta M \psi_1 - \psi_2^\dagger \delta M \psi_2|$, this agrees with the result from first order perturbation theory~\cite{Cundy:2007df}
\begin{gather}
\delta \psi_i = \psi_i'-\psi_i = (1-\psi_i\psi_i^\dagger) \frac{1}{M-\lambda} \delta M \psi_i.
\end{gather}
In this case, the change in the eigenvector is proportional to $\delta M$. However, if the eigenvalues are near-degenerate (i.e. the difference between them is of O($\|\delta M\|$)), then $\beta$ may be large even as $\delta M \rightarrow 0$, and the eigenvectors would evolve discontinuously.

As we will be interested in the spatial derivatives of the eigenvectors of the closed Wilson Loop, I will set $M$ to be the Wilson Loop operator $W[C_s^r]$, and consider the derivatives in two directions, one parallel to the Wilson Line, and one perpendicular to it in the plane of the Wilson Loop. Any small shift in location in an arbitrary direction can be constructed in terms of a shift parallel to the Wilson Loop followed by a shift perpendicular to the Wilson Loop. $r$ is an index denoting the size of the Wilson Loop; an increase in $r$ means that we move to the next of the nested Wilson Loops used to define $\theta$, so both the spatial and temporal extents of the curve $C_s^r$ would increase by a small amount $2\delta r$. For an $R\times T$ planar Wilson Loop with $T>R$ we have $r = R/2$. We investigate the change of the Wilson loop both when we have a small shift along the curve, $s \rightarrow s + \delta \sigma$, and a small shift perpendicular to the curve, $r \rightarrow r+\delta r$.

Firstly, there will clearly be discontinuities in $\theta$ when the gauge field is discontinuous, but as we require a continuous gauge field for the decomposition to be valid we neglect this case. We wish to consider if there can be discontinuities in the eigenvectors of the Wilson line even for a continuous gauge field.

In the direction parallel to the line, we have
\begin{gather}
W[C^{r}_{s+\delta \sigma}] = U^\dagger_s  W[C_s^r] U_s,
\end{gather}
and with $U_s =e^{-ig\delta\sigma A_s^r}$, we have for a smooth gauge field
\begin{gather}
\delta W[C_s^r] = -ig \delta\sigma (A_s^r W[C_s^r] - W[C_s^r] A_s^r).
\end{gather}
This gives (if the eigenvalues of the Wilson Loop are not degenerate)
\begin{align}
\tan2\beta^r_s =& \frac{2 \delta\sigma g \sqrt{-\psi_2^\dagger (A_s^r W[C_s^r] - W[C_s^r] A_s^r) \psi_1 \psi_1^\dagger(A_s^r W[C_s^r] - W[C_s^r] A_s^r)\psi_2}}{\lambda_1 - \lambda_2 - ig\delta\sigma\psi_1^\dagger (A_s^r W[C_s^r] - W[C_s^r] A_s^r) \psi_1 + \psi_2^\dagger (A_s^r W[C_s^r] - W[C_s^r] A_s^r) \psi_2 i g \delta\sigma}\nonumber\\
=&\frac{2 \delta\sigma g \sqrt{(\lambda_1 - \lambda_2)^2\psi_2^\dagger (A_s^r) \psi_1 \psi_1^\dagger A_s^r \psi_2}}{\lambda_1 - \lambda_2}\nonumber\\
= &  \text{sign}(\lambda_1 - \lambda_2)\delta\sigma g \sqrt{\tr [(T^1 + i T^2)\theta^\dagger A_s^r \theta]\tr [(T^1 - i T^2)\theta^\dagger A_s^r \theta]}.\nonumber\\
e^{2i\gamma^r_s}= & \frac{\tr [(T^1 + i T^2)\theta^\dagger A_s^r \theta]}{\tr [(T^1 - i T^2)\theta^\dagger A_s^r \theta]}
\end{align}
$\beta$ is proportional to $\delta\sigma$, so, baring changes in the phase, the eigenvectors evolve smoothly around the Wilson Line.

Considering changes perpendicular to the Wilson Line, we have
\begin{gather}
W[C^{r+\delta r}_{s}] = i g \delta\sigma  \delta_r\sum_{m}U^r_sU^r_{s+\delta\sigma} \ldots  \partial_r( U^r_{s+\delta\sigma}  ) (U^r_{s+m\delta\sigma})^\dagger \ldots (U_{s+\delta\sigma}^r)^\dagger (U_s^r)^\dagger W[C_s^r]+ W[C^r_s],
\end{gather}
which gives, since $\lambda_1\lambda_2 = 1$ for $W \in SU(2)$,
\begin{align}
\tan2\beta^{r} = & 2\frac{ \delta\sigma \delta r g \sqrt{\sum_{s',s''}e^{-2i \int_{s'}^{s''} ds''' u^3_{s'''}}\tr[(T^1 + i T^2)\theta^\dagger_{s'}\partial_r A_{s'} \theta_{s'} ]\tr[(T^1 - i T^2)\theta^\dagger_{s''}\partial_r A_{s''} \theta_{s''} ]}}{\lambda_1 - \lambda_2 + ig\delta\sigma\delta r\sum_{s'} \tr( T_3 {\theta_{s'}^r}^\dagger  \partial_r A_{s'}^r W_L^{s',r}\theta_{s'}^r )}\nonumber\\
e^{2i\gamma^{r,s}}= & \frac{\sum_{s'} \delta\sigma \delta r e^{-2i \int_{s}^{s'} ds''' u^3_{s'''}}\tr[(T^1 + i T^2)\theta^\dagger_{s'}\partial_r A_{s'} \theta_{s'} ]}{\sum_{s'} \delta\sigma \delta r e^{2i \int_{s}^{s'} ds''' u^3_{s'''}}\tr[(T^1 - i T^2)\theta^\dagger_{s'}\partial_r A_{s'} \theta_{s'} ]}\label{eq:thetardiff}
\end{align}
We see that, expect when the eigenvalues are near degenerate, $\beta$ will be proportional to $\delta r$ -- though not $\delta \sigma$ as $\sum_{s,s''}\delta\sigma^2\ldots \sim O(1)$. This means that the $\theta$ field is discontinuous when the Wilson Loop has degenerate eigenvalues -- and we shall investigate the consequences of this later. For now, we just consider the other situations with infinitesimal $\beta$, and ask if it is possible that there can be a large change in $\theta$ without a large change in $\beta$.  We obtain
\begin{align}
\theta^{r+\delta r}_s \rightarrow &\theta^r_s B\nonumber\\
B=& \left(\begin{array}{c c}\cos \beta^{r} & -\sin\beta^{r} e^{-i\gamma^{r,s}}\\
\sin\beta^{r} e^{i\gamma^{r,s}} &\cos\beta^{r} \end{array}\right)e^{i\delta \lambda^3},
\end{align}
where the phase $\delta$ is chosen so that whatever fixing condition we have chosen is satisfied. We denote $B^r_s$ as the appropriate transformation when performing a small translation in the direction of increasing $s$ and $B^{s,r}$ when translating in the direction of increasing $r$. In section \ref{sec:3}, we shall use these relations in our discussion of string breaking.

There is one further discontinuity as we evolve $\theta$ over $s$ and $r$. We may parametrise $\theta$ as
\begin{gather}
\theta = \left(\begin{array}{c c} \cos a & i \sin a e^{ic}\\ i \sin a e^{-ic}&\cos a\end{array}\right) \left(\begin{array}{cc} e^{id}&0\\0&e^{-id}\end{array}\right),\label{eq:thpar}
\end{gather}
with $0\le a \le \pi/2$, $c\in\mathbb{R}$ and $d \in \mathbb{R}$. Note that the parameters $a$, $c$ and $d$ are not gauge invariant, so this parametrisation is only defined after we have fixed to a suitable gauge.
The phase $d$ makes no physical contribution, and should be fixed by some suitable condition, but for completeness we still give it explicitly. We perform a small translation in either the $r$ or $s$ direction, and obtain $\theta'$ which is parametrised by $a'$, $c'$ and $d'$:
\begin{gather}
\theta' = \left(\begin{array}{c c}
\cos a \cos \beta e^{id} + i \sin a \sin\beta e^{i (c + \gamma - d)}& -\cos a \sin \beta e^{i(d-\gamma)} + i \sin a \cos \beta e^{i (c-d)}\\
\cos a \sin \beta e^{-i(d-\gamma)} + i \sin a \cos \beta e^{-i (c-d)}&\cos a \cos \beta e^{-id} - i \sin a \sin\beta e^{-i( c + \gamma - d)}\end{array}\right)e^{i\delta\lambda^3},
\end{gather}
which allows us to read off the new parameters $a'$, $c'$ and $d'$,
\begin{align}
\cos a' =& \sqrt{\cos^2 a \cos^2 \beta + \sin^2 a \sin^2 \beta - 2 \cos a \cos \beta \sin a \sin \beta \sin (c + \gamma - 2d)}\nonumber\\
\tan (d'-d-\delta) = & \frac{\sin a \sin \beta \cos(c+\gamma - 2d)}{\cos a \cos \beta - \sin a \sin \beta \sin(c+\gamma - 2d)}\nonumber\\
\tan (c'+d-d'+\delta) = & \frac{\cos a \sin \beta \cos(2d-\gamma) + \sin a \cos \beta \sin(c)}{-\cos a \sin \beta \sin(2d-\gamma) + \sin a \cos \beta \cos(c)}.\label{eq:cpc}
\end{align}
It is clear that when $\cos a$ and $\sin a \gg\beta $, we have $a' = a + O(\beta)$, $c' = c + O(\beta)$ and $d' = d + O(\beta)$, which means that $\theta$ is differentiable with respect to the spatial coordinates at least once.
However, if $a \sim \pi/2$ we have $d' = c + \gamma - d +\delta + (\nu+\frac{1}{2}) \pi$ and $c' = 2c - 2d +\gamma + (\nu + \frac{1}{2}) \pi$, where $\nu$ is an integer, and, at $ a \sim 0 $, $c' \sim 2d - \gamma + (\nu + \frac{1}{2}) \pi$. At both $a \sim 0$ (where $\theta$ is diagonal) and $a \sim \pi/2$ (where $\theta$ is off-diagonal), the parameters which describe $\theta$ are discontinuous even when $\beta$ is small.
$a$ itself, however, remains continuous and the eigenvalues of $\theta$ are dependent only on $a$ once we fix the phase $d$. $\delta$ can be chosen according to our rule for fixing the U(1) phase $d$.
 For example, if we enforce $d = d'= 0$, then $\delta = -\arctan\left( \frac{\sin a \sin \beta \cos(c+\gamma )}{\cos a \cos \beta - \sin a \sin \beta \sin(c+\gamma )} \right)$.

This discontinuity of $\theta$ may be in principle be at a single point, or there may lines in the surface spanned by $r$ and $s$ where these discontinuities occur (if they occur), and these lines will either be open, form closed loops or will be infinite in extent. If they are lines, then we require that $a'=a=\pi/2$ (for example), which means that $\theta$ will be unchanged as we make a small displacement in space in some arbitrary direction. If $\theta$ is unchanged, then the Wilson Loop must also be unchanged as we move across space time. However, baring a remarkable gauge field configuration, it is unlikely that this will occur for any more than an infinitesimal distance. We may as well then neglect the possibility that these discontinuities are in lines, and focus on the case that they are at isolated points. 

\subsection{SU(3)}
In SU(3), we introduce the following operators which form a complete basis of the Hermitian traceless generators of SU(3)
\begin{align}
\phi_1 = & \left(\begin{array}{c c c}
0&e^{ic_1}&0\\
e^{-ic_1}&0&0\\
0&0&0\end{array}\right)&\bar{\phi}_1 = & \left(\begin{array}{c c c}
0&ie^{ic_1}&0\\
-ie^{-ic_1}&0&0\\
0&0&0\end{array}\right)\nonumber\\
\phi_2 = & \left(\begin{array}{c c c}
0&0&e^{ic_2}\\
0&0&0\\
e^{-ic_2}&0&0\end{array}\right)&
\bar{\phi}_2 =& \left(\begin{array}{c c c}
0&0&ie^{ic_2}\\
0&0&0\\
-ie^{-ic_2}&0&0\end{array}\right)\nonumber\\
\phi_3 = & \left(\begin{array}{c c c}
0&0&0\\
0&0&e^{ic_3}\\
0&e^{-ic_3}&0\end{array}\right)&
\bar{\phi}_3 = & \left(\begin{array}{c c c}
0&0&0\\
0&0&ie^{ic_3}\\
0&-ie^{-ic_3}&0\end{array}\right)\nonumber\\
\lambda^3=&\left(\begin{array}{c c c}
1&0&0\\
0&-1&0\\
0&0&0\end{array}\right)&\lambda^8=&\frac{1}{\sqrt{3}}\left(\begin{array}{c c c}
1&0&0\\
0&1&0\\
0&0&-2\end{array}\right).\label{eq:39z}
\end{align}
Any SU(3) matrix, $\theta$, can then be parametrised in terms of the eight quantities $a_1$, $a_2$, $a_3$, $c_1$, $c_2$, $c_3$, $d_3$ and $d_8$ as
\begin{align}
\theta =& e^{ia_3\phi_3} e^{ia_2\phi_2} e^{ia_1 \phi_1} e^{id_3 \lambda^3} e^{id_8\lambda^8}.\label{eq:l1}
\end{align}
As in SU(2), these parameters are not gauge invariant.

Multiplying out the various components of $\theta$ gives
\begin{gather}
\theta = \left(\begin{array}{ccc}\theta_{11}&\theta_{12}&\theta_{13}\\
\theta_{21}&\theta_{22} &\theta_{23}\\
\theta_{31}&\theta_{32} &\theta_{33} \end{array} \right)
\left(\begin{array}{ccc}e^{id_3 + i\frac{d_8}{\sqrt{3} }}&0&0\\
0& e^{ i\frac{d_8}{\sqrt{3} }-id_3}&0\\
0&0&e^{-i\frac{2d_8}{\sqrt{3}}}\end{array}\right),
\end{gather}
with
\begin{align}
\theta_{11} = &\cos a_2\cos a_1\nonumber\\
\theta_{12} = & i\cos a_2 \sin a_1 e^{ic_1}\nonumber\\
\theta_{13} = & i \sin a_2 e^{ic_2}\nonumber\\
\theta_{21}=&-\cos a_1 \sin a_3 \sin a_2 e^{i(c_3 - c_2)} + i \cos a_3 \sin a_1 e^{-ic_1}\nonumber\\
\theta_{22}=&\cos a_3 \cos a_1 - i \sin a_1 \sin a_2 \sin a_3 e^{i (c_1 + c_3 - c_2)}\nonumber\\
\theta_{23}=&i\sin a_3 \cos a_2 e^{ic_3}\nonumber\\
\theta_{31}=&i \cos a_3 \cos a_1 \sin a_2 e^{-ic_2} - \sin a_3 \sin a_1 e^{-i(c_1 + c_3)}\nonumber\\
\theta_{32} = &i \sin a_3 \cos a_1 e^{-ic_3} - \cos a_3 \sin a_1 \sin a_2 e^{i(c_1 - c_2)}\nonumber\\
\theta_{33}=&\cos a_3 \cos a_2 \label{eq:thpar2}
\end{align}

Note that $d_3$ and $d_8$ can easily be constructed from the phase of the $11$ and $33$ components of $\theta$, $c_2$ and $a_2$ from the $13$ component once the phase has been removed, and $a_3$ and $c_3$ and $a_1$ and $c_1$ from the $23$ and $12$ components of $\theta$. We can again construct the differential of $\theta$ with respect to a coordinate by multiplying $\theta$ by a matrix $B$ constructed in the same parametrisation as $\theta$:
\begin{gather}
B = \left(\begin{array}{ccc}\beta_{11}&\beta_{12}&\beta_{13}\\\beta_{21}&\beta_{22}&\beta_{23}\\\beta_{31}&\beta_{32}&\beta_{33} \end{array} \right)
\left(\begin{array}{ccc}e^{i\delta_3 + i\frac{\delta_8}{\sqrt{3} }}&0&0\\
0& e^{ i\frac{\delta_8}{\sqrt{3} }-i\delta_3}&0\\
0&0&e^{-i\frac{2\delta_8}{\sqrt{3}}}\end{array}\right),
\end{gather}
with $\beta_{ij}$ taken from equation (\ref{eq:thpar2}) but with the substitution $\theta_{ij} \rightarrow \beta_{ij}$, $a_i \rightarrow \beta_i$, $c_i \rightarrow \gamma_i$.

Again, except in the neighbourhood of degenerate eigenvalues of the Wilson Loop, $\beta$ will be infinitesimal and we can neglect terms of O($\beta^2$) and higher. With the shifted $\theta' = \theta B$, we have
\begin{align}
\theta'_{11} =&\left( \cos a_1 \cos a_2 e^{i( \frac{d_8}{\sqrt{3}}+d_3)} - \beta_1 \cos a_2 \sin a_1 e^{i (c_1 - \gamma_1 +\frac{d_8}{\sqrt{3}}-d_3) } - \beta_2 \sin a_2e^{i (c_2-\gamma_2 - \frac{2d_8}{\sqrt{3}})}\right)e^{i\delta_3 +i \frac{\delta_8}{\sqrt{3}}}\nonumber\\
\theta'_{12} =& \left(i\cos a_1 \cos a_2 \beta_1 e^{i(\gamma_1+ \frac{d_8}{\sqrt{3}}+d_3)} +i  \cos a_2 \sin a_1 e^{i (c_1  +\frac{d_8}{\sqrt{3}}-d_3) } - \beta_3 \sin a_2e^{i (c_2-\gamma_3 - \frac{2d_8}{\sqrt{3}})}\right)e^{i \frac{\delta_8}{\sqrt{3}}-i\delta_3}\nonumber\\
\theta'_{13} =& \left(i\cos a_1 \cos a_2 \beta_2 e^{i(\gamma_2+ \frac{d_8}{\sqrt{3}}+d_3)} - \beta_3  \cos a_2 \sin a_1 e^{i (c_1 +\gamma_3 +\frac{d_8}{\sqrt{3}}-d_3) } +i  \sin a_2e^{i (c_2 - \frac{2d_8}{\sqrt{3}})}\right)e^{-2i\frac{\delta_8}{\sqrt{3}}}\nonumber\\
\theta'_{23} =& \bigg(i \beta_2 e^{i(\gamma_2 + \frac{d_8}{\sqrt{3}}+d_3)}(i \cos a_3 \sin a_1 e^{-ic_1} - \cos a_1 \sin a_3\sin a_2 e^{i(c_3 - c_2)}) +\nonumber\\
&\spa i\beta_3 e^{i(\gamma_3+\frac{d_8}{\sqrt{3}}-d_3)}(\cos a_3 \cos a_1 - i \sin a_1 \sin a_2 \sin a_3 e^{i(c_1 + c_3 - c_2)}) + \nonumber\\
&\spa i \sin a_3 \cos a_2 e^{i(c_3\frac{2d_8}{\sqrt{3}})}\bigg)e^{-2i\frac{\delta_8}{\sqrt{3}}}\nonumber\\
\theta'_{33} =& \bigg(i \beta_2 e^{i(\gamma_2 + \frac{d_8}{\sqrt{3}}+d_3)}(i \cos a_3 \cos a_1 \sin a_2 e^{-ic_2} -  \sin a_3\sin a_1 e^{-i(c_3 + c_1)}) +\nonumber\\
& \spa i \beta_3e^{i(\gamma_3 + \frac{d_8}{\sqrt{3}}-d_3)}(i\sin a_3 \cos a_1e^{-ic_3} - \sin a_1 \sin a_2 \cos a_3 e^{i(c_1 - c_2)}) + \nonumber\\
& \spa \cos a_3 \cos a_2 e^{-i\frac{2d_8}{\sqrt{3}}}\bigg)e^{-2i\frac{\delta_8}{\sqrt{3}}},
\end{align}
where, for the $d=0$ fixing condition, $\delta$ is tuned to ensure that $\theta'_{11}$ and $\theta'_{33}$ are real.
Once again, we see that $a_i$, $c_i$ and $d_i$ will be differentiable at least once where $\beta$ is continuous except where either $a_1 = \pi/2$, $a_2 = \pi/2$, $a_3 = \pi/2$, $a_1 = 0$, $a_2 = 0$, or $a_3 = 0$, where the parameters $c_i$ are discontinuous (possibly because the $\delta$ required to fix $d$ is not infinitesimal).

\section{Topological solutions}\label{sec:2c}
\subsection{SU(2)}
Now suppose that $\hat{u}^j$ contains a non-analyticity. We integrate the field around a loop $\tilde{C}$ parametrised by $\tilde{\sigma}$ surrounding the discontinuity, bounding the surface integral by an additional line integral
$
\oint_{\tilde{C}} d\tilde{\sigma} \hat{u}^j_{\tilde{\sigma}},
$
far enough away from the discontinuity that $\hat{u}^j_{\tilde{\sigma}}$ is analytic along $\tilde{C}$.
We define $\{\tilde{C}_n\}$ as the set of curves surrounding all these discontinuities, and $\tilde{\Sigma}$ the area bound within these curves.
We can write
\begin{gather}
e^{i \lambda^j\delta\tilde{\sigma} \hat{u}^j_{\mu,x}} = \theta^\dagger_x \hat{X}^\dagger_{\mu,x} \theta_x \theta^\dagger_x U_{\mu,x} \theta_{x + \delta\tilde{\sigma}},\label{eq:str1}
\end{gather}
and since $\hat{u}$ is continuous on $\tilde{C}$, after fixing the gauge we can expand $U = 1 -i\frac{1}{2}g\delta\tilde{\sigma} A^a \lambda^a$ and $\theta^\dagger_x \theta_{x+\delta\tilde{\sigma}} = 1 + \delta\tilde{\sigma} \theta^\dagger \partial_{\tilde\sigma} \theta$. We define $X_0 \equiv \frac{1}{2}\theta^\dagger(X + X^\dagger)\theta$.  For example, in SU(2) we may parametrise $\hat{X}$ as
\begin{gather}
\hat{X} = \left(\begin{array}{c c} \cos x& i \sin x e^{iy}\\ i \sin x e^{-iy} & \cos x \end{array}\right)\left(\begin{array}{c c}e^{iw}&0\\0&e^{-iw}\end{array}\right),
\end{gather}
then $\hat{X}+\hat{X}^\dagger = 2 \cos x \cos w I = 2X_0$, where $I$ is the identity operator (this expression will obviously not be exact in higher gauge groups). If everything is analytic along $\tilde{C}_s$, and we consider infinitesimal displacements, we expect $\hat{X}$ and $\hat{U}$ to be close to the identity operator, and $x$ and $w$ will both be O($\delta\sigma$). $X_0$ is thus $I + O(\delta \sigma^2)$.

This gives
\begin{multline}
i\delta\tilde{\sigma} \hat{u}^j_{\mu,x} = \frac{1}{\tr (\lambda^j)^2}\text{Im}\left(\;\tr \left[\lambda^j\theta^\dagger_x \hat{X}^\dagger_{\mu,x} \theta_x \theta_x^\dagger U_{\mu,x}\theta_{x+\delta\tilde{\sigma} \hat{\mu}}\right]\right)=\\
\frac{1}{2\tr (\lambda^j)^2}\tr[\lambda^j \theta^\dagger_x (\hat{X}^\dagger - \hat{X}) \theta_x - \frac{1}{2}i \lambda^j \delta \tilde{\sigma}\theta^\dagger_x(\hat{X}^\dagger + \hat{X})  gA^a_{\mu,x} \lambda^a\theta_x + \lambda^j \theta^\dagger_x(\hat{X}^\dagger + \hat{X}) \theta_x \delta \tilde{\sigma}\theta_x^\dagger\partial_{\tilde{\sigma}} \theta]\label{eq:9}.
\end{multline}

Using (\ref{eq:defeq2}) the first term in equation (\ref{eq:9}) gives zero, while the second term will not contribute to an integral around $\tilde{C}$ if $U$ and $X_0$ are continuous and the area of the loop is small enough. We therefore concentrate on the contribution from the final term. Equation (\ref{eq:7}) is then replaced by
\begin{gather}
\theta^\dagger_s W[C_s]\theta_s = e^{i  \lambda^j\left[ \int_{(x \in \Sigma) \cap (x\not{\in} \tilde{\Sigma})} d\Sigma_{\mu\nu} \hat{F}^j_{\mu\nu} + \sum_n\oint_{\tilde{C}_n} d\tilde{\sigma}\frac{1}{\tr (\lambda^j)^2} \tr [\lambda^j X_0  \theta^\dagger \partial_{\tilde{\sigma}}\theta]\right]},
\end{gather}
where $d\Sigma$ is an element of area.

In SU(2), after fixing to a suitable gauge (where the gauge field $U$ is smooth), we parametrise $\theta$ as in equation (\ref{eq:thpar}),
\begin{align}
\theta =&(\cos a + i \sin a \phi)e^{id_3\lambda^3}& \phi = & \left(\begin{array}{cc} 0&  e^{i c}\\ e^{-ic} &0 \end{array}\right)\nonumber\\
\bar{\phi} = &  \left(\begin{array}{cc} 0&  ie^{i c}\\ -ie^{-ic} &0 \end{array}\right)& \lambda^3 = & \left(\begin{array}{cc} 1& 0\\ 0 &- 1\end{array}\right),\label{eq:thdec}
\end{align}
with $c \in \mathbb{R}$, $0 \le a \le \pi/2$ and $d_3$ determined by the fixing condition. As this contains the eigenvectors of $W[C_s]$, it is differentiable except where $W[C_s]$ has degenerate eigenvalues and those locations where $a = 0$ or $a \approx \pi/2$, where $c$ is ill-defined.
The parameter $c$ may wind itself around those points where $a=\pi/2$ or $a=0$.
We may parametrise the plane of the Wilson Loop in polar coordinates $(r,\phi)$, with the origin at the location where $a = \pi/2$ or $a = 0$.
The continuity of $\theta$ away from $r\neq 0$ implies that at infinitesimal but non-zero $r$,  $c(r,\phi =0) = c(r,\phi = 2\pi) + 2\pi \nu_n$ for integer winding number $\nu_n$, and if $c$ is ill defined at $r = 0$ then we may find that $\nu_n \neq 0$.

It is straight-forward to calculate
\begin{gather}
\theta^\dagger\partial_\sigma \theta = e^{-id_3\lambda^3}\left[ i \partial_\sigma a \phi + i \lambda^3 \partial_\sigma d_3 + i \sin a \cos a \bar{\phi}\partial_\sigma c - i \sin^2 a \partial_\sigma c \lambda^3 \right]e^{id_3 \lambda^3}.\label{eq:tdtsu2}
\end{gather}
We integrate around a curve at fixed $a$, and we may choose a fixing condition which keeps $d$ constant. This leads to
\begin{align}
 \theta^\dagger_s W[C_s]\theta_s=& e^{i \lambda^3\left[\int_{(x \in \Sigma) \cap (x\not{\in} \tilde{\Sigma})} d\Sigma_{\mu\nu} \hat{F}^j_{\mu\nu}- \sum_n\oint_{\tilde{C}_n} d\tilde{\sigma} \sin^2 a\partial_\sigma c \frac{1}{2}\tr [X_{0n}]\right]} \nonumber\\
=& e^{i \lambda^3 \left[\int_{(x \in \Sigma) \cap (x\not{\in} \tilde{\Sigma})} d\Sigma_{\mu\nu} \hat{F}^j_{\mu\nu}-\sum_n 2\pi \sin^2 a\nu_n \lambda^3\right]},
\end{align}
where the last equality is only valid if $X_{0}$ is close to the identity operator, as we expect. In SU(2) (but not higher gauge groups, where the expression depends on more parameters) we discover that were we to use an infinitesimal loop around the discontinuity, then the area integral gives $ \theta^\dagger_s W[C_s]\theta_s \sim e^{2\pi i\nu\lambda^3}$ for integer $\nu$, which cannot contribute. However, if we integrate at a larger radius, where $a$ has fallen to some average background value $a_0$, we do see a contribution to the Wilson Loop from the discontinuity,  proportional to $2\pi \nu_n \sin^2 a_0$, and we may still find a large contribution to the string tension. The structure in the field is extended over a small region of space rather than being a $\delta$-function spike. This conclusion may also be reached by considering gauge invariance: while $\hat{F}$ is gauge invariant, $\theta$ is not and therefore the precise location of the discontinuity in $\theta$ will depend on the gauge (although it should not depend too strongly on the gauge). The peak in $\hat{F}$ cannot therefore be precisely at the discontinuity in $\theta$ in every gauge, so it cannot be a $\delta$-function peak around that centre. We shall discuss gauge invariance in more detail in section \ref{sec:4.4}.

We will refer these objects in the field strength as CDG \monsp. They are topological objects in the homotopy group $\pi_1$ (distinct from monopoles which have the homotopy group $\pi_2$, and instantons which are of the group $\pi_3$). They can be identified by a large value of $\hat{F}$, $a = \pi/2$ or $a = 0$, and a non-zero winding of the parameter $c$ around the \monp.

It is also important to note that the \mons at $a = 0$ and $a = \pi/2$ with the same winding contribute with opposite signs. If we compare the results for the integrals of the restricted potential around fixed curves at $a = \delta$ and $a = \pi/2 - \delta$, for some small $\delta$, we find that the term proportional to the winding number gives
\begin{align}
 \oint_{a = \delta} g\hat{A}_\mu dx_\mu = & \lambda_3 \sin^2 \delta 2\pi \nu + \ldots = \lambda_3 \delta^2 2\pi \nu + \ldots\nonumber\\
 \oint_{a = \pi/2 - \delta} g\hat{A}_\mu dx_\mu = & \lambda_3 \sin^2 (\pi/2 -\delta) 2\pi \nu + \ldots = \lambda_3 (1-\delta^2) 2\pi \nu + \ldots \nonumber\\
\end{align}
When this is multiplied by $i$ and exponentiated,  $e^{i2\pi \nu} = 1$ can't contribute to the Wilson Loop, so we can neglect this part of the expression and write,
\begin{gather}
 \oint_{a = \pi/2 - \delta} g\hat{A}_\mu dx_\mu =   -\lambda_3 \delta^2 2\pi \nu + \ldots
\end{gather}
Thus with the same winding number, the two topological objects contribute with opposite signs to the Wilson Loop. This is important as we consider how the \mons with positive and negative winding numbers can be distributed across the lattice (see figure \ref{fig:WindingIllustration}). Assuming that $c$ is continuous apart from at the \monsp,  then we must have a positive winding number object next to a negative winding number object. Thus two $a = \pi/2$ objects next to each other (with opposite winding numbers) will give contributions to the Wilson Loop which approximately cancel (ultimately leading to a total contribution proportional to the square root of the number of objects, i.e. a perimeter law scaling for the Wilson Loop), and the same can be said for two $a = 0$ objects neighbouring each other. However, if we have a $a = \pi/2$ \mon with positive winding number situated as a spatial neighbour to an $a = 0$ \mon with negative winding number, then the contributions of the two objects would be additive (ultimately leading to a total contribution proportional to the number of objects, i.e. an area law scaling).\footnote{An alternative, which we don't consider in this discussion, as that there is another type of topological object, such as a monopole, neighbouring the \monp.}

\begin{figure}
 \begin{center}
  \includegraphics[height = 3cm]{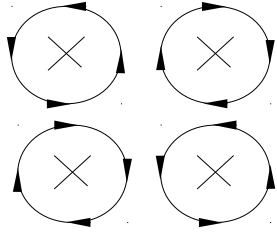}
 \end{center}
\caption{An Illustration of how topological objects can be distributed across a plane. The cross indicates the centre of the topological object. The circles and arrows indicate the winding of the parameter $c$ around each topological object. In this simplified picture, it is clear that if $c$ is continuous outside the topological object, the arrows of neighbouring objects must point in the same direction at the point of contact between them. In this simple picture, each object of positive winding number (clockwise arrows) must have objects with negative winding number neighbouring to it. While this can be avoided by having the direction of flow of $c$ doubling back on itself, such a scenario would require that the flow should be three times as large on another part of the circle, meaning that the problem is not eliminated but just transferred elsewhere.}\label{fig:WindingIllustration}
\end{figure}

The total Wilson Loop will be the product of a boundary term and contributions from all the \mons contained within the planar Wilson Loop. As we can expect the number of \mons to be proportional to the area of the loop, this leads to an area law scaling for the Wilson Loop and a linear string tension.

\subsection{SU3}
The SU(3) system is considerably more cumbersome than the SU(2) theory. We start by constructing the operators $n^3 = \theta \lambda^3 \theta^\dagger$, $n^8 = \theta \lambda^8 \theta^{\dagger}$ and $\theta^\dagger \partial_\mu \theta$, where we parametrise $\theta = e^{ia_3\phi_3}e^{ia_2\phi_2}e^{ia_1\phi_1}e^{id_3 \lambda^3 + i d_8 \lambda^8}$.
We find that
\begin{align}
\theta^\dagger \partial_\mu \theta = & S_{1\mu}+S_{2\mu}+S_{3\mu}+S_{4\mu}\\
S_{1\mu}=& e^{-id_3 \lambda^3 -i d_8 \lambda^8} \partial_\mu e^{id_3 \lambda^3 +i d_8 \lambda^8}\\
S_{2\mu}=&e^{-id_3 \lambda^3 -i d_8 \lambda^8}e^{-ia_1\phi_1} \partial_\mu (e^{ia_1\phi_1})e^{id_3 \lambda^3 +i d_8 \lambda^8}\\
S_{3\mu}=&e^{-id_3 \lambda^3 -i d_8 \lambda^8}e^{-ia_1\phi_1} e^{-ia_2\phi_2} \partial_\mu (e^{ia_2\phi_2}) e^{ia_1 \phi_1}e^{id_3 \lambda^3 +i d_8 \lambda^8}\\
S_{4\mu}=&e^{-id_3 \lambda^3 -i d_8 \lambda^8}e^{-ia_1\phi_1} e^{-ia_2\phi_2}e^{-ia_3\phi_3} \partial_\mu (e^{ia_3\phi_3})e^{ia_2\phi_2} e^{ia_1 \phi_1}e^{id_3 \lambda^3 +i d_8 \lambda^8},
\end{align}
where explicit expressions for $n^3$, $n^8$ and $S_{1\mu},\ldots,S_{4\mu}$ are given in  \ref{app:su3}.

Close to the \mon we can write the Abelian decomposition as
\begin{gather}
U_\sigma = \hat{X}_\sigma \theta_\sigma e^{i \delta\sigma\lambda^3 \hat{u}_{3\sigma} +i \delta\sigma\lambda^8 \hat{u}_ {8\sigma}} \theta^\dagger_{\sigma+\delta \sigma},
\end{gather}
which gives
\begin{gather}
e^{i\delta\sigma\lambda^3  \hat{u}_{3\sigma} +i\delta\sigma \lambda^8 \hat{u}_{8\sigma}} = \theta_\sigma^\dagger\hat{X}^\dagger U\theta_{\sigma+\delta \sigma}.
\end{gather}
The two components of the field strength will be given by the integral of $\hat{u}_3$ or $\hat{u}_8$ around a loop surrounding the \mon at an $a_i$ that have fallen to some background level. If the fields $\hat{X}$ and $U$ are smooth, and with $\theta$ differentiable along this path, we have
\begin{align}
\hat{u}_{3\sigma}  = &-i \frac{1}{\tr (\lambda^3)^2} \tr(\lambda^3 \theta^\dagger \partial_\sigma \theta)\nonumber\\
\hat{u}_{8\sigma}  = &-i \frac{1}{\tr (\lambda^8)^2} \tr(\lambda^8 \theta^\dagger \partial_\sigma \theta)\label{eq:62}
\end{align}
which implies that (using equations (\ref{eq:D1}) - (\ref{eq:D2}))
\begin{align}
\hat{u}^3_{\sigma} =&  \left(\partial_\sigma d_3 - \sin^2 a_1 \partial_\sigma c_1 + \frac{\sin^2 a_2}{2} \cos(2a_1) \partial_\sigma c_2 + \partial_\sigma c_3\frac{\sin^2 a_3}{2} {(\cos (2a_2) -2)} \cos(2a_1) \right)-\nonumber\\
& \left(\partial_\sigma c_3 \frac{\sin(2a_3)}{2} \sin a_2 \sin 2a_1 \sin(c_1 + c_3 - c_2)\right)+\partial_\sigma a_3 \sin a_2 \sin(2a_1) \cos(c_1+ c_3 - c_2)\nonumber\\
\hat{u}^8_{\sigma} =& \left(\partial_\sigma d_8 - \partial_\sigma c_2 \frac{\sqrt{3}\sin^2 a_2}{2} + \sqrt{3} \partial_\sigma c_3 \frac{\sin^2 a_3}{2}\cos(2a_2)\right).
\end{align}
The components of the field strength tensor are then given by (excluding terms which will not contribute either because they are multiples of $2\pi$ or zero)
\begin{align}
\hat{F}^3_{\mu\nu} d\Sigma^{\mu\nu} = & \oint_{\tilde{C}} d\sigma \bigg( - \sin^2 a_1 \partial_\mu c_1 + \frac{\sin^2 a_2}{2} \cos(2a_1) \partial_\mu c_2 + \partial_\mu c_3\frac{\sin^2 a_3}{2} (\cos (2a_2) -2) \cos(2a_1)-\nonumber\\
&\phantom{\oint_C d\sigma \bigg(}\partial_\mu c_2 \frac{\sin 2 a_3}{2} \sin a_2 \sin2a_1 \sin(c_1 + c_3 - c_2)\bigg)\nonumber\\
\hat{F}^8_{\mu\nu} d\Sigma^{\mu\nu} = & \oint_{\tilde{C}} d\sigma \left(- \partial_\mu c_2 \frac{\sqrt{3}\sin^2 a_2}{2} + \sqrt{3} \partial_\mu c_3 \frac{\sin^2 a_3}{2}\cos(2a_2)\right).\label{eq:fsu3}
\end{align}
where $d\Sigma$ is an element of area which covers the \mon and $\tilde{C}$ the curve which surrounds the \monp. Only those $c_i$s which wind around the \mon will vary on this infinitesimal loop; all the other parameters can be treated as constant. These field strengths are non-zero if one of the $c_i$ has a non-zero winding number around the \monp, and they will not always be an integer multiple of $2\pi$. In particular, for the terms proportional to $\partial_\mu c_2$ and $\partial_\mu c_3$ we may pick infinitesimal loops where the $a$ parameters do not vary significantly. The \mons thus will lead to a non trivial structure in the restricted field strength, and if the number of \mons within an area is proportional to that area, an area law for the expectation value of the Wilson loop and confinement.
\subsection{The Electromagnetic field}\label{sec:4.3}
We saw in equation (\ref{eq:7}) that the field strength can be written as,
\begin{gather}
g\hat{F}_{\mu\nu}^j = \partial_\mu \hat {u}^j_\nu - \partial_\nu \hat {u}^j_\mu
\end{gather}
with (from equation (\ref{eq:reallydull}))
\begin{gather}
\hat{u}^j_\mu = -\frac{1}{2}\tr(n^j(gA_\mu - i\theta \partial_\mu \theta^\dagger)).
\end{gather}
In~\cite{Cho:1980}, the potential is split into two parts $A_E=\tr\; n^j A_\mu$ and $A_M=\frac{1}{2g} i \tr\; n^j \theta \partial_\mu \theta^\dagger$. The field strength,
\begin{gather}
\hat{F}_{\mu\nu}^j = \frac{1}{2}\left[\partial_\mu \tr (n^j A_\nu) - \partial_\nu \tr (n^j A_\mu)\right] + \frac{i}{8g}\tr(n^j[\partial_\mu n^k,\partial_\nu n^k]),
\end{gather}
is often split into two components,\footnote{In the literature, $\hat{E}$ is known as `electric' and $\hat{H}$ as `magnetic', but we will not be following this terminology but refer to the $ti$ components of the fields as electric and the $ij$ components as magnetic, where $i$ and $j$ are spatial directions. For a monopole, $\hat{H}$ is purely magnetic and the terminology is sensible; however \mons are not monopoles and will usually also carry an electric field.  }
\begin{align}
\hat{E}_{\mu\nu}^j =&\frac{1}{2}\left[\partial_\mu \tr (n^j A_\nu) - \partial_\nu \tr (n^j A_\mu)\right] \nonumber\\
\hat{H}_{\mu\nu}^j =& \frac{i}{8g}\tr(n^j[\partial_\mu n^k,\partial_\nu n^k]).
\end{align}
The magnetic current $k^\mu$ can be defined through Maxwell's equation
\begin{gather}
\frac{1}{2} \epsilon^{\mu\nu\rho\sigma}\partial_\nu \hat{F}_{\rho\sigma} = k^\mu,
\end{gather}
and while $\frac{1}{2} \epsilon^{\mu\nu\rho\sigma}\partial_\nu \hat{E}_{\rho\sigma} = 0$, as in electromagnetism, in general $\frac{1}{2} \epsilon^{\mu\nu\rho\sigma}\partial_\nu \hat{H}_{\rho\sigma} \neq 0$. This means that this field strength, coming from the colour field rather than the gauge field, is associated with a magnetic source. Note, however, that it is not in general a pure monopole (except for certain specific choices of $\theta$), since the electric current, defined through $\partial^\mu \hat{H}_{\mu\nu} \neq 0$. Neither $\hat{H}_{\mu\nu}$ nor $\hat{E}_{\mu\nu}$ are by themselves gauge invariant, but only $\hat{F} = \hat{E}+\hat{H}$. We expect that close to the \mon{} solutions $\hat{F}$ will be dominated by $\hat{H}$ since this is the part of $\hat{F}$ which is sensitive to the winding of $c$ around the centre of the \monp. In analogy to QED, we shall refer to the components of $\hat{F}$ (and also $\hat{H}$) as the electromagnetic fields.
\subsubsection{SU(2)}
In SU(2), we see that (using equations (\ref{eq:tdtsu2}) and (\ref{eq:F1}) and setting $d=0$),
\begin{align}
\hat{H}_{\mu\nu}^3 =&-\frac{i}{2g}\tr \lambda^3 \left[  \partial_\mu a \phi + \sin a \cos a \bar{\phi}\partial_\mu c -  \sin^2 a \partial_\mu c \lambda^3,\partial_\nu a \phi + \sin a \cos a \bar{\phi}\partial_\nu c -  \sin^2 a \partial_\nu c \lambda^3 \right]\nonumber\\
=&-\frac{\sin (2a)}{g} (\partial_\mu a \partial_\nu c - \partial_\mu c \partial_\nu a).\label{eq:H3}
\end{align}
The Yang-Wu monopole~\cite{YangWu,tHooft1976,Polyakov:1974} is one classic magnetic monopole solution for  winding number $\nu = 1$. Here we parametrise the coordinates as (with the origin at the location of the monopole)
\begin{gather}
\left(\begin{array}{c} t\\x\\y\\z\end{array}\right) = r \left(\begin{array}{c}\cos\psi_3\\ \sin\psi_3\cos\psi_2\\\sin \psi_3\sin\psi_2\cos\psi_1\\\sin\psi_3\sin\psi_2\sin\psi_1 \end{array}\right).\label{eq:cs1}
\end{gather}
The reason we choose this parametrisation is so that we can easily respect the symmetries of the theory. Since we are considering point-like objects, it is most natural to use a spherical polar coordinate system. QCD, after averaging over the gauge links, possesses a spherical symmetry, which we have broken by inserting the Wilson Loop in the $xt$ plane. However, we should still be free to rotate the coordinate axes in the $yz$ plane. This particular parametrisation allows us to absorb such rotations into a shift of $\psi_1$. If, however, we parametrised the coordinates differently (for example by exchanging $t$ and $z$ in equation (\ref{eq:cs1})), a rotation in the $yz$ plane would require some non-trivial change of all the angular variables. The only changes we can easily make to the coordinate parametrisation are to interchange the $t$ and $x$ axes or rotate in the $yz$ plane (we shall only consider rotations of $\pi/2$ since we are primarily interested in comparing with a lattice theory).

For the Yang-Wu monopole we select the solution with $c = \psi_1$ and $a = \psi_2/2$.\footnote{Here one would have to choose the coordinates parametrising $\theta$ within the ranges $0<a<\pi$, $-\pi/2<d<\pi/2$, $c \in \mathbb{R}$, rather than as we set it out earlier. Note that for this monopole, $\theta$ is discontinuous ($\theta(\psi_2 = 2\pi) = - \theta(\psi_2 = 0)$) but $n^3$ remains continuous.} This leads to  the solution for the electric and magnetic fields $\mathbf{E} = 0$, $\mathbf{B} = \frac{1}{2g} \frac{\mathbf{r_s}}{r_s^3}$, where $\mathbf{r_s} = (x,y,z)$; in other words the true magnetic analogue of a point like electric charge. Note that there is no time dependence in the monopole field strength: these objects will reveal themselves as stings of electromagnetic flux extending in the time direction.

The literature concentrates on the possibility of finding these monopoles in the colour field, which are then assumed to condense and provide confinement through a dual Meissner effect. We will later, therefore, in addition to the \mon solution suggested above, consider what structures may appear in the fields if these Yang-Wu monopoles dominate.

In our case, we have something a little different. We will start by considering the coordinate system given in equation (\ref{eq:cs1}). We consider the objects close to $a = \pi/2$ here; the objects at $a = 0$ can be treated in a similar way. We can then write the variables, at all locations away from the singularity, as $a = \pi/2 - G(r,\psi_1,\psi_2,\psi_3)$ where $G$ is zero at $r = 0$ and $c = R(r,\psi_1,\psi_2,\psi_3)$; and for a winding number one \mon we require $R(\psi_3+2\pi) = R(\psi_3)+2\pi$. The functions $R$ and $G$ will depend on the gauge, but in a way which allows $\hat{F}^3$ to remain gauge-invariant. This gives us
\begin{align}
\tan\psi_1 = &\frac{z}{y}&\tan\psi_2 = &\frac{\sqrt{z^2 + y^2}}{x}\nonumber\\
\tan\psi_3 =& \frac{\sqrt{x^2 + y^2 + z^2}}{t}& r = & \sqrt{x^2 + y^2 + z^2 + t^2},
\end{align}
and
\begin{align}
\frac{\partial \psi_1}{\partial z} =& \frac{\cos \psi_1}{r \sin\psi_2 \sin \psi_3} =\frac{y}{z^2 + y^2}\nonumber\\
\frac{\partial \psi_1}{\partial y} =& -\frac{\sin \psi_1}{r \sin\psi_2 \sin \psi_3} = -\frac{z}{z^2 + y^2}\nonumber\\
\frac{\partial \psi_2}{\partial z} =& \frac{\sin \psi_1 \cos \psi_2}{r  \sin \psi_3} = \frac{xz}{\sqrt{z^2 + y^2}(x^2 + y^2 + z^2)} \nonumber\\
\frac{\partial \psi_2}{\partial y} =& \frac{\cos \psi_1 \cos\psi_2}{r  \sin \psi_3}= \frac{xy}{\sqrt{z^2 + y^2}(x^2 + y^2 + z^2)}\nonumber\\
\frac{\partial \psi_2}{\partial x} =& -\frac{\sin\psi_2}{r  \sin \psi_3}=-\frac{\sqrt{z^2 + y^2}}{x^2 + y^2 + z^2}\nonumber\\
\frac{\partial \psi_3}{\partial z} =& \frac{\cos\psi_3\sin\psi_2\sin \psi_1}{r } = \frac{zt}{\sqrt{x^2 + y^2 + z^2}(x^2 + y^2 + z^2 + t^2)} \nonumber\\
\frac{\partial \psi_3}{\partial y} =& \frac{\cos\psi_3\sin\psi_2\cos \psi_1}{r}=\frac{yt}{\sqrt{x^2 + y^2 + z^2}(x^2 + y^2 + z^2 + t^2)}\nonumber\\
\frac{\partial \psi_3}{\partial x} =& \frac{\cos\psi_3\cos\psi_2}{r}= \frac{xt}{\sqrt{x^2 + y^2 + z^2}(x^2 + y^2 + z^2 + t^2)}\nonumber\\
\frac{\partial \psi_3}{\partial t} =& -\frac{\sin\psi_3}{r}=-\frac{\sqrt{x^2 + y^2 + z^2}}{(x^2 + y^2 + z^2 + t^2)}\nonumber.
\end{align}
The simplest possibility is to set $G=G(kr)$ for some $k$ and $R = R(\psi_3)$, neglecting the $\psi_1$ and $\psi_2$ dependence. We shall consider a more general solution later. Using equation (\ref{eq:H3}), we obtain,
\begin{align}
\hat{H}^3_{zt} = \hat{E}^3_z = & \frac{1}{g}\sin (2 G) R' G' \left(\frac{\partial r}{\partial z} \frac{\partial \psi_3}{\partial t} -  \frac{\partial r}{\partial t} \frac{\partial \psi_3}{\partial z} \right)\nonumber\\
=&\frac{\sin (2 G)}{r} R' G' \left( -\sin^2\psi_3 \sin\psi_2 \sin\psi_1 - \cos^2 \psi_3 \sin\psi_2 \sin \psi_1\right)\nonumber\\
=&-\frac{\sin (2 G)}{r} R' G' \frac{z}{\sqrt{x^2 + y^2 + z^2}},
\end{align}
where $R'$ and $G'$ indicate the first derivative of these functions and $\mathbf{E}^3$ is the electric field. Similarly,
\begin{align}
\hat{H}^3_{yt} = & \hat{E}^3_y = -\frac{1}{g}\frac{\sin (2 G)}{r} R' G' \frac{y}{\sqrt{x^2 + y^2 + z^2}}\nonumber\\
\hat{H}^3_{xt} = & \hat{E}^3_x = -\frac{1}{g}\frac{\sin (2 G)}{r} R' G' \frac{x}{\sqrt{x^2 + y^2 + z^2}}.
\end{align}
For the magnetic field $\mathbf{B}^3$ we find,
\begin{align}
\hat{H}^3_{zy} = -\hat{B}^3_x =&   \frac{1}{g}{\sin 2G} G' R' \left(\frac{z}{r} \frac{\partial \psi_3}{\partial y} - \frac{y}{r} \frac{\partial \psi_3}{\partial z}\right)\nonumber\\
=&-\frac{1}{g}\frac{\sin 2G}{r} G' R' \left( \sin\psi_3\cos \psi_3 \sin^2 \psi_2\sin\psi_1 \cos \psi_1 - \sin\psi_3\cos \psi_3 \sin^2 \psi_2\sin\psi_1 \cos \psi_1\right)\nonumber\\
 =& 0\nonumber\\
\hat{H}^3_{xz} = -\hat{B}^3_y = & 0\nonumber\\
\hat{H}^3_{yx} = -\hat{B}^3_z = & 0.
\end{align}
Here, the \mon generates a pure electric field.
This is, of course, not the only way we can parametrise the coordinates. For example, we could have chosen the coordinate system
\begin{gather}
\left(\begin{array}{c} z\\y\\t\\x\end{array}\right) = r \left(\begin{array}{c} \sin \psi_3\sin\psi_2\sin\psi_1\\\sin\psi_3\sin\psi_2\cos\psi_1\\ \sin\psi_3\cos\psi_2\\\cos\psi_3\end{array}\right),
\end{gather}
which in effect swaps $t$ and $x$. This gives a different solution for the field strength tensor
\begin{align}
\hat{H}^3_{xt} = \hat{E}^3_x = & \frac{1}{g}\frac{\sin (2 G)}{r} R' G' \frac{t}{\sqrt{t^2 + y^2 + z^2}}\nonumber\\
\hat{H}^3_{yx} = -\hat{B}^3_z = & -\frac{1}{g}\frac{\sin (2 G)}{r} R' G' \frac{y}{\sqrt{z^2 + y^2 + t^2}}\nonumber\\
\hat{H}^3_{zx} = \hat{B}^3_y = &-\frac{1}{g}\frac{\sin (2 G)}{r} R' G' \frac{z}{\sqrt{z^2 + y^2 + t^2}},\label{eq:ds2}
\end{align}
with the other components of $\hat{H}^3$ zero. We could also use any other rotation of the coordinates consistent with the space time symmetries.

For a general function $G(r,\psi_1,\psi_2,\psi_3)$ and $R(\psi_1,\psi_2,\psi_3)$ (neglecting any radial dependence of $c$, which may in practice only be valid for small enough distances and certain choices of gauge), with $\partial_i \equiv \partial/\partial \psi_i$ and $r^2_{yz} = y^2+z^2$ and $r_{xyz}^2 = x^2+y^2+z^2$, we find for the parametrisation (\ref{eq:cs1})
\begin{align}
\hat{H}^3_{zt} = &-\frac{1}{g}\sin 2G\bigg(\partial_r  G\partial_3 R \frac{z}{r r_{xyz}} - (\partial_1 G \partial_3 R - \partial_3G\partial_1R) \frac{y r_{xyz}}{r^2 r_{yz}^2} - \nonumber\\
&\phantom{-\frac{1}{g}\sin 2G ()}(\partial_2G\partial_3R-\partial_3G \partial_2 R)\frac{zx}{r_{yz}r_{xyz} r^2} + \partial_1R \partial_r G \frac{t y}{r r_{yz}^2} + \partial_rG  \partial_2 R \frac{zxt}{r_{yz}r^2_{zyx} r}  \bigg)
\nonumber\\
\hat{H}^3_{yt} = &-\frac{1}{g}\sin 2G \bigg(\partial_r G \partial_3 R  \frac{y}{r_{xyz} r} + \partial_r G \partial_2 R \frac{xyt}{r_{yz}r^2_{zyx} r} - \partial_r G \partial_1 R \frac{zt}{r r_{yz}^2} + \nonumber\\
&\phantom{-\frac{1}{g}\sin 2G ()}(\partial_1R \partial_3 G - \partial_1 G \partial_3 R) \frac{z r_{xyz}}{r^2 r_{yz}^2} + (\partial_2 G \partial_3 R- \partial_2 R \partial_3 G)\frac{xy}{r^2 r_{xyz}r_{yz}} \bigg)
\nonumber\\
\hat{H}^3_{xt} = &-\frac{1}{g}\sin2G\bigg(\partial_rG  \partial_3 R \frac{x}{r r_{xyz}} - \partial_2 G \partial_3 R \frac{r_{yz}}{r^2 r_{xyz}} - \partial_r G \partial_2 R \frac{t r_{yz}}{r r^2_{zyx}} + \partial_3 G \partial_2 R \frac{r_{yz}}{r^2 r_{xyz}} \bigg)
\nonumber\\
\hat{H}^3_{zy} = &\frac{1}{g}\sin 2G \bigg(-\partial_r G \partial_1 R \frac{1}{r} + (\partial_1 G \partial_2 R  - \partial_2 G \partial_1 R)\frac{x}{r_{yz} r^2_{zyx}} + (\partial_1G  \partial_3 R  - \partial_3 G \partial_1 R)\frac{t}{r_{xyz} r^2} \bigg)
\nonumber\\
\hat{H}^3_{zx} = &\frac{1}{g}\sin 2G \bigg(-\partial_r G \partial_1 R \frac{xy}{r r_{yz}^2} - \partial_r G \partial_2 R \frac{z}{r r_{yz}} +(\partial_2 G \partial_1 R - \partial_1 G \partial_2 R)\frac{y}{r_{yz} r_{xyz}^2} + \nonumber\\
&\phantom{-\frac{1}{g}\sin 2G ()}(\partial_1 G \partial_3 R - \partial_3 G \partial_1 R) \frac{yxt}{r^2 r_{xyz} r_{yz}^2} + (\partial_2 G \partial_3 R - \partial_3 G \partial_2 R) \frac{zt}{r^2 r_{xyz} r_{yz}}\bigg)
\nonumber\\
\hat{H}^3_{yx} = &\frac{1}{g}\sin 2G \bigg(\partial_r G \partial_1 R \frac{xz}{r r_{yz}^2} - \partial_r G \partial_2 R \frac{y}{r r_{yz}} + (\partial_1 G \partial_2 R - \partial_1 R \partial_2 G) \frac{z}{r_{yz} r^2_{zyx}} + \nonumber\\
&\phantom{-\frac{1}{g}\sin 2G ()}(\partial_3 G \partial_1 R - \partial_1 G \partial_3 R) \frac{zxt}{r^2 r_{yz}^2 r_{xyz}}\bigg)\label{eq:biggy}
\end{align}
We must also consider every permutation and rotation of the coordinates to obtain other possible solutions.

At small $r$, each of these terms gives a finite and non-zero result if the derivatives of $G$ and $R$ are non-zero and both $G$ and $\partial_{\psi_i} G$ are proportional to $r$.
Our ultimate goal is to measure the electromagnetic fields on the lattice to try to identify these structures, and this means evaluating how the fields behave at large and intermediate distances.
 This is complicated by the unknown functions $G$ and $R$ (which are also not gauge invariant and therefore uniquely defined). We require that at least one of the $\partial_{\psi_i} R \neq 0$ (and, in practice, on average take an integer value).
 We can, however, study the explicit coordinate dependence within equation (\ref{eq:biggy}) for $G \sim r$ (obviously this is in practice only an approximation, as we only consider part of the expression for the field strength, but nonetheless we may expect these features to be visible to a certain extent).
We can see three types of behaviour at large distance. Firstly, there are terms such as $(\partial_1 G \partial_2 R - \partial_2 R \partial_1 G) \frac{z}{r_{yz} r^2_{xyz}}$ (taken from $\hat{H}^3_{yx}$) which fall away at least as $1/z_i$ for large spatial distances, while increasing at large $t$ if
 $\partial_i G$ is linear in $r$ -- we find a possible solution where the field grows with distance, although we may hope that these are excluded in practice.
Secondly, we have terms which form a line of large electromagnetic field in one of
the directions, for example in $\hat{H}^3_{xz}$ we find the term
\begin{gather}
\frac{1}{g}\sin2G\partial_rG  \partial_1 R \frac{xy}{r r_{yz}^2} \le \frac{1}{g}\partial_rG  \partial_1 R \frac{xy}{r r_{yz}^2} ,
\end{gather}
and the maxima of this function fall off at large distances from the origin as $1/z^3$, $1/y^2$ and $1/t$ but remain roughly independent of $x$. This will correspond to a line of large magnetic field (which may be broken or change sign) in the $x$-direction, and we will label it as a $\mu$-string, where $\mu$ indicates the direction (in practice, the line of field strength will only have a finite extent because of the $G$ and $R$ dependence). We do not mean anything by this terminology beyond that the large field may extend in the $\mu$-direction and falls of rapidly in the other directions. We require that the leading order $r$ dependence in $G$ should be linear if we wish to avoid zero or infinite field strengths at $r=0$. Alternatively, by permuting the coordinates, we also find these strings may exist in the other electromagnetic fields. Finally, we have what we will label as points where the electromagnetic flux falls as at least $1/r$ in each direction at large distances, $r\gtrsim 1/k$, for example the term $\sin 2G \partial_rG \partial_1R \frac{1}{r}$ in $\hat{H}^3_{yz}$.

We will concentrate on solutions with $\partial_1 R = 0$ and $\partial_2 R = 0$. Our reason for doing so is that we expect the discontinuity to occur in a point; thus there is no obvious reason why the objects should wind around this discontinuity in the $yz$ plane and not the $xt$ plane; if the origin is at the discontinuity and $c$ proportional to the angle we might expect $c(-\delta x_\mu) = c(\delta x_\mu) + \pi$, and only setting $c= \psi_3$ accomplishes this. If there is no winding number related to variations in $\psi_1$ and $\psi_2$, then we might expect any fluctuations in the fields in these directions to cancel out as we integrate around a loop surrounding the discontinuity. Also, only winding in the $xt$ plane can contribute to confinement. In our numerical simulations, we only use Wilson Loops in the $xt$-plane, so we restrict ourselves to these loops here. This means that we should expect a symmetry between the structures in the $x$ and $t$ components of the electro-magnetic field, and also between the $y$ and $z$ components of the field,    but not necessarily between (for example) the $x$ and $y$ components. This all suggests that it is safe to neglect this angular dependence in $R$. With this simplification, we obtain,
\begin{align}
\hat{H}^3_{zt} = \hat{E}_z^3=&-\frac{1}{g}\sin 2G\bigg(\underbrace{\partial_r  G\partial_3 R \frac{z}{r r_{xyz}}}_{\text{point}} - \underbrace{\partial_1 G \partial_3 R \frac{y r_{xyz}}{r^2 r_{yz}^2}}_{\text{x-string}} -\underbrace{\partial_2G\partial_3R\frac{zx}{r_{yz}r_{xyz} r^2}}_{\text{point}}  \bigg)
\nonumber\\
\hat{H}^3_{yt} = \hat{E}_y^3=&-\frac{1}{g}\sin 2G \bigg(\underbrace{\partial_r G \partial_3 R  \frac{y}{r_{xyz} r}}_{\text{point}} - \underbrace{\partial_1 G \partial_3 R \frac{z r_{xyz}}{r^2 r_{yz}^2}}_{\text{x-string}} + \underbrace{\partial_2 G \partial_3 R\frac{xy}{r^2 r_{xyz}r_{yz}}}_{\text{point}} \bigg)
\nonumber\\
\hat{H}^3_{xt} =\hat{E}_x^3= &-\frac{1}{g}\sin2G\bigg(\underbrace{\partial_rG  \partial_3 R \frac{x}{r r_{xyz}}}_{\text{point}} - \underbrace{\partial_2 G \partial_3 R \frac{r_{yz}}{r^2 r_{xyz}}}_{\text{point}}  \bigg)
\nonumber\\
\hat{H}^3_{zy} = \hat{B}_x^3 =&\frac{1}{g}\sin 2G \bigg( \underbrace{\partial_1G  \partial_3 R \frac{t}{r_{xyz} r^2}}_{\text{t-string}} \bigg)
\nonumber\\
\hat{H}^3_{zx} = \hat{B}_y^3=&\frac{1}{g}\sin 2G \bigg(\underbrace{\partial_1 G \partial_3 R  \frac{yxt}{r^2 r_{xyz} r_{yz}^2}}_{\text{t-string}} + \underbrace{\partial_2 G \partial_3 R  \frac{zt}{r^2 r_{xyz} r_{yz}}}_{\text{t-string}}\bigg)
\nonumber\\
\hat{H}^3_{yx} = \hat{B}_z^3=&-\frac{1}{g}\sin 2G \bigg( \underbrace{\partial_1 G \partial_3 R \frac{zxt}{r^2 r_{yz}^2 r_{xyz}}}_{\text{t-string}}\bigg).\label{eq:70}
\end{align}
In this parametrisation of the coordinates, the electromagnetic fields will therefore appear either as points (not necessarily spherically symmetric, but falling in every direction) or strings in either the $x$ or $t$ direction. We also consider the three other parametrisations of the coordinates which arise from the exchanges $x \leftrightarrow t$ and $y \leftrightarrow z$. We summarise the results in table \ref{tab:lotsofstrings}. With these choices of coordinate parametrisations, we see that $\hat{E}_x$ will only have points, $\hat{E}_y$ and $E_
z$ strings along the $x$ axis, $\hat{B}_x$ may have a string along the $x$ or $t$ axis, while $\hat{B}_y$ and $\hat{B}_z$ strings along the time axis. Of course, in practice this picture will be obscured by fluctuations in the fields and effects from the unknown functions $G$ and $R$; we also cannot say how long the strings will be. It also depends on that $G$ has some angular dependence, and in particular has some variation in the $yz$ plane (as parametrised by the angle $\psi_1$). Equally, the strings depend on different angular differentials of $G$; if $G$ has no angular dependence they will not be present, and if it depends on $\psi_2$ but not on $\psi_1$ only some of these strings will appear. We will search for these patterns of electromagnetic fields in in lattice QCD in section \ref{sec:6.3}.
\begin{table}
\begin{center}
\begin{tabular}{l l l l l l l}
\hline
Parametrisation&$\hat{E}^j_x$&$\hat{E}^j_y$&$\hat{E}^j_z$&$\hat{B}^j_x$& $\hat{B}^j_y$&$\hat{B}_z^j$
\\
\hline
Equation (\ref{eq:cs1})&point&x-string&x-string&t-string&t-string&t-string\\
$t\leftrightarrow x$&point & x-string&x-string&x-string&t-string&t-string\\
$y\leftrightarrow z$&point&x-string&x-string&t-string&t-string&t-string\\
$t\leftrightarrow x$,$y\leftrightarrow z$&point&x-string&x-string&x-string&t-string&t-string\\
\hline\hline
\end{tabular}
\end{center}
\caption{Some Possible configurations of the restricted electromagnetic fields; every time we specify a string there may also additionally be points. This table lists the fields which could in principle be present; it is not necessary that each of them would be visible for a given \monp in practice.}\label{tab:lotsofstrings}
\end{table}

We can also consider the patterns we can expect from the Wang-Yu monopole and other solutions for the fields. We have already seen that the monopole solution with the coordinate parametrisation given in equation (\ref{eq:cs1}) gives t-strings in the magnetic field; we may also study the other options when we permute the coordinates. Additionally, we list a few similar simple alternative forms for the functions $R$ and $G$. These are summarised in table \ref{tab:los2}. It should be observed that if $a$ is linked to an angular variable then either $\hat{E}_x=0$ or $\hat{B}_x=0$; while if $c$ is linked to the angular variable $\psi_2$ rather than $\psi_3$ with $G\sim r$, then the Electric fields have t-strings and the magnetic fields x-strings: the opposite of our expected model in table \ref{tab:lotsofstrings}. We note that a combination of $a = \psi_3/2;c=\psi_1$ monopoles and $a = \psi_3/2; c = \psi_2$ structures may have the same pattern of points and strings as the configuration of table $\ref{tab:lotsofstrings}$; the
difference is that there is in this case no correlation between the spatial locations of the strings and the points in the $\hat{E}^3_x$ field.
\begin{table}
\begin{center}
\begin{tabular}{l l l l l l l}
\hline
Parametrisation&$\hat{E}^j_x$&$\hat{E}^j_y$&$\hat{E}^j_z$&$\hat{B}^j_x$& $\hat{B}^j_y$&$\hat{B}^j_z$\\
\hline\hline
Equation (\ref{eq:cs1})&0&0&0&t-string&t-string&t-string\\
$t\leftrightarrow x$&0 & x-string&x-string&x-string&0&0\\
$y\leftrightarrow z$&0&0&0&t-string&t-string&t-string\\
$t\leftrightarrow x$,$y\leftrightarrow z$&0&x-string&x-string&x-string&0&0\\
\hline\hline
Equation (\ref{eq:cs1})&point&point&point&0&point&0\\
$t\leftrightarrow x$&point &0&point&0&point&point\\
$y\leftrightarrow z$&point&point&point&0&0&point\\
$t\leftrightarrow x$,$y\leftrightarrow z$&point&point&0&0&point&point\\
\hline\hline
Equation (\ref{eq:cs1})&0&point&point&point&point&point\\
$t\leftrightarrow x$&0&point&point&point&point&point\\
$y\leftrightarrow z$&0&point&point&point&point&point\\
$t\leftrightarrow x$,$y\leftrightarrow z$&0&point&point&point&point&point\\
\hline\hline
Equation (\ref{eq:cs1})&t-string&t-string&t-string&0&point&point\\
$t\leftrightarrow x$&x-string&point&point&0&x-string&x-string\\
$y\leftrightarrow z$&t-string&t-string&t-string&0&point&point\\
$t\leftrightarrow x$,$y\leftrightarrow z$&x-string&point&point&0&x-string&x-string\\
\hline\hline
\end{tabular}
\end{center}
\caption{Some Possible configurations of the restricted electromagnetic fields for the parametrisations $a = \psi_2/2$, $c = \psi_1$ (including the Wang-Yu monopole) (top); $a = \psi_3/2; c = \psi_2$ (second top); $a =\psi_3/2; c=\psi_1$ (second bottom); $G=kr;c=\psi_2$ (bottom) }\label{tab:los2}
\end{table}

That the structures in the field strength are not point like but contain vortices is important because it means that we cannot evade the large field strength by distorting the surface $\Sigma$. For example, we may create a bump in $\Sigma$ to avoid the \mon{} and the structure in the $\hat{F}_{xt}$ field: however, this will require integrating over a surface in (for example) the $yt$ plane where $\hat{F}_{yt}$ is large, and the $yx$ plane where $F_{yx}$ is large. This means that it is not naively ruled out that the surface integral bound by the Wilson loop is not affected by the choice of surface to integrate over.
\subsubsection{SU(3)}
 SU(3) gives a similar picture as SU(2); here we only outline the argument about which structures will emerge in the field strength. The starting point is equation (\ref{eq:62}), and by differentiating this we can construct the two field strengths $\hat{F}^3$ and $\hat{F}^8$. For example, we find that
\begin{multline}
 \hat{F}^8_{\mu\nu} = \partial_\mu \hat{u}^8_\nu - \partial_\nu \hat{u}^8_\mu = -\sqrt{3}\sin (2 a_2) (\partial_\mu a_2 \partial_\nu c_2 - \partial_\mu c_2 \partial _\nu a_2) +\\ \sqrt{3} \cos(2a_2) \sin(2a_3) (\partial_\mu a_3 \partial_\nu c_3 - \partial_\mu c_3 \partial _\nu a_3) +\\ \sqrt{3} \sin^2 a_3 \sin (2a_2) (\partial_\nu c_3 \partial_\mu a_2 - \partial_\mu c_3 \partial_\nu a_2).
\end{multline}
The first two terms are similar to the quantities already discussed in SU(2), and the same analysis applies. The third term, mixing the differential of $c_3$ with that of $a_2$, is new and we need to consider this as well. $\hat{F}^3_{\mu\nu}$ is similar to $\hat{F}^8_{\mu\nu}$, with similar expressions arising, plus additionally terms proportional to the product of the differential of to of the $c_i$s, for example $\partial_\mu c_1 \partial_\nu c_2$. We therefore expect to see the same structures as in SU(2), plus possibly some additional structures from these extra terms. For example, in $\hat{F}^8_{it}$ with $c_3 \sim \psi_3$, we find the contribution $\frac{x_i t}{r^2 r_{xyz}} \frac{\partial a_2}{\partial t} - \frac{r_{xyz}}{r^2 } \frac{\partial a_2}{\partial x_i}$, while in $\hat{F^8}_{ij}$ we find the terms $\frac{x_i t}{r^2 r_{xzy}} \frac{\partial a_2}{\partial x_j} - \frac{x_j t}{r^2 r_{xzy}} \frac{\partial a_2}{\partial x_i}$. In general $a_2$, unlike $a_3$, will not be at its minimum or maximum value at the \mon{}. There is therefore no reason to require that $\partial a_2/\partial \psi_{1,2}$ will be proportional to $r$ at small $r$ (we can discount the possibility that $c_2$ and $c_3$ both have winding around the same point). We can also neglect terms such as $\partial_\mu c_1 \partial_\nu c_2$ because we do not expect both $c_1$ and $c_2$ to wind around the same point. This means that these additional objects in the electromagnetic field strength will just be points. The same argument applies to all the additional terms in the $\hat{F}^3_{\mu\nu}$. Therefore, the electromagnetic structure in SU(3) will be the same as in SU(2) but there might be additional point-like structures appearing in all the components of the electromagnetic field; this applies equally for both of the electromagnetic field tensors, $\hat{F}^{3}_{\mu\nu}$ and $\hat{F}^{8}_{\mu\nu}$.
\subsection{Gauge Invariance}\label{sec:4.4}
$\hat{F}^j$ is gauge invariant, which is what we expect for the mechanism behind confinement. However $\theta$, $X$ and $U$ are dependent on the gauge. This presents a conceptual difficulty, since we have stated that the gauge invariant field strength is influenced by objects which arise due to features in the gauge dependent parametrisation of $\theta$. The key contribution, however, comes from the winding of $c$ around the \monp, and we here show that this is gauge-invariant (at least in SU(2)).

Since this formalism requires that the gauge field is continuous, we can only use continuous gauge transformations. However, to undo the winding in $\theta$ we require a discontinuous gauge transformation.
Thus the discontinuity in $\theta$ at $a = \pi/2$ will survive any smooth gauge transformation, and the \mons will remain. For example, in SU(2), we can parametrise an infinitesimal gauge transformation as
\begin{gather}
\Lambda = \left(\begin{array}{cc}\cos l_1& i \sin l_1 e^{il_2}\\ i \sin l_1 e^{-il_2} &\cos l_1\end{array}\right)\left(\begin{array}{cc} e^{il_3}&0\\0&e^{-il_3}\end{array}\right),
\end{gather}
with $l_1$ and $l_3$ infinitesimal and $0<l_2<2\pi$. Recall that under a gauge transformation $\theta \rightarrow \Lambda\theta$.

Multiplying $\Lambda\theta$ together (choosing $d_3 = 0$) gives
\begin{gather}
 \Lambda\theta = \left( \begin{array}{cc}
                         \cos l_1 \cos a - l_1 \sin a e^{i(l_2-c-2l_3)}&i\sin a e^{i(c+2l_3)} + i \cos a l_1 e^{il_2}\\
                         i\sin a e^{-i(c+2l_3)} + i \cos a l_1 e^{-il_2}&\cos l_1 \cos a - l_1 \sin a e^{-i(l_2-c-2l_3)}
                        \end{array}\right)
                        \left(\begin{array}{cc}
                               e^{il_3}&0\\
                               0&e^{-il_3}
                              \end{array}\right).
\end{gather}
If we define
\begin{gather}
 \cos a' = \sqrt{\cos^2 a (1-l_1^2) - 2l_1 \sin a\cos a \cos (l_2-c) + l_1^2 \sin^2 a }, \label{eq:cosap}
\end{gather}
then we can rewrite this as
\begin{multline}
 \Lambda\theta = \left( \begin{array}{cc}
                         \cos a'&i \sin a' e^{ic'}\\
                         i\sin a' e^{-ic'}&\cos a'
                        \end{array}\right)\\
                        \left(\begin{array}{cc}
                               \frac{\cos l_1 \cos a - l_1 \sin a e^{i(l_2-c-2l_3)}}{\cos a'}e^{il_3}&0\\
                               0&\frac{\cos l_1 \cos a - l_1 \sin a e^{-i(l_2-c-2l_3)}}{\cos a'}e^{-il_3}
                              \end{array}\right)\nonumber
                              \end{multline}
                              \begin{gather}
\sin a'e^{ic'} = \frac{(\sin a e^{i(c+2l_3)} + \cos a l_1 e^{il_2})(\cos a - l_1 \sin a e^{i(l_2-c-2l_3)})}{\cos a'}.
\end{gather}
We fix $d_3 = 0$, so we remove the matrix on the right by a $U(1)$ rotation.
From the remaining term, we can read off $e^{ic'}$ as the phase of the top right component of the matrix on the left. In particular, we see that
\begin{align}
 \sin a' e^{i(c'-c)} =& (\sin a e^{i(2l_3)} + \cos a l_1 e^{i(l_2-c)})\frac{\cos a - l_1 \sin a e^{i(l_2-c-2l_3)}}{\cos a'}\nonumber\\
 \sin(c'-c) =& \frac{\sin a \cos a}{\sin a' \cos a'}\sin 2l_3 + l_1 \frac{\cos^2 a}{\sin a' \cos a'} \sin(l_2-c) -l_1 \frac{\sin^2 a \sin(l_2-c)}{\sin a' \cos a'}.
\end{align}
If $ a \not{\!\!\sim}\; \pi/2$ and $a \not{\!\!\sim}\; 0$, we can read off, to lowest order in $l_1$ and $l_3$ and given that $c = c'$ at $l_1 = l_3 = 0$ (i.e. there is no gauge transformation),
\begin{align}
 a'-a =& l_1 \cos(l_2-c)\nonumber\\
 c'-c = & 2l_3 + l_1 \sin(l_2-c)\cot a -l_1 \sin(l_2-c) \tan a
\end{align}

The winding around the \mon becomes $\lambda^3 \frac{1}{2}\tr[X_0] \oint \partial_{\tilde{\sigma}} c' d\tilde{\sigma} = \lambda^3 \frac{1}{2}\tr[X_0] \oint \partial_{\tilde{\sigma}} (c'-c) d\tilde{\sigma} + \lambda^3 \pi \nu \tr X_0$, and since $l_1$ and $l_3$ are continuous functions and $c$ only contributes to $(c'-c)$ within a trigonometric function, the winding of $\theta$ around the \mon is unaffected by a continuous gauge transformation. The location where $a = \pi/2$ or $a=0$ may, however, be shifted by a small amount. The \mons, of course, are not $\delta$-function peaks, but extended over a small area.

\subsection{Wilson Loop Dependence of the \Mon Solution.}
Each \mon is constructed from a particular configuration of the $\theta$ field, and each $\theta$ field is given by the eigenvectors of a particular set of nested Wilson Loops. This means that in addition to the gauge field, the particular topological objects we find will depend on the choice of which set of Wilson Loops is being studied. So if we choose one particular set of Wilson Loops, we will uncover one distribution of \monsp, while if we choose a different set of Wilson Loops, we may well uncover a different distribution of the objects. An obvious question is whether this eliminates the usefulness of the \mons as an explanation for confinement. After all, an instanton, Dirac magnetic monopole, or vortex are dependent only on the gauge field, and so can be measured without making any arbitrary choices concerning the $\theta$ field. The global topological charge of the field is an observable which is independent of any choices that we make. If this can be reduced to a set of local topological structures, then the natural thing is to say that these structures should also be independent of any arbitrary decisions we should make when parametrising the gauge field (such as our choice of Wilson Loop).

Obviously this issue would not be relevant if the distribution of \mons was (to a good approximation) independent of our choice of Wilson Loop. However, our numerical tests indicate that this is not the case: we see no significant correlation between the topological field strength when we use one set of Wilson Loops to construct $\theta$ compared to when we use a different set of Wilson Loops.

The Wilson loop represents the propagation of a static quark/anti-quark pair through time. The choice of Wilson Loop is equivalent to choosing where the quark/anti-quark pair under study is located. So the question posed in the first paragraph of this subsection is asking whether the topological objects which explain the confinement of quarks (if quark confinement is explained by topological objects) should manifest themselves independently of the presence of those quarks. There is no obvious reason why that should be the case. Discussion of confinement makes no sense unless there are quarks present to confine; the study of confinement in general can be reduced to the study of why each particular pair of quarks is confined. As the explanation of the confinement of a particular pair of quarks, where the position of the quarks determines the Wilson Loop which then fixes the $\theta$ we require for the decomposition and uniquely determines the topological structures contained within $\theta$, the topological objects under study here work as well as those which are present in the gauge field independently of any choice of Wilson Loop. There is no obvious reason why the key quantity that determines the strength of the confining potential (which is an observable independent of which quarks are being studied), the density (after averaging over the gauge fields) of \mon- anti-\mon pairs, should be dependent on which Wilson Loop is chosen. Thus we feel that the arbitrariness of the \mon fields does not weaken this model's use as an explanation of quark confinement.

\section{String breaking}\label{sec:3}
So far, we have not discussed the effects  of the discontinuity in $\theta$ when the Wilson Loop has degenerate eigenvalues. Although the main focus of our work is elsewhere, some discussion of this topic is required for completeness; however we will do no more than present a simple result leaving the questions arising for the physical consequences of this result unanswered. We will find that if the Wilson Loop has degenerate eigenvalues, then the confining potential will weaken at larger distances. This is similar to the effect known in dynamical simulations (including quark loops) known as string breaking. It is usually expressed in terms of the shielding of the potential by the quark loops. Once the quarks become separated far enough, then the binding energy becomes sufficiently large that it is advantageous to pull a quark anti-quark pair out of the vacuum, and have two shorter `strings' plus an extra pair of quarks rather than one really long `string'. Obviously, our numerical results are in a pure gauge theory, and we do not expect to observe string breaking. If our model is correct, then there must be something suppressing the occurrence of degenerate eigenvalues of the Wilson Loop in pure gauge theory, with this suppression weakened in the dynamical theory. However, this statement remains just a hypothesis, and a dedicated study, well beyond the scope of this work, is required to either confirm or deny it.

We expect that the same model of confinement will be valid in both quenched QCD (Pure Yang-Mills theory) and full QCD. This means that if confinement is explained by some particular mechanism just involving topological objects in the gauge field, then there should be some explanation in full QCD in terms of those objects or the theory surrounding them why the string is broken at large distances. But given that the explanation itself does not depend on the presence of dynamical quarks -- only the weighting given to various gauge field configurations -- then there should be some allowance for string breaking within the model, and this will apply for both quenched and full QCD. String breaking would then be activated by certain gauge fields which are suppressed in quenched QCD but present in full QCD. That our model has an allowance for string breaking even though we only consider the pure Yang Mills theory is thus not surprising.

We wish to evaluate the change of the individual gauge links of the Wilson line as we gradually increase the spatial extent of the Wilson Loop. Again, we consider a $T\times R$ rectangular Wilson Loop, of length $L=2(R+T)$ and $T>R$. We parametrise the position along the Loop using $0\le s < L$ and the spatial extent of the loop with $r = R/2$. We again neglect the effects of the corners of the Loop.

The building block of the Wilson Loop is
\begin{gather}
e^{i\delta\sigma \hat{u}_s^{rj}\lambda^j} = (\theta_s^r)^\dagger e^{i \delta\sigma A_{\mu(s)}(s+\frac{\delta\sigma}{2},r)} \theta_{s+\delta \sigma}^r,
\end{gather}
writing $\theta = e^{i a \phi}$ (neglecting the U(1) phase), we can use equation (\ref{eq:l3}) to write
\begin{align}
\hat{u}_s^{rj}\lambda^j \equiv& \Delta_{1s}^r\nonumber\\
= & A_{\mu(s)}(s+\frac{\delta \sigma}{2},r) +i \partial_\sigma (-ia\phi) + \sum_{n > 0} \frac{1}{n!} \left[(-ia\phi)^n \left(A_{\mu(s)}(s+\frac{\delta \sigma}{2},r) + \frac{i}{n+1}\partial_\sigma (-ia\phi)\right)\right],
\end{align}
where we use the notation
\begin{gather}
[X^n Y] \equiv \underbrace{[X,[X,\ldots [X}_{n\text{ terms}},Y]\ldots]].
\end{gather}
$a$ and $\phi$ were chosen so that $\Delta_{1s}^r$ is Hermitian, traceless and diagonal.

We now consider the evolution of $\hat{u}$ with $r$, using the machinery developed in section \ref{sec:2b} and equation (\ref{eq:thetardiff}). We write
\begin{gather}
\theta^{r+\delta r}_s = \theta^r_s e^{i \beta^{r,s} \Gamma^{r,s}},
\end{gather}
with
\begin{gather}
\Gamma^{r,s} = \left(\begin{array}{c c}0&e^{i\gamma^{r,s}}\\-e^{-i\gamma^{r,s}}&0\end{array}\right).
\end{gather}
We want to calculate
\begin{gather}
e^{i \delta \sigma \hat{u}_s^{r+\delta r,j}\lambda^j} = e^{-i\beta^{r,s}\Gamma^{r,s}} e^{-ia\phi} e^{i \delta \sigma(A_{\mu(s)}(s + \frac{\delta s}{2},r) + \delta r \partial_r A_{\mu(s))}(s + \frac{\delta s}{2},r)} e^{ia\phi} e^{i\beta^{r,s+\delta\sigma} \Gamma^{r,s+\delta \sigma}}
\end{gather}
Using equation (\ref{eq:l3}), we can again write this in the form
\begin{align}
\hat{u}_s^{r+\delta r,j}=& \frac{1}{2} \tr \lambda^j\left(\Delta_2 +i \delta_\sigma (-i\beta \Gamma) + \sum_{n>0}\frac{1}{n!}\left[(-i\beta \Gamma)^n(\Delta_2 + \frac{i}{n+1} \partial_\sigma(-i \beta \Gamma) )\right]\right)+\nonumber\\
&\frac{1}{2}\tr \lambda^j\left(\Delta_1 + \sum_{n > 0}\frac{1}{n!} [(-i\beta\Gamma)^n \Delta_1]\right)\label{eq:Idunnowhattocallthis}.
\\
\Delta_2 = &\delta r \partial_r A_\mu + \sum_{n>0}\frac{1}{n!} \delta r\left[(-ia\phi)^n \partial_r A\right]
\end{align}
We can evaluate the last term in equation (\ref{eq:Idunnowhattocallthis}) using equation (\ref{eq:lemma:su2}), which gives
\begin{gather}
\Delta_1 + \sum_{n > 0}\frac{1}{n!} [(-i\beta\Gamma)^n \Delta_1] = \cos(2\beta) (\Delta_1 - \frac{1}{2}\Gamma \tr \Gamma \Delta_1) + \frac{\sin 2\beta}{2} [\Gamma,\Delta_1] + \frac{1}{2} \Gamma \tr (\Gamma \Delta_1).
\end{gather}
Given that $\Delta_1\propto \lambda^k$ (i.e. one of the diagonal Gell-Mann operators; $\lambda^3$ in SU(2)), we can see that $\tr \Gamma \Delta_1 = 0$ and $\tr (\lambda^j [\Gamma,\Delta_1]) = 0$. This means that from these two terms only the $\frac{1}{2}\tr\; \lambda^j\cos (2\beta) \Delta_1$ part contributes to $u_s^{r+\delta r,j}$. In most situations, $\beta$ is infinitesimal, and this just leads to $\hat{u}_s^{r+\delta r,j} = \hat{u}_s^{r,j} + \ldots$ where the $\ldots$ comes from the remaining terms in equation (\ref{eq:Idunnowhattocallthis}) and is of O($\delta r$). $\hat{u}_s^{r+\delta r}$ is $\hat{u}_s^{r}$ plus a small amount, and therefore its expected absolute value increases (as we have seen in the previous sections) with $r$ once it is clear of the Wilson Loop eigenvalue degeneracy at $r=0$.  However, where there is a degeneracy in the eigenvalues of the Wilson Loop, $\beta$ will be large, and at some point close to the degeneracy $\cos 2 \beta $ will cross zero. While the first terms in equation
(\ref{eq:Idunnowhattocallthis}) will no longer all be infinitesimal, $\hat{u}_s^{r+\delta r,j}$ will lose its direct dependence on $\hat{u}_s^{r,j}$. Effectively, the value of the Wilson Loop is reset and loses all knowledge of what occurred at $R < 2 r$.

Let us suppose that for a particular configuration, these eigenvalue degeneracies in the Wilson Loop occur at distances $R = 0,R_0,R_1,\ldots$. In this case, once we are clear of $R=0$ we may expect $\hat{u}$ (on average) to rise linearly with $R$ until we reach $R_0$, where it will be reset to some value $\tilde{u}$.
Subsequently, it will rise linearly again with $R-R_0$ until $R = R_1$ and it is reset again before rising with $R-R_1$. The location of $R_0$, $R_1$ etc. will vary from configuration to configuration, which means that when we average over configurations, we can consider some mean separation between those Wilson Loops with degenerate eigenvalues $\langle R_0 \rangle$.
 The pattern we may expect is, apart from some non-linear behaviour at small $R$ as the loop escapes the degeneracy at $R=0$ the potential will increase linearly while $R \ll \langle R_0\rangle$. For $R \gg \langle R_0\rangle$ we may expect the potential to be flat,
 since on each configuration it will be somewhere between $\langle\tilde{u}\rangle$ and $\langle\tilde{u}\rangle + \rho \langle R_0 \rangle$, and on the average over configurations this will become a constant as long as the distribution of $R_0,R_1,\ldots$ is not so sharply peaked that they occur in the same place on each configuration.
There will obviously be some intermediate region between these two regimes where we transit from the linear to constant behaviour.

This, of course, precisely as we expect in dynamical QCD: at large distances, the quark loops screen the static quark potential leading to a breaking of the string. This has been observed numerous times in lattice QCD, for example in~\cite{Bali:2005fu}, although in our lattice study which excludes quark loops we do not see it.

\section{Numerical results}\label{sec:4}
We generated $16^332$ and $20^340$ quenched (i.e. pure Yang Mills) lattice QCD (SU(3)) configurations with a Tadpole Improved Luscher-Weisz gauge action~\cite{TILW,TILW2,TILW3,TILW4,TILW5} using a Hybrid Monte Carlo routine~\cite{HMC} (see table \ref{tab:1}).
The lattice spacing was measured using the string tension $\rho \sim (420 \text{MeV})^2$. We fixed to the Landau gauge and applied ten steps of improved stout smearing~\cite{Morningstar:2003gk,Moran:2008ra} with parameters $\rho = 0.015$ and $\epsilon = 0$.
$\theta$ and $\hat{U}$ were extracted from the gauge field numerically by solving equations (\ref{eq:evth}), (\ref{eq:defeq1}) and (\ref{eq:defeq2}), as outlined in  \ref{app:C}.
We constructed $\theta$ from planar Wilson Loops in the $xt$ plane. $\hat{F}_{xt}\equiv \hat{E}_x$ is therefore the component of the field strength tensor responsible for confinement. All our error analysis was performed using the bootstrap method (except for the parameters obtained from fits which estimated the errors from the $\chi^2$ distribution).
\begin{table}
{
\begin{center}
\begin{tabular}{|l l l l l|}
\hline
Lattice size & $L$ (fm)&$\beta$&$a$ (fm)& $\#$\\
\hline
$16^332$& 2.30& 8.0&0.144(1) &91
\\
$16^332$&1.84& 8.3 &0.115(1) &91
\\
$16^332$&1.58 & 8.52 &0.099(1) &82
\\
$20^340$&2.30 & 8.3&0.115(1) &20\\
\hline
\end{tabular}
\end{center}
}
\caption{Parameters for our simulations: the lattice size, measured in units of the lattice spacing, the spatial extent of the lattice, $L$, in physical units, the inverse gauge coupling $\beta$, the lattice spacing $a$, and the number of configurations in each ensemble $\#$. }\label{tab:1}
\end{table}

The main difficulty we have faced while investigating this proposed mechanism of confinement is the non-gauge invariance of $\theta$. We need to work in a gauge where the gauge field is continuous, while on the lattice there is no such gauge.
 However, to completely establish the proposed mechanism of confinement, we would need to calculate the variables $a$ and $c$ used to parametrise $\theta$.  Extracting the components $a$ and $c$ from $\theta$ is straightforward, and we presented some results for the winding of $c$ around the peaks in~\cite{Cundy:2012ee}, with some encouraging results. They are, however, gauge dependent, and there is therefore no unambiguous definition of these quantities. While the winding number is protected from continuous gauge transformations in the continuum, on the lattice it is difficult even to ensure that you are in the right gauge.  Therefore for this paper, we have concentrated on gauge invariant quantities such as the static potential and the restricted field strength.

Our aim is to show that the string tension is reproduced by that of the restricted field, and the restricted field dominated by the contribution from $\theta$, or the \mon term in the field strength (constructed from the potential $\tr \lambda^j \theta^\dagger \partial_\mu \theta$). We should show that the restricted field strength $\hat{F}_{xt}$ is dominated by points  (which would appear on the lattice as peaks a single lattice spacing across), while the other components of $\hat{F}$ contain structures extended in the directions specified in table \ref{tab:lotsofstrings}.

Measuring the string tension associated with the restricted potential is straight-forward, since it is gauge invariant, so it can be extracted from the Wilson Loop using standard methods. Equally, the restricted field strength can be measured using the plaquette definition
\begin{gather}
e^{i\hat{F}_{xt}^j(x+\frac{1}{2}a,t+\frac{1}{2}a) \lambda^j} \sim \theta^\dagger_x \hat{U}_x(x,t) \hat{U}_t(x+a,t) \hat{U}^\dagger_x(x,t+a)\hat{U}^\dagger_t(x,t)  \theta_x,
\end{gather}
and as this is gauge invariant it can be measured without any ambiguities or difficulties. Extracting $\hat{F}^j_{\mu\nu}$ from $e^{i \hat{F}^j \lambda^j}$, however, can only be done modulo $2\pi$ for $F^3$ and $4\pi/\sqrt{3}$ for $F^8$. For smooth gauge fields, where $\hat{F}$ is small (as is the case in most lattice applications), this is not a concern. We, however, as interested in those cases where $\hat{F}$ is not small, and we shall find that in some applications $|\hat{F}^j|$ may be larger than $\pi$. For example, if $\hat{F}^3 \sim 2\pi$ it is impossible to distinguish this large value from a fluctuation around zero. This is not of significance for the physics, since these shifts in $\hat{F}$ will not affect the value of the Wilson Loop, but it may be of importance in some of the following as we attempt to identify peaks in the field strength.

However, constructing the \mon contribution to the restricted field strength is more challenging because there is no clean gauge-invariant way to observe this on the lattice. In the continuum, the observable is $\lim_{a\rightarrow 0}\frac{1}{a}\tr \lambda^j (\theta^\dagger_x \theta_{x + a\hat{\mu}} - 1) \sim \tr \lambda^j (\theta^\dagger_{x + \frac{1}{2}a\hat{\mu}}\partial_\mu \theta_{x + \frac{1}{2}a\hat{\mu}})$ for lattice spacing $a$.
Since this is not gauge invariant, a gauge transformation which would be discontinuous in the continuum would lead to additional discontinuities appearing in the observable, or the removal of the discontinuities we want to observe.
 Fixing the gauge would not help us, because we could well be fixing to a discontinuous gauge.
 We solved this problem by introducing a gauge invariant quantity which measures the \mon contribution to the potential.

We introduce a field $\tilde{U}$, which is the gauge field $U$ subjected to a large number of coarse stout smearing
sweeps.
 A small amount of smearing is useful to remove discontinuities on the order of the lattice spacing; a large amount is usually seen as dangerous as it removes the physical features of the gauge field; however we can use this to our advantage.
Each Stout smearing sweep replaces $U_{x,\mu} \rightarrow U'_{x,\mu} = e^{iQ_x} U_{x,\mu}$ where $Q$ is a Hermitian operator constructed from closed loops of gauge links starting and finishing at $x$ containing the gauge link $U_{x,\mu}$ (the original smearing algorithm~\cite{Morningstar:2003gk} used plaquettes; we also used $2\times 1$ rectangles following~\cite{Moran:2008ra}).
 After a gauge transformation $U_{x,\mu} \rightarrow \Lambda_x U_{x,\mu} \Lambda_{x+\hat{\mu}}$, the smeared link transforms in the same way, $U'_{x,\mu} \rightarrow \Lambda_x U'_{x,\mu} \Lambda_{x+\hat{\mu}}$.
 The second effect of stout smearing is to smooth the gauge field, and reduce the value of the plaquette to zero and thus the potential to a constant, and in practice the potential becomes zero. Eventually, after repeating the smearing a very large number of times, the over-smeared link $\tilde{U}$ becomes, in effect, a gauge transformation of something close to the identity operator.\footnote{It has been shown that a large number of infinitesimal smearing sweeps -- not quite what we are doing here, but similar -- is equivalent to a flow in the configuration space towards an ensemble generated with infinite gauge coupling~\cite{Luscher:2009eq}.}
 Thus $\theta^\dagger_x \tilde{U}_{x,\mu} \theta_{x+\hat{\mu}}$ is both gauge invariant and in a gauge where $U$ is continuous would measure $\theta^\dagger \partial_\mu \theta$, the observable we need.
We therefore use this observable to measure the \mon portion of the field strength. In effect, we are replicating the calculation for the restricted field, but with the gauge field replaced by something close to zero, just leaving the \mon contribution to the restricted field.
 To extract the Abelian component of $\theta^\dagger_x \tilde{U}_{x,\mu} \theta_{x+\hat{\mu}}$ we perform another Abelian decomposition and extract the restricted field $\hat{ \tilde{U}}_{x,\mu}$
which satisfies $[\hat{\tilde{U}}_{x,\mu} ,\lambda^j] = 0$
and $\tr(\lambda^j(\tilde{X} -\tilde{X}^\dagger)) = 0$ with $\tilde{X} = \theta^\dagger_x \tilde{U}_{x,\mu} \theta_{x+\hat{\mu}} (\hat{ \tilde{U}}_{x,\mu})^{-1}$.
 We can then compare the field strength from this restricted field, which represents the \mon portion of the restricted field strength, with that of the restricted field $\hat{U}$.
Our expectation is that the observables calculated from the \mon field and the restricted field should be similar: the string tensions should be in agreement, and the field strength should contain a similar pattern of peaks. We denote $H^j_{\mu\nu}$ as the field strength tensor constructed using $\hat{\tilde{U}}$.

\subsection{String tension -- Fixed Wilson Loop}

Firstly, in figure \ref{fig:5} and table \ref{tab:2}, we extract the string tension, $\rho$, for the original gauge field $U$, the restricted field $\hat{U}$ and the \mon component of the restricted field $M$. We have calculated this in two ways: firstly by calculating the expectation value of the $R\times T$ Wilson Loop for one of the fields, and fitting its logarithm firstly to $VT + a + b/T$ and then $V$ to $\rho R +c+d/R$ for additional fit coefficients $a,b,c$ and $d$; and secondly by fitting the whole data set according to $\rho RT + a R + b T + c +d/R +e/T +f/(RT) + g R/T + h T/R$. The results from these two approaches were in good agreement. We tabulate the results from the combined $RT$ fit, while plot the results of the single fits.
The errors on the expectation value were calculated using the bootstrap method; we calculated the errors on the fitted parameters by measuring the surfaces of constant $\chi^2$.  In principle we should recalculate a new $\theta$ field for each Wilson Loop; however the numerical cost to do so was prohibitive for an initial study, and we firstly show results obtained by using the same $\theta$ field calculated from one set of Wilson Loops in the $xt$ plane for our entire simulation~\cite{Cundy:2013xsa}. We discuss this simplification a little more in the next subsection. We average over all Wilson Loops for the restricted field, and not just those with the correct $\theta$ field. Were we to recalculate a different $\theta$ for each Wilson Loop, the string tension extracted from the $U$ and $\hat{U}$ fields would be identical. Because we have not done this, we expect the two quantities to differ by a small amount.

We calculate $\tilde{U}$ after 100, 300, 500, 600, 800 and 1000 sweeps of stout smearing with parameters $\epsilon = 0, \rho = 0.015$ following the notation of~\cite{Moran:2008ra}, and extract the \mon portion of string tension by constructing Wilson Loops from the Abelian decomposition of $\theta^\dagger\tilde{U} \theta$. We also show the string tension for $\tilde{U}$, and can confirm that it is much smaller than that of the original gauge field and decreases as we increase the level of smearing. The observable we are using to measure the \mon portion of the string tension is unaffected by the smearing, suggesting that we have indeed measured the contribution from $\theta \partial_\mu \theta^\dagger$ rather than any remnant of the gauge field remaining after the smearing. We see a good agreement between the \mon and restricted field string tensions, within 10\% on the smaller lattices (although the difference is larger on the larger lattice where we have less statistics), suggesting that it is indeed structures
within the $\theta$ field that lead to confinement. There is no obvious change in this pattern as we change the lattice spacing.

\begin{figure}
\begin{center}
{\tiny
\input{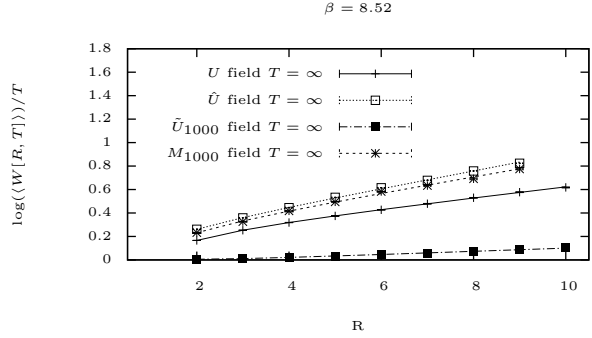}
}
\end{center}
\caption{The string tension extrapolated to infinite time for the original gauge field $U$, the restricted gauge field $\hat{U}$, the over-smeared field $\tilde{U}$ and the \mon field $M$ for the $\beta = 8.52$ ensemble}\label{fig:5}
\end{figure}

\begin{table}
\begin{center}
{\begin{tabular}{|l|l l l l|}
\hline
$\beta$&8.0&8.3&8.52&8.3L\\
\hline
$U$&0.094(2)&0.0590(8)&0.0442(6)&0.057(2)
\\
$\hat{U}$&0.116(4)&0.095(2)&0.077(1)&0.099(1)
\\
$\tilde{U}_{100}$& 0.0828(4)&0.0545(3)&0.0433(7)&0.0594(8)
\\
$M_{100}$ &0.129(4)  & 0.090(3)& 0.075(3)&0.100(1)
\\
$\tilde{U}_{300}$& 0.0460(3)&0.0301(3)&0.0239(4)&0.0316(2)
\\
$M_{300}$ &0.122(4)  & 0.086(3)& 0.072(3)&0.102(1)
\\
$\tilde{U}_{500}$&0.0316(2) &0.0216(2)&0.0174(3)&0.0218(1)
\\
$M_{500}$ & 0.124(5) & 0.088(3)& 0.072(3)&0.104(1)
\\
$\tilde{U}_{600}$& 0.0273(2)&0.0185(1)&0.0149(2)&0.0179(1)
\\
$M_{600}$ &0.103(10)  &0.087(5) & 0.076(3)&0.099(1)
\\
$\tilde{U}_{800}$&0.0213(1) &0.0147(1) &0.0124(2) &0.0144(1)
\\
$M_{800}$ &0.104(8)&0.087(5) &0.076(3)&0.099(1)
\\
$\tilde{U}_{1000}$&0.0174(1) &0.0122(1) &0.0106(2)&0.0121(1)
\\
$M_{1000}$ &0.105(9) &0.087(5) &0.077(3)&0.098(2)
\\
\hline
\end{tabular}
}
\end{center}
\caption{The string tension extrapolated to infinite time  across all our ensembles. 8.3L refers to the $20^3\times40$ ensemble.}\label{tab:2}
\end{table}
\subsection{String Tension -- Varied Wilson Loop.}
We generated the results presented in the previous section as a quick and dirty initial study, largely to demonstrate to ourselves and the wider community that this method was worth pursuing, and we first published these results in~\cite{Cundy:2013xsa} (these are the 10 smearing sweeps results of table \ref{tab:2b}). Since the initial drafts of this paper, we have commenced a more accurate study removing the simplification of the previous section, this time recomputing the Abelian decomposition for each Wilson Loop under consideration. This obviously maintains the identity between the confining potential of the full and restricted gauge fields. Our results so far are only tentative, and far from being finalised. We present early results here mainly to update our initial presentation of this study in~\cite{Cundy:2014nna}. In those proceedings, we found that although the restricted and full QCD potentials were in perfect agreement, the topological contribution to the string tension seemed to be relatively small.

If this result is genuine, there are two possible explanations: firstly, that the Maxwell term in the restricted potential, contrary to expectations, also has a large contribution to confinement. The second possibility is that our means of extracting the topological part of the string tension is flawed, and what we were measuring was not the topological part of the string tension but something else. Since our initial publication, however, we have expanded our study, and have noted another issue. We apply a small number of smearing sweeps to the gauge field initially, before applying the Abelian decomposition,  to remove various unphysical fluctuations of the order of the inverse lattice spacing. This is necessary to allow us to obtain a clear result. How many smearing sweeps should be used is difficult to assess, and for a full presentation of the results we should measure the string tension at different levels of smearing to gain some idea of the systematic error caused by this arbitrary choice. As we did so, we observed that the measured topological contribution to the smearing depended strongly on the number of initial smearing sweeps. Furthermore (unlike in the results of the previous section), we also found that at a larger number of initial smearing sweeps, there was a strong lattice spacing dependence on the ratio between the restricted and topological potentials: as the lattice spacing grew smaller, the restricted potential became increasingly dominated by the topological part. This means that our full analysis is considerably more difficult to do accurately than we had first envisaged, and, given our limited computing resources, will take longer to complete.

Some preliminary results (which are not finalised, and could change when we make our final publication) are shown table \ref{tab:2b}. As we vary the number of smearing sweeps, the string tension of the original gauge field first increases as the unphysical dislocations are removed, and then slowly decreases as the smearing gradually distorts the gauge field. This obviously begs the question of which value we should use, and our intention has been to use the number of smearing sweeps where the string tension is at its maximum (at the present time, we have kept the smearing parameters fixed at $\epsilon = 0$, $\rho = 0.015$ for this study). So far, we have only conducted this analysis for the $16^3\times 32$ configurations. Table \ref{tab:2b} lists the $\beta$ value of the configuration, the number of smearing sweeps where the maximum value of the string tension was reached (out of those which we have sampled so far), and the measured string tension for the QCD (and restricted) gauge fields, and the proportion of that string tension which we have measured as coming from the topological part of the restricted potential (after 2500 smearing sweeps).
\begin{table}
\begin{center}
 \begin{tabular}{|c|ccc|}
 \hline
  $\beta$&8.0&8.3&8.52\\
  \hline
  Smearing Sweeps&10&10&10
  \\
  $\rho$&0.0976(36)&0.0603(19)&0.0416(17)
  \\
  Proportion&0.254(14)&0.277(14)&0.334(19)\\
  \hline
   Smearing Sweeps&16&16&16
  \\
  $\rho$&0.0954(40)&0.0740(17)&0.0489(14)
  \\
  Proportion&0.247(17)&0.332(15)&0.462(27)\\
  \hline
   Smearing Sweeps&22&22&22
  \\
  $\rho$&0.0944(26)&0.0.0731(14)&0.0501(13)
  \\
  Proportion&0.334(15)&0.400(14)&0.526(27)\\
  \hline
  \end{tabular}
 \end{center}
\caption{The number of smearing sweeps where the string tension of the full QCD field was at its maximum, the string tension $\rho$ (calculated from the $U$ field) and the ratio of the topological part of the string tension to the full QCD string tension. Only the statistical errors are recorded. We list results at those numbers of smearing sweeps where, out of the measurements we have taken so far, the string tension for the $\beta 8.0$, $\beta 8.3$ and $\beta 8.52$ ensembles respectively are maximised.}\label{tab:2b}
\end{table}

It can be seen that firstly the maximum string tensions for the $\beta 8.3$ and $\beta 8.52$ configurations are higher than the results given above (they were calculated following a different amount of initial smearing). Secondly, the proportion of the string tension attributable to the topological part increases as we decrease the lattice spacing. This pattern is not merely a matter of the results being taken after a different number of smearing steps, but is apparent at each value of the initial smearing, although it is more pronounced as we increase the smearing. Our analysis is less developed on the larger lattice, but the results we have suggest that there is no significant volume dependence.

This strong dependence on the lattice spacing is awkward because it necessitates an extrapolation to the continuum limit. This is challenging with the data we have gathered so far, firstly because we do not know of a good formula to guide the extrapolation, and secondly because we are still on a fairly coarse lattice. We need to finalise our results on these small lattices, and then expand our data set so that we have access to finer lattice spacings before it is useful to attempt to perform such an extrapolation; we will have to leave the completion of this study to a future publication.

However, it is not unreasonable to suggest on the basis of this data that the string tension will be dominated by the topological part of the restricted potential (as we have measured it) in the continuum limit. This still needs to be confirmed by a complete and more careful analysis, which we are currently undertaking.

\subsection{Presence of peaks in $\hat{F}_{xt}$}

The next question is whether the $\hat{E}_x$ component of the gauge invariant restricted field strength is dominated by point-like objects as expected in the theoretical model, and we plot the distribution of $\hat{F}^3_{xt}$ in figure \ref{fig:4} on two neighbouring slices of the lattice, using a contour plot to display the data. The corresponding plot for $\hat{F}^8_{xt}$ has a similar set of structures. It can be seen that $\hat{F}^3_{xt}$ is indeed dominated by these objects a single lattice spacing across (or, on occasion, two lattice spacings -- see also figure \ref{fig:ClusterSize}). There is no relationship between the peaks on neighbouring slices of the lattice, suggesting that these are indeed points rather than strings or surfaces. Are these peaks from the \mon part of the restricted potential? The background shading of figure \ref{fig:4} gives a similar plot for the field strength extracted from $\hat{\tilde{U}}$, and the is a strong correspondence between the location of the peaks in the \mon field and the restricted
field, although
on occasion one of the peaks may be shifted by a lattice spacing (which is the limit of the resolution of the lattice, so we can expect errors of up to a half lattice spacing on each of the data sets), and a few objects in the \mon field are not visible in the restricted field. Nonetheless, the similarity between the two sets of data indicates that these peaks are caused by the magnetic portion of the ensemble as predicted by the theory. Although figure \ref{fig:4} shows data from just one slice of one configuration for one ensemble, the pattern of peaks is similar across all our ensembles.


\begin{figure}
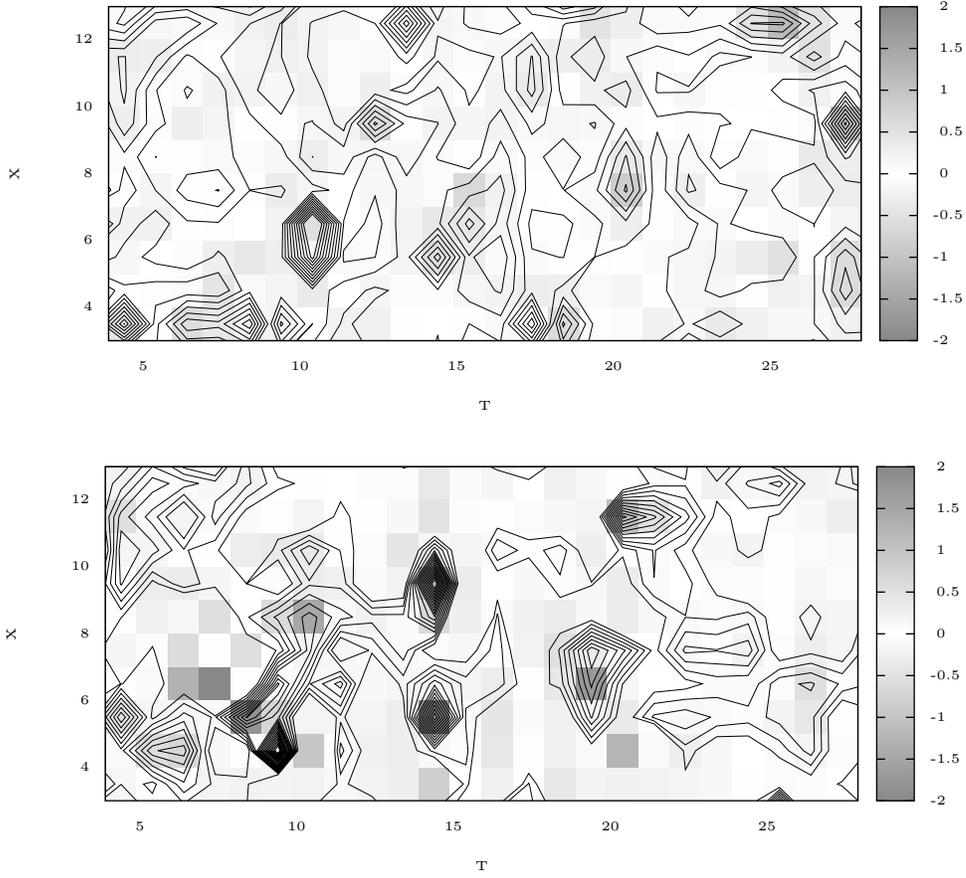

\begin{center}
\begin{tabular}{c}
{\tiny
\input{w3_10.ltx}
}\\
{\tiny
\input{w3_11.ltx}
}
\end{tabular}
\end{center}

\caption{A comparison between the peaks in $\hat{F}^{3}_{xt}$ (contours) against the \mon field strength extracted $\hat{ \tilde{U}}$ (shaded background). In this presentation the negative and positive peaks cannot be distinguished.  The top plot is taken for slice of the lattice $Y=0$, $Z=10$ while the bottom plot uses the slice at $Y=0$, $Z=11$ on the same $\beta = 8.52$ configuration. Due to the limited resolution of the lattice, the extrapolated contour lines and the shading have an error of up to one lattice spacing.}\label{fig:4}
\end{figure}

\begin{figure}
\begin{center}
\begin{tabular}{c}
{
\input{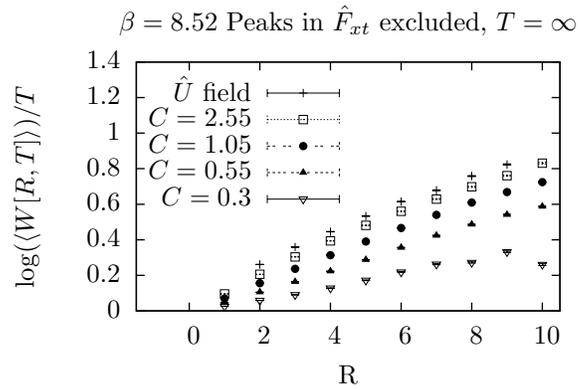}
}
\end{tabular}
\end{center}
\caption{The $\hat{U}$ string tension, $\rho$, excluding Wilson loops containing peaks of height $|\hat{F}_{xt}| > \mathcal{C}$ from the average, on the $\beta 8.52$ lattice. }\label{fig:1}
\end{figure}
\begin{table}
\begin{center}
{
\begin {tabular}{|l|lllll|}
\hline
$C$&$2.55$&$2.30$&$2.05$&$1.80$&$1.55$\\
\hline
$\beta =8.0$&$0.1178(9)$&$0.1201(9)$&$0.1213(9)$&$0.1225(8)$&$0.1199(8)$
\\
$\beta =8.3$&$0.0943(8)$&$0.0955(7)$&$0.0961(8)$&$0.0971(8)$&$0.0957(8)$
\\
$\beta =8.52$&$0.0796(7)$&$0.0809(7)$&$0.0820(6)$&$0.0821(6)$&$0.0820(6)$
\\
$\beta =8.3L$&$0.0975(8)$&$0.0965(8)$&$0.0954(8)$&$0.0964(8)$&$0.0963(8)$
\\
\hline\hline
$C$&$1.30$&$1.05$&$0.80$&$0.55$&$0.30$\\
\hline
$\beta =8.0$&$0.1144(8)$&$0.1049(7)$&$0.1002(7)$&$0.0925(7)$&$0.04998(7)$
\\
$\beta =8.3$&$0.0899(8)$&$0.0866(7)$&$0.0809(6)$&$0.0741(6)$&$0.0382(7)$
\\
$\beta =8.52$&$0.0791(6)$&$0.0747(6)$&$0.0697(5)$&$0.0665(5)$&$0.0472(6)$
\\
$\beta =8.3L$&$0.0902(7)$&$0.0843(7)$&$0.0770(7)$&$0.0706(5)$&$-$
\\
\hline

\end{tabular}
}
\end{center}
\caption{The $\hat{U}$ string tension, $\rho$ , excluding Wilson loops containing peaks of height $|F_{xt}| > \mathcal{C}$ from the average. The $\beta 8.3L$ data refers to the largest lattice; there was insufficient data to get a reliable estimate of the string tension at the lowest cut-off on this lattice.}\label{tab:1b}
\end{table}

In figure \ref{fig:1} and table \ref{tab:1b}, we investigate whether these peaks are responsible for the string tension, i.e. if excluding the peaks would reduce or eliminate the confining potential.  Usually, we measure the expectation value of the Wilson Loop by averaging over every planar loop in the configuration. However, it is a straight-forward exercise to restrict the average to loops in the $xt$ plane and either only include or exclude those Wilson Loops which contain one of the peaks from the average. We show in figure \ref{fig:1} the result of excluding the Wilson Loops which contain peaks. We introduce a cut-off $\mathcal{C}$, and do not include any Wilson Loops where $|\hat{F}| > \mathcal{C}$ for one of the lattice sites contained within the loop. If the peaks are responsible for
confinement, we would expect to see the string tension disappear as we exclude more of the peaks. If they are merely incidental, and confinement is dominated by some other mechanism than that proposed here, we might expect to see little difference in the string tension as we change the cut-off. The picture is, however, consistent with the idea that these peaks are responsible for confinement.

\subsection{Electromagnetic field tensor}\label{sec:6.3}
We now look at the full strength tensor to examine how closely it resembles our expectations from section \ref{sec:4.3}, although for conciseness we concentrate only on $\hat{H}^8_{\mu\nu}$ (the results for $\hat{H}^3_{\mu\nu}$ and $\hat{F}^j_{\mu\nu}$ are indistinguishable). First of all, we plot histograms of each component of the field strength in figure \ref{fig:hist}, showing what proportion of lattice sites had a particular value of the electric or magnetic field. It can be seen that the distribution for $\hat{E}_y,\hat{B}_y,\hat{E}_z$ and $\hat{B}_z$ are all similar, while the distribution for $\hat{E}_x$ and $\hat{B}_x$ differ. Most of the plots show similar features: a large peak around $\hat{H}^8_{\mu\nu}=0$, meaning that most of the points have no field strength beyond a small fluctuation around zero, and a plateau for larger field strength up to the maximum value of $\hat{H}^8_{\mu\nu}$, meaning that on a minority of lattice sites there is a large contribution to the field strength, where the number of sites with that particular value of the field strength  seems to be independent of the field strength. This is as we expected. We also
see that the
proportion of the points with small values of the field strength increases as we
decrease the lattice spacing, which suggests that these objects have a fixed density per physical volume but the size of the objects decreases as the lattice spacing tends to zero, which is consistent with the idea that these are thin points or string-like objects. $\hat{B}_x$ and $\hat{E}_x$ are different; there are twice as many lattice sites with a large value of $\hat{B}_x$ compared to the fields in the $y$ and $z$ direction, and very few lattice sites with a large value of $\hat{E}_x$. There also seems to be little volume dependence.

The picture for all the fields is consistent with having some (perhaps Gaussian distributed) fluctuations around zero, and then some additional objects at large field strength. For $\hat{E}_x$, the number of the objects diminishes as the field strength increases, while for the remaining fields the proportion of lattice sites with a particular field strength seems to be roughly independent of the field strength, at least for those field strengths we can measure. This means that there may be peaks in the field with a field strength $\sim 2\pi$ for $\hat{F}^3$ or $\sim 4\pi\sqrt{3}$ for $\hat{F}^8$ which are impossible for us to distinguish from zero. It should therefore be borne in mind that not every object in these fields will be visible to us.
\begin{figure}
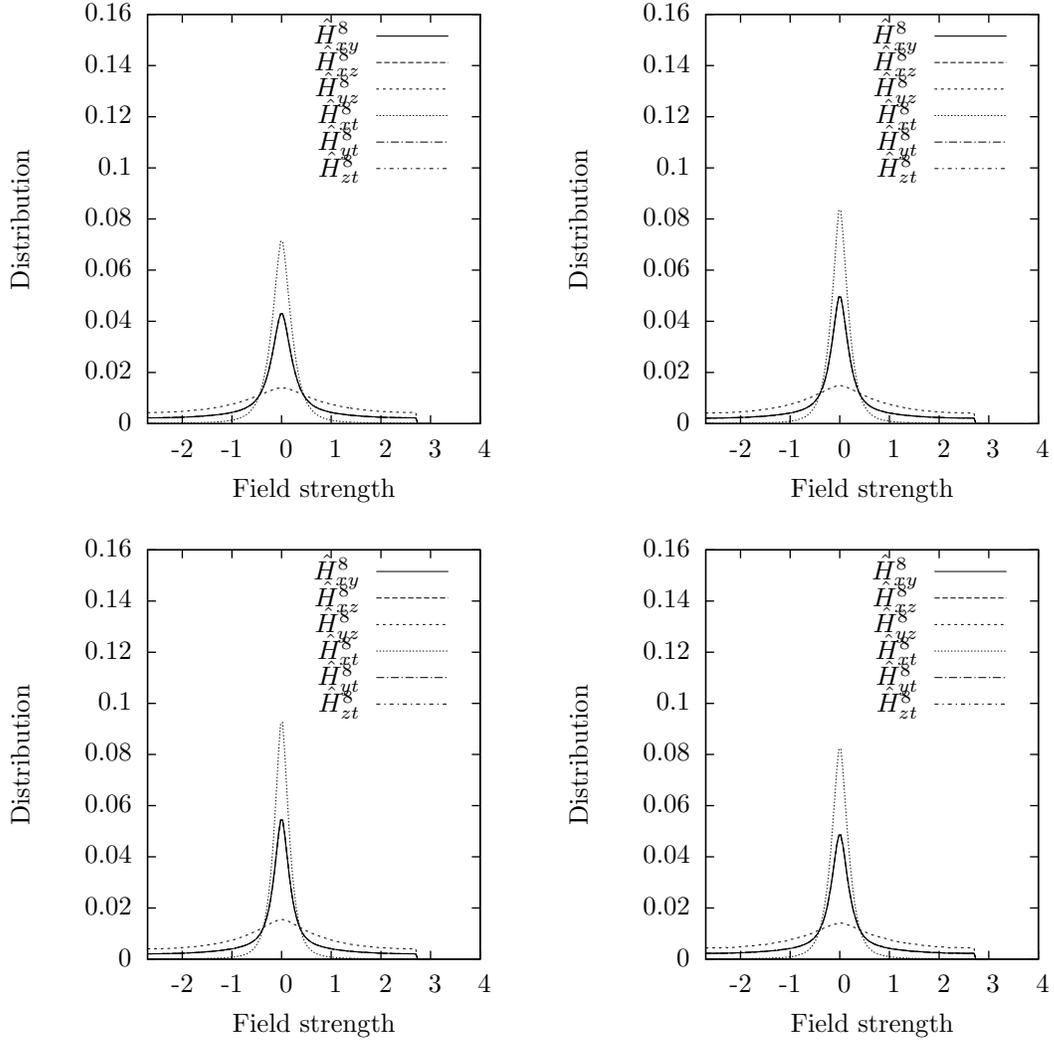

\begin{center}
\begin{tabular}{cc}
\input{hist_H8_80.ltx}&
\input{hist_H8_83.ltx}\\
\input{hist_H8_852.ltx} &
\input{hist_H8_83_s20.ltx}
\end{tabular}

\end{center}
\caption{Histogram of the various components of the magnetic field strength for $\beta = 8.0$ (top left), $\beta = 8.3$ (top right) and $\beta = 8.52$ (bottom left) on the $16^332$ lattices, and the $\beta = 8.3$ $20^340$ lattice (bottom right).}\label{fig:hist}
\end{figure}

This picture can be studied in more detail by plotting slices of the electromagnetic fields: figure \ref{fig:fmunu_xt} shows this in the $x t$ plane; figure \ref{fig:fmunu_yt} in the $yt$ plane, \ref{fig:fmunu_zt} in the $zt$ plane.
\begin{figure}
\begin{center}
\begin{tabular}{cc}
\input{E_x_b8-52_70_xt.ltx}&
\input{B_x_b8-52_70_xt.ltx}\\
\input{E_y_b8-52_70_xt.ltx}&
\input{B_y_b8-52_70_xt.ltx}\\
\input{E_z_b8-52_70_xt.ltx}&
\input{B_z_b8-52_70_xt.ltx}
\end{tabular}

\end{center}
\caption{The electromagnetic fields on the slice with $y = 5$, $z = 5$ on one $\beta = 8.52$ configuration. The dotted contours indicate a positive field strength, the solid contours indicate a negative field strength. The top plots show the fields in the $x$ direction, the middle plots the $y$ direction, and the bottom plots the $z$ direction, while the left column shows the electric field and the right column the magnetic field.}\label{fig:fmunu_xt}
\end{figure}

\begin{figure}
\begin{center}
\begin{tabular}{cc}
\input{E_x_b8-52_70_yt.ltx}&
\input{B_x_b8-52_70_yt.ltx}\\
\input{E_y_b8-52_70_yt.ltx}&
\input{B_y_b8-52_70_yt.ltx}\\
\input{E_z_b8-52_70_yt.ltx}&
\input{B_z_b8-52_70_yt.ltx}
\end{tabular}

\end{center}
\caption{The electromagnetic fields on the slice with $x = 5$, $z = 5$ on one $\beta = 8.52$ configuration. The dotted contours indicate a positive field strength, the solid contours indicate a negative field strength.  The top plots show the fields in the $x$ direction, the middle plots the $y$ direction, and the bottom plots the $z$ direction, while the left column shows the electric field and the right column the magnetic field.}\label{fig:fmunu_yt}
\end{figure}

\begin{figure}
\begin{center}
\begin{tabular}{cc}
\input{E_x_b8-52_70_zt.ltx}&
\input{B_x_b8-52_70_zt.ltx}\\
\input{E_y_b8-52_70_zt.ltx}&
\input{B_y_b8-52_70_zt.ltx}\\
\input{E_z_b8-52_70_zt.ltx}&
\input{B_z_b8-52_70_zt.ltx}
\end{tabular}
\end{center}
\caption{The electromagnetic fields on the slice with $x = 5$, $y = 5$ on one $\beta = 8.52$ configuration. The dotted contours indicate a positive field strength, the solid contours indicate a negative field strength.  The top plots show the fields in the $x$ direction, the middle plots the $y$ direction, and the bottom plots the $z$ direction, while the left column shows the electric field and the right column the magnetic field.}\label{fig:fmunu_zt}
\end{figure}
These objects are not just isolated magnetic monopoles: there are peaks in the electric field, and the magnetic fields are not spherically symmetric. They may, of course, be some magnetic monopoles and some other structures, or possibly a signature of a monopole/anti-monopole condensate. The results are consistent with our expectations given in table \ref{tab:lotsofstrings}, and the patterns seen on this configuration and slice of the lattice are duplicated across all our ensembles. The electric field is clearest (although not completely unambiguous): $\hat{E}_x$ is a point, while $\hat{E}_y$ and $\hat{E}_z$ form strings parallel to the $x$-axis. The magnetic field is a little less clear, and expected string like behaviour of $\hat{B}_y$ and $\hat{B}_z$ is less obvious, although we have seen clear examples of t-strings. $\hat{B}_x$ is, as expected, extended along both the $x$ and $t$ axes, though not the $z$ or $y$ axis. In particular, other orientations of the fields with strings in directions perpendicular to the $xt$ plane are inconsistent with this data.

\subsection{Cluster analysis}

To investigate the restricted electromagnetic fields more thoroughly, we employ a simple strategy to determine the size and shape of these structures. We first of all find clusters, connected regions of sign coherent high electric or magnetic field strength. We perform the analysis for each component of the field strength tensor separately. While we have performed this analysis for $\hat{F}^3_{\mu\nu}$, $\hat{F}^8_{\mu\nu}$, $\hat{H}^3_{\mu\nu}$ and $\hat{H}^8_{\mu\nu}$ individually, there is no significant difference in the results, so here we just show the results for $\hat{H}^8_{\mu\nu}$.

We define a cluster as all connected sites with a value of $|\hat{F}_{\mu\nu}| > 1$ and identical sign of $\hat{F}_{\mu\nu}$ (i.e. except for clusters containing just one lattice site, each site in the cluster is the nearest neighbour of at least one more site in the cluster and all its nearest neighbours which satisfy the bound $|\hat{F}_{\mu\nu}| >1$ are within the same cluster and each cluster is internally connected so it is possible to reach any site within the cluster from any other by crossing to a neighbouring site while remaining within the same cluster, while different clusters are not internally connected with each other). We have also studied (but do not show here) results with numerous values of this cut-off from 0.8 to 2.0, but we do not see any important difference as a result of varying this cut-off.

For all the fields except $\hat{E}_x$, the vast majority of the lattice sites where the restricted field satisfied our condition were contained within a large cluster which spanned the entire lattice. This is in certain respects similar to the structures found in the full field strength tensor~\cite{Horvath:2003yj,Horvath:2005rv}. The theoretical picture we were testing described a number of strings in either the $x$ or $t$ direction. This large structure may be formed if the density of these strings is sufficiently large that every string has at least one neighbour on an adjacent lattice site in the $y$ or $z$ directions. It may, however, relate to a different picture of the vacuum. There were a few clusters not attached to this large structure; and most of our analysis is dominated by results from these smaller structures, as our averaging with respect to the number of structures weights against the smaller number of structures with a large number of lattice sites. We also present a few results which specifically target these larger structures. Here we see the expected string behaviour, with a large number of objects at just a single point. We emphasise that $\hat{E}_x$ (or $\hat{F}_{xt}$) differed, in
that the field
strength did not contain this large structure, but only single lattice site structures. To attempt to analyse the shape of this large cluster, we have also performed a similar analysis, but where the clusters are constructed from connected sites with a two dimensional slice of the lattice.

When we show our results, we will label those plots where we have included all clusters as \textit{unfiltered}, those which only average over clusters with at least four lattice sites as \textit{filtered}, and those which average over clusers with at least two hundred lattice sites as \textit{Large Clusters}.

The theoretical expectation was that if $\hat{E}_x$ is dominated only by points of large field, then the remaining fields should be dominated by objects which are either points or strings in the $x$ or $t$ direction. The only expected difference in the results for the $\hat{E}_y$, $\hat{E}_z$, $\hat{B}_z$ and $\hat{B}_y$ fields is the orientation of the string. A point will show itself as a cluster containing just a few lattice sites (ideally only one lattice site, but due to lattice artefacts it may be smeared across two or three sites). A string would display itself as a line of flux, i.e. each site within a cluster would have only two nearest neighbours within a cluster if it is in the centre of the string and one if it is at the end of the string, and it will only be extended in one direction. In practice, the situation will be messier than this. Neighbouring strings may be counted as the same cluster, and there may be ambiguities caused by the poor resolution of the lattice leading to a thickness of the string
being more than one lattice spacing in places, two different strings may overlap (and the expectation is only accurate if the functions $G$ and $R'$ have no strong angular dependence). Nonetheless, we would hope to see some signal of these strings when considering an average size and shape of the cluster.

We first plot the distribution of the number of lattice sites in each cluster in figure \ref{fig:ClusterSize}. The picture is similar across all our ensembles, so we have only plotted the data for the $\beta = 8.52$ ensemble. It can be seen that all the components of the field strength contain a large number of points with only one or two sites contributing to the cluster. However, the $\hat{E}_x$ field only contains these objects, while the remaining fields also contain a substantial number of structures extended across several lattice sites. This means that if the picture outlined in the earlier sections is correct, then the $\hat{B}_x$, $\hat{E}_y$, $\hat{E}_z$, $\hat{B}_y$ and $\hat{B}_z$ fields can only contain points and the strings described in table \ref{tab:lotsofstrings}. We also note that the data for the $\hat{E}_y$, $\hat{E}_z$, $\hat{B}_y$ and $\hat{B}_z$ fields are indistinguishable, in line with our expectations.

\begin{figure}
\begin{center}
\input{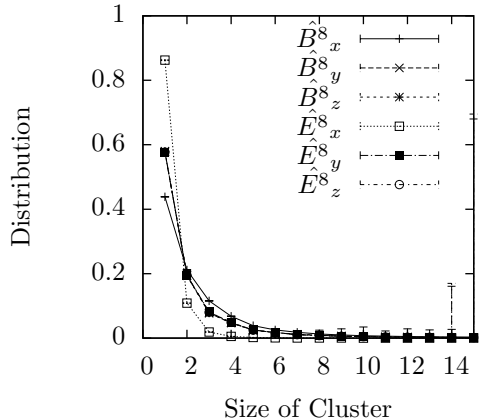}
\end{center}
\caption{The distribution of the number of the lattice sites in each cluster for the $\beta = 8.52$ ensemble.}\label{fig:ClusterSize}
\end{figure}

We next study the number of nearest neighbours within a cluster in figures \ref{fig:ClusterNeigh}, \ref{fig:ClusterfNeigh} and \ref{fig:ClusterffNeigh}. The first of these plots shows the data across all the clusters, while the second filters only those clusters containing four or more lattice sites, so eliminating the points. We have studied the effects of changing this filter from all clusters with at least three sites to all clusters with at least six lattice sites, and there are no substantial changes in the results. The third plot focusses on the very large clusters. A point should have no nearest neighbours (or possibly on occasion one), a string one or two neighbours, a two dimensional surface one, two three or four neighbours (with the proportion depending on the size and shape of the surface), and higher dimensional objects would have a larger number of nearest neighbours within the cluster. The data for $\hat{E}_x$ is again unambiguous: these are points. For the remaining fields, it is less clear. Focussing on figure \ref{fig:ClusterfNeigh}, we see that firstly there is no substantial difference between the results for the different fields: the same structures appear in the fields with the same
frequency. The majority of the lattice sites
within the larger clusters had two or three neighbours, with a small number having four. We can rule out that these are extended in more than two dimensions. This may indicate  strings where two strings are on neighbouring lattice sites and thus counted in the same cluster (and giving objects with three or more nearest neighbours), or it may indicate that there are also some two dimensional surfaces with a large electromagnetic field. There is a small but statistically significant increase in the proportion of sites with two nearest neighbours as we decrease the lattice spacing, and no noticeable difference in the distribution as we increase the volume of the lattice. For the very large clusters, which are extended across the whole lattice, we again see that the sites have two or three nearest neighbours in the cluster. This is consistent with, but does not demonstrate, that these clusters are constructed from densely packed one dimensional objects.

\begin{figure}
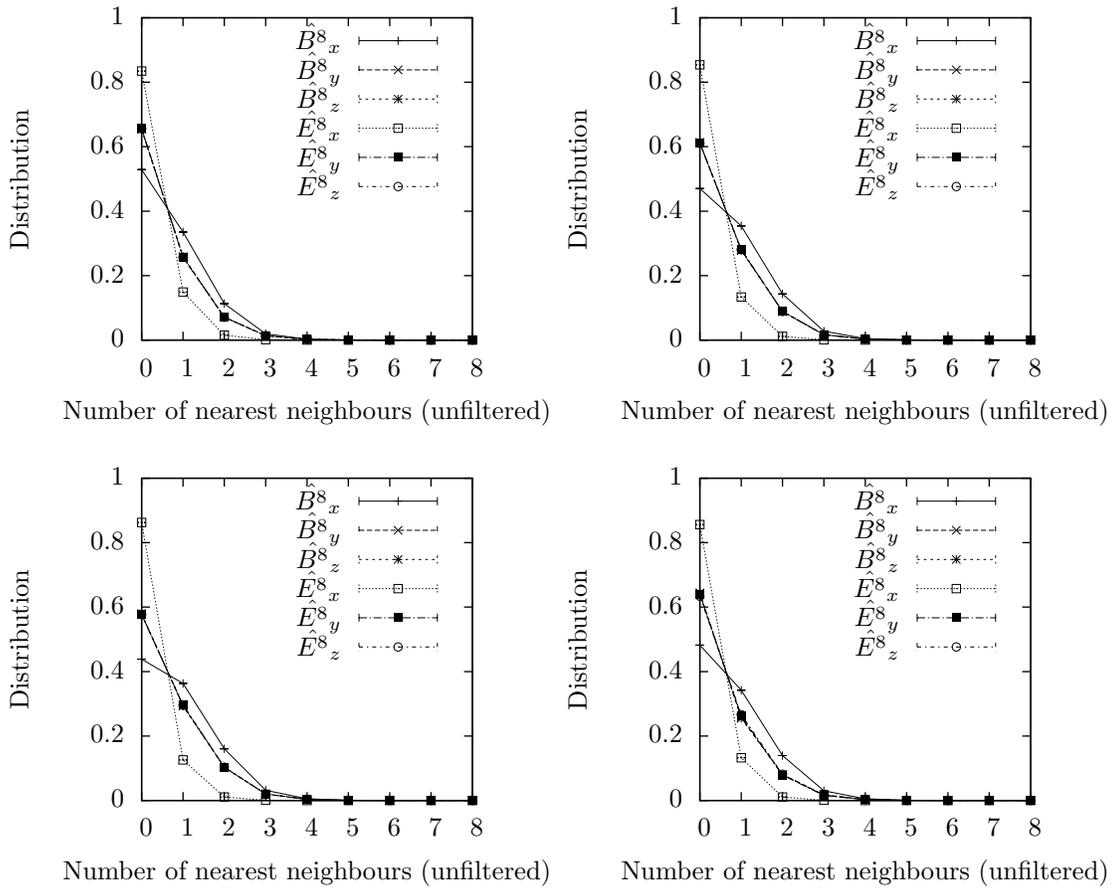

\begin{center}
\begin{tabular}{cc}
\input{Cluster_neighbours_b8_0.ltx}&
\input{Cluster_neighbours_b8_3.ltx}\\
\input{Cluster_neighbours_b8_52.ltx}&
\input{Cluster_neighbours_b8_3_s20.ltx}
\end{tabular}
\end{center}
\caption{The distribution of the number of the number of nearest neighbours within a cluster for each site in a cluster for $\beta = 8.0$ (top left), $\beta = 8.3$ (top right), $\beta = 8.52$ (bottom left) $16^332$ lattices and the $\beta = 8.3$ $20^340$ lattice (bottom right).}\label{fig:ClusterNeigh}
\end{figure}

\begin{figure}
\begin{center}
\begin{tabular}{cc}
\input{Cluster_f_neighbours_b8_0.ltx}&
\input{Cluster_f_neighbours_b8_3.ltx}\\
\input{Cluster_f_neighbours_b8_52.ltx}&
\input{Cluster_f_neighbours_b8_3_s20.ltx}
\end{tabular}
\end{center}
\caption{The distribution of the number of the number of nearest neighbours within a cluster for each site in a cluster for $\beta = 8.0$ (top left), $\beta = 8.3$ (top right), $\beta = 8.52$ (bottom left) $16^332$ lattices and the $\beta = 8.3$ $20^340$ lattice (bottom right), only including those clusters containing at least four lattice sites.}\label{fig:ClusterfNeigh}
\end{figure}

\begin{figure}
\begin{center}
\begin{tabular}{cc}
\input{Cluster_ff_neighbours_b8_0.ltx}&
\input{Cluster_ff_neighbours_b8_3.ltx}\\
\input{Cluster_ff_neighbours_b8_52.ltx}&
\input{Cluster_ff_neighbours_b8_3_s20.ltx}
\end{tabular}
\end{center}
\caption{The distribution of the number of the number of nearest neighbours within a cluster for each site in a cluster for $\beta = 8.0$ (top left), $\beta = 8.3$ (top right), $\beta = 8.52$ (bottom left) $16^332$ lattices and the $\beta = 8.3$ $20^340$ lattice (bottom right), only including those clusters containing at least 200 lattice sites.}\label{fig:ClusterffNeigh}
\end{figure}

We predicted that the strings in the $\hat{B}_x$ field could be in either the $x$ or $t$ direction, while the electric fields $\hat{E}_y$ and $\hat{E}_z$ would extend in the $x$ direction only and the magnetic fields $\hat{B}_y$ and $\hat{B}_z$ in the time direction only. In figure \ref{fig:ClusterExtentBx}, we plot the distribution of the spatial extent in each direction of each cluster (for example, a cluster that stretched from $(t,x,y,z) = (3,4,5,7)$ to $(6,9,8,13)$ would have extents $(6-3+1=4,6,4,7)$; the plot shows the proportion of clusters with a particular extent in each direction). An idealised $x$-string (which would not occur in practice) would have a large extent in the $x$-direction, and extend for only one lattice site in all other directions. This is not fully realised in our data: there is a noticeable number of structures which are extended for up to four or five lattice sites in the `unwanted' directions. However, for each field, the clusters extend considerably more in the directions of the expected strings. We once again restrict the results we present to those clusters containing at least four lattice sites, although once again we have studied numerous values of this filter without seeing any significant difference in the results. The shape of the curve changes as we only include larger clusters (larger extents become more likely as we increase the cutoff for the cluster size), although the maximum in the $X$ and $T$ directions remains at four lattice spacings. This may
be consistent with our
expectations, with strings overlapping each other, or it may indicate that there are additionally other types of objects in the electromagnetic field.

\begin{figure}
\begin{center}
\begin{tabular}{cc}
\input{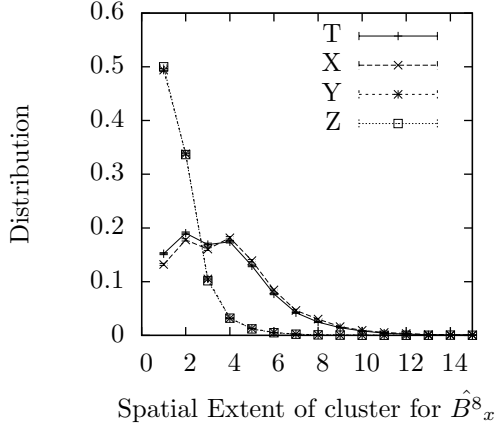} \\
\input{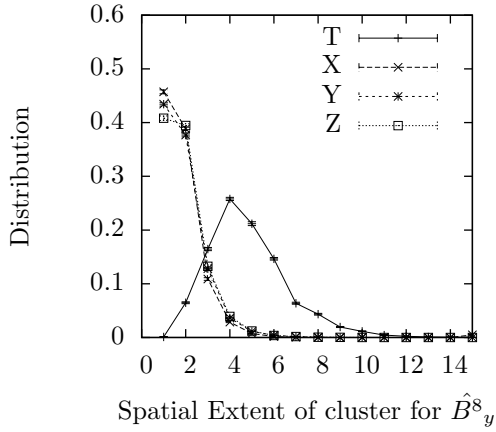} &
\input{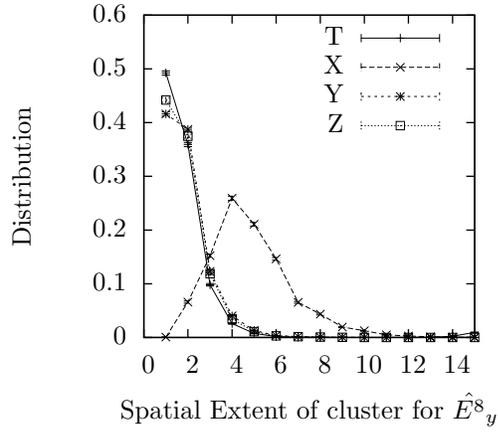}\\
\input{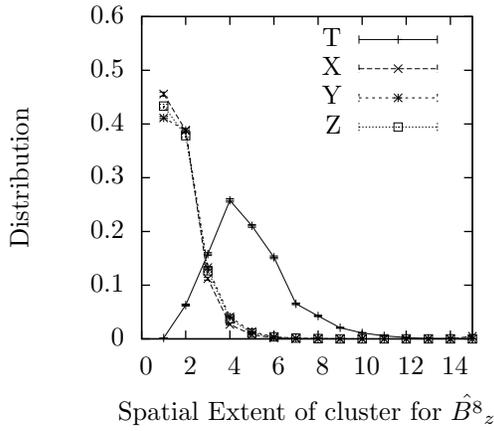} &
\input{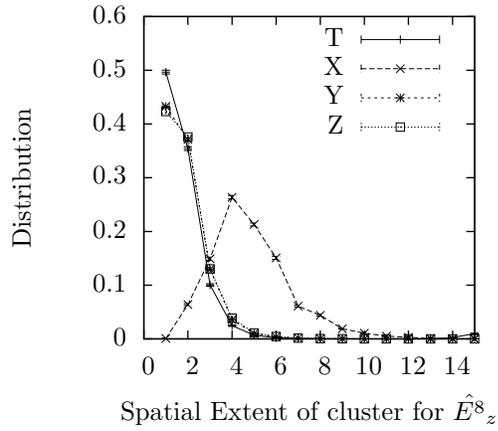}
\end{tabular}
\end{center}
\caption{The distribution of the spatial extent of the clusters for the $\beta = 8.52$ ensemble, only including those clusters containing at least four lattice sites.}\label{fig:ClusterExtentBx}
\end{figure}

One question that can be raised is whether that some clusters have a large thickness in the unwanted directions and some sites within the clusters challenge whether these large fields are arranged in strings one lattice spacing thick. Clearly, in the context of this picture, this can be explained by having strings on neighbouring lattice sites, and some degree of the breakdown of the naive behaviour should be expected. To better judge whether this is indeed the case, we have generated configurations where the field strength is distributed as expected. We have counted the number of lattice sites where $F^j_{\mu\nu} > 1$ and $F^j_{\mu\nu} < -1$, and generated two different sets of configurations of the restricted field strength. In each case, we ensured that the same number of lattice sites had the large $F^j$ as on the actual quenched QCD configuration. In the first random configuration, these sites were allocated completely random positions around the lattice. In the second random configuration, one quarter
of the sites were positioned randomly, and the rest were arranged into strings in either the $x$ or $t$ direction as appropriate for the field. The principle ambiguity is what to use as the distribution for the length of these strings, since the distribution given in figure \ref{fig:ClusterExtentBx} will contain examples from multiple strings. We examined several Poisson and Gaussian distributions, and the data we present here is taken from a Poisson distribution with the mean value extracted from the data presented in figure \ref{fig:ClusterExtentBx}. Only the analogue of the curve in \ref{fig:ClusterExtentBx} in the direction of the string depends strongly on the distribution of the length of the string. We cannot predict the shape of this curve for this simple model. We then simply repeat the cluster analysis on these configurations.

Some sample results are plotted in figures \ref{fig:ran1} and \ref{fig:ran2}.

 \begin{figure}
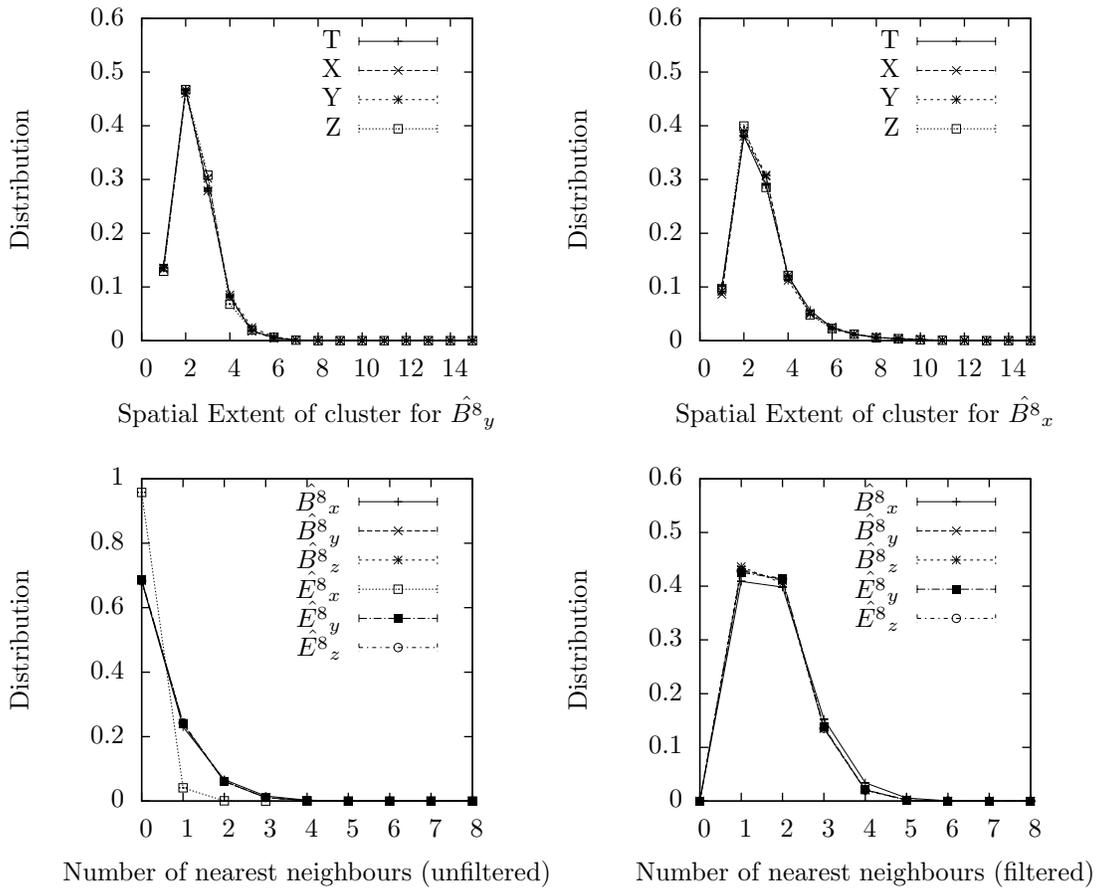

\begin{center}
\begin{tabular}{cc}
\input{Cluster_extent_B_y_b8_52_ran1.ltx}&
\input{Cluster_extent_B_x_b8_52_ran1.ltx} \\
\input{Cluster_neighbours_b8_52_ran1.ltx} &
\input{Cluster_f_neighbours_b8_52_ran1.ltx}
\end{tabular}
\end{center}
\caption{The distribution of the spatial extent of the clusters for where the field strength is selected according to the distribution given in figure \ref{fig:ClusterSize} for the $\beta = 8.52$ ensemble for the $\hat{B}_y$ (top left) and $\hat{B}_x$ (top right) fields, but with the lattice sites randomised; the distribution of the number of nearest neighbours (bottom left), and the filtered distribution of the number of nearest neighbours (bottom right) for this random ensemble. }\label{fig:ran1}
\end{figure}

 \begin{figure}
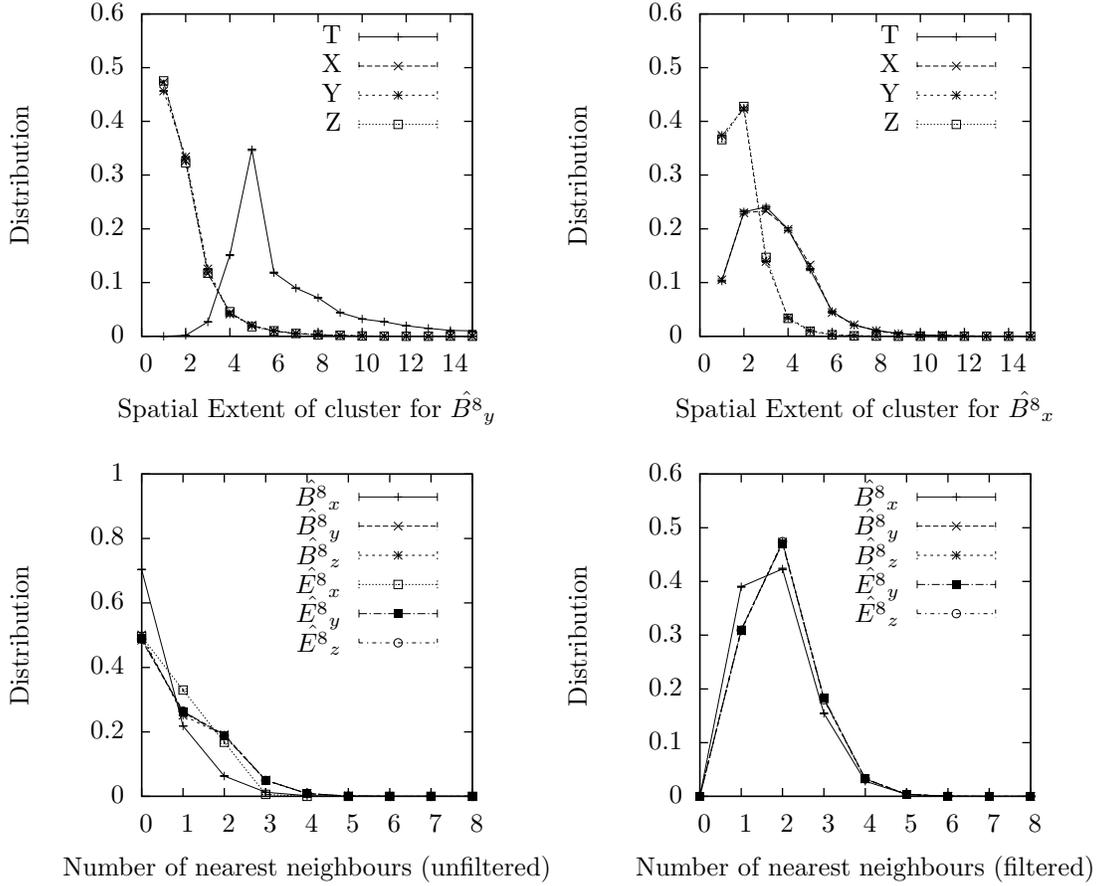

\begin{center}
\begin{tabular}{cc}
\input{Cluster_extent_B_y_b8_52_ran2.ltx}&
\input{Cluster_extent_B_x_b8_52_ran2.ltx} \\
\input{Cluster_neighbours_b8_52_ran2.ltx} &
\input{Cluster_f_neighbours_b8_52_ran2.ltx}
\end{tabular}
\end{center}
\caption{The distribution of the spatial extent of the clusters for where the field strength is selected according to the distribution given in figure \ref{fig:ClusterSize} for the $\beta = 8.52$ ensemble for the $\hat{B}_y$ (top left) and $\hat{B}_x$ (top right) fields, but with the lattice sites randomly distributed with three quarters of the larger field strengths arranged into strings; the distribution of the number of nearest neighbours (bottom left), and the filtered distribution of the number of nearest neighbours (bottom right) for this random ensemble.}\label{fig:ran2}
\end{figure}
It can be seen that the data for figure \ref{fig:ran1} is qualitatively similar to the actual QCD data for $\hat{E}_x$, while there are disagreements when comparing the model and actual data sets for the other fields. The second model, with strings, compares well with the numerical data for $\hat{B}_y$, $\hat{B}_z$, $\hat{E}_y$ and $\hat{E}_z$; in particular the plots for the number of nearest neighbours and the extent of cluster in the directions away from the string are close to the QCD data. The extent of the cluster in the direction of the string strongly depends on the distribution for the string length, so we do not expect a perfect agreement here. In particular, that the `thickness' of the string agrees with the data, and seems to be largely model independent, suggests that this cannot be used to rule out the predicted string model of flux. However, the broad agreement between our model calculation  and the data implies that the QCD $\hat{E}_x$ field is dominated by points and the other fields by a mixture of points
and strings.

Finally, we check to see whether the locations of the clusters in the different fields are correlated. If our model is correct, then the space-time location of many of the peaks in $E_x$ should be correlated with high field strengths in the other fields. To do this, we take each of the peaks in $\hat{E}_x$ (i.e. each point where $|\hat{H}_{xt}| > 1$) and plot a histogram of the value of the other fields around that site. For comparison, we produce a similar plot measuring the field strength around random lattice sites rather than the peaks of the $\hat{E}_x$ field. We would expect the plot for the actual data to have a higher concentration of large field strengths than the random data, and a smaller concentration of lower field strengths. There is one complication: we measure the field strength at different locations in space-time for each field. For example, we measure $\hat{F}_{xt}(t + \frac{1}{2},x+\frac{1}{2},y,z) \equiv \hat{E}_x(t,x,y,z)$ and $\hat{F}_{yz}(t ,x+\frac{1}{2},y,z+ \frac{1}{2}) \equiv \hat{B}_y(t,x,y,z)$. This means that there may be a shift of up to one lattice spacing in the location of the peak within each of the measured fields. We therefore take the maximum value of $|\hat{B}_y(t,x,y,z)|$, $|\hat{B}_y(t+1,x,y,z)
|$, $|\hat{B}_y(t,x,y,z-
1)|$ and $|\hat{B}_y(t+1,x,y,z-1)|$ (for example). This creates a bias towards larger field strengths for both the actual locations of the peaks in $\hat{E}_x$ and the random data. This is particularly pronounced for $\hat{B}_x$, where there are 16 lattice sites which we need to consider, and coupled with the higher density of lattice sites with a large value of the $\hat{B}_x$ field this means that actual data from $\hat{B}_x$ does not differ greatly from the randomised data. We plot the results in figure \ref{fig:strhist} for the $\beta = 8.52$ ensemble. There is a statistically significant difference between the curves for the $\hat{E}_x$ peaks and the random data for all the fields, with a larger proportion of sites with a large value of the $\hat{B}_y$, $\hat{B}_z$ $\hat{E}_y$ and $\hat{E}_z$ fields and a smaller proportion of sites with a small value of these field surrounding the peaks in $\hat{E}_x$ than for the random data. There is, however, still a large number of sites where there does not seem to be a great deal of correlation between the peaks in (for example)
 $\hat{E}_x$ and $E_
y$; partly this can be explained by those sites where in practice $\hat{E}_y \sim  \pm\sqrt{3}\pi$ but we measure $\hat{E}_y = 0$, but it still seems as though not every peak in $\hat{E}_x$ corresponds to a peak in every one of the other fields. This is not inconsistent with equation (\ref{eq:70}); for example if $\partial_1 G$ is small then only three of the other fields are large when there is a contribution in $\hat{E}_x$, and if both $\partial_1 G$ and $\partial_2 G$ are small then only two of them do. This data is then qualitatively but not conclusively in line with the expectations: there is a clear correlation between the fields, however it is not large enough to say that all the the topological structures in the different components of $\hat{F}_{\mu\nu}$ are correlated with each other.

\begin{figure}
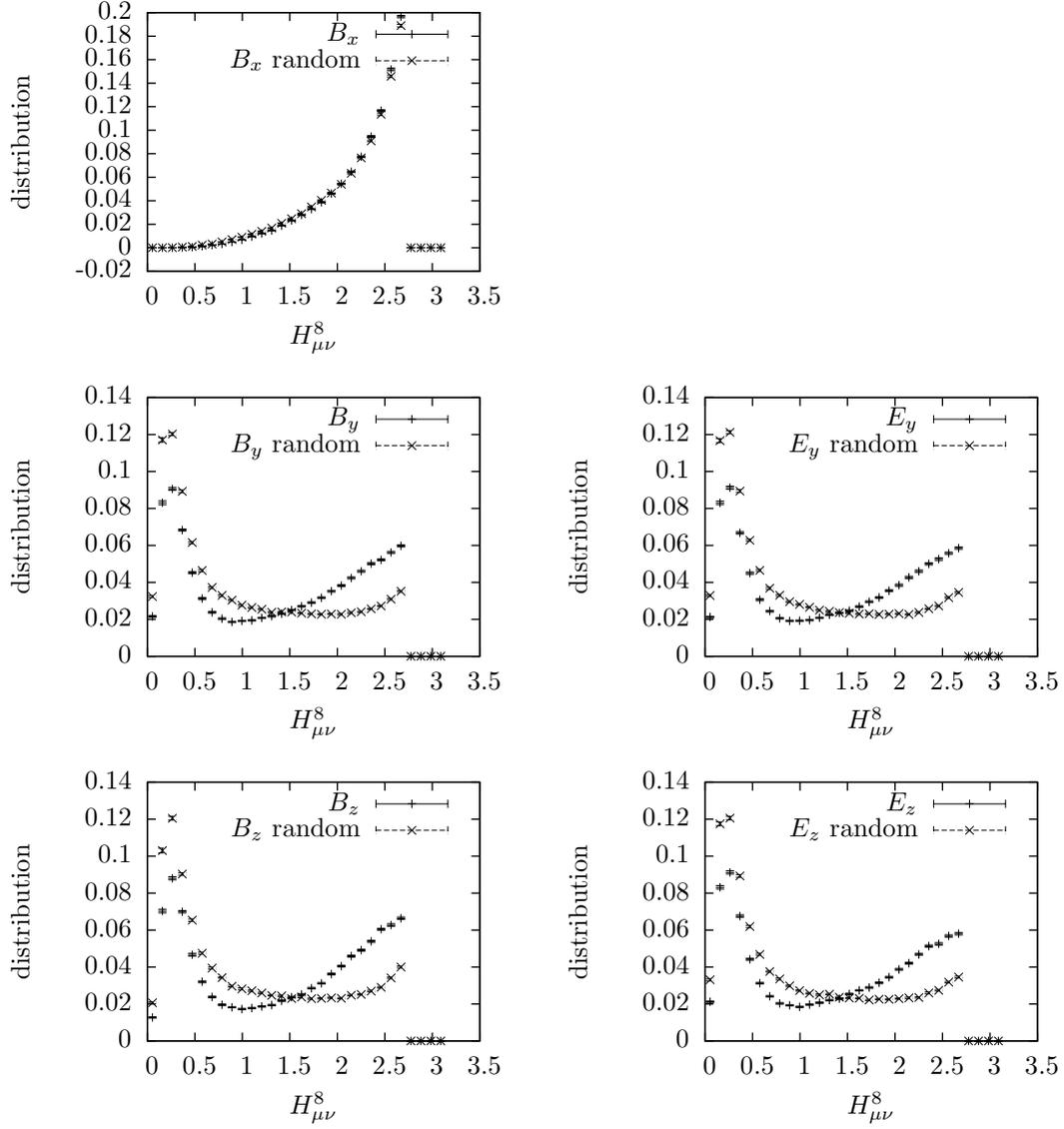

\begin{center}
\begin{tabular}{cc}
\input{strhist_Bx_852.ltx} \\
\input{strhist_By_852.ltx} &
\input{strhist_Ey_852.ltx}\\
\input{strhist_Bz_852.ltx} &
\input{strhist_Ez_852.ltx}
\end{tabular}
\end{center}
\caption{The distribution of the value of the $\hat{H}^8_{\mu\nu}$ field at the location of the peaks in the $\hat{E}_x$ field and at random lattice sites for the $\beta = 8.52$ ensemble.}\label{fig:strhist}
\end{figure}

\begin{figure}
\begin{center}
\begin{tabular}{cc}
\input{Cluster_extent_B_x_b8_0_XT.ltx} \\
\input{Cluster_extent_B_y_b8_0_XT.ltx} &
\input{Cluster_extent_E_y_b8_0_XT.ltx}\\
\input{Cluster_extent_B_z_b8_0_XT.ltx} &
\input{Cluster_extent_E_z_b8_0_XT.ltx}
\end{tabular}
\end{center}
\caption{The distribution of the spatial extent of the clusters in the $xt$ planefor the $\beta = 8.0$ ensemble, only including those clusters containing at least four lattice sites.}\label{fig:ClusterExtentBx_XT}
\end{figure}
\begin{figure}
\begin{center}
\begin{tabular}{cc}
\input{Cluster_extent_B_x_b8_0_YT.ltx} \\
\input{Cluster_extent_B_y_b8_0_YT.ltx} &
\input{Cluster_extent_E_y_b8_0_YT.ltx}\\
\input{Cluster_extent_B_z_b8_0_YT.ltx} &
\input{Cluster_extent_E_z_b8_0_YT.ltx}
\end{tabular}
\end{center}
\caption{The distribution of the spatial extent of the clusters in the $yt$ planefor the $\beta = 8.0$ ensemble, only including those clusters containing at least four lattice sites.}\label{fig:ClusterExtentBx_YT}
\end{figure}
\begin{figure}
\begin{center}
\begin{tabular}{cc}
\input{Cluster_extent_B_x_b8_0_XY.ltx} \\
\input{Cluster_extent_B_y_b8_0_XY.ltx} &
\input{Cluster_extent_E_y_b8_0_XY.ltx}\\
\input{Cluster_extent_B_z_b8_0_XY.ltx} &
\input{Cluster_extent_E_z_b8_0_XY.ltx}
\end{tabular}
\end{center}
\caption{The distribution of the spatial extent of the clusters in the $xy$ planefor the $\beta = 8.0$ ensemble, only including those clusters containing at least four lattice sites.}\label{fig:ClusterExtentBx_XY}
\end{figure}

This analysis has concentrated on the smaller, isolated structures. We still have to consider the larger structure which extends throughout the lattice, which we have seen from the nearest neighbour analysis is either one or two dimensional. This is consistent with our expectations if it is constructed from neighbouring strings along the $X$ or $T$ directions (so we may have an $x$-string at $y=z=t = 0$ and extending from $x = 0$ to $x = 5$, and another string at $y =1$, $z = t = 0$ extending from $x = 5$ to $x = 10$, and so on; if these are closely enough packed then our analysis will place them into the same cluster). We investigate this structure by modifying the cluster algorithm. We firstly select two directions; and we construct the cluster by only considering nearest neighbours in those two directions. This restricts us to isolated objects. By varying the two directions, we can see if the strings predominantly extend in particular directions.

We display the extent of the strings in figures \ref{fig:ClusterExtentBx_XT}, \ref{fig:ClusterExtentBx_YT} and \ref{fig:ClusterExtentBx_XY}. We have shown the $xt$, $xy$ and $yt$ planes as examples. There was again little difference between our ensembles. Again, we restrict ourselves to clusters with at least four lattice sites, so there must be an extent of on average two or three lattice sites in at least one direction. We again see that the magnetic fields $\hat{B}_y$ and $\hat{B}_z$ are extended far more parallel to the $t$ axis than in the other directions, $\hat{E}_x$ and $\hat{E}_z$ along the $x$-axis, while $\hat{B}_x$ is extended along both the $t$ and $x$ axes. None of the fields show similar features parallel
to the $y$ and $z$ axis.

The picture we see is consistently that the restricted $\hat{E}_x$ field is dominated by points, while the $\hat{B}_x$ field by strings in both directions in the $xt$ plane, and the other electric fields by strings in the $x$ direction, while the other magnetic fields by strings in the $t$ direction. These are not magnetic monopoles (or not only monopoles): the behaviour of the electric fields and $\hat{B}_x$ rules this out (and note that we presented results for the supposedly `magnetic' $\tr n [\partial_\mu n ,\partial_\nu n]$ part of the field strength; although the picture for the full field strength is the same). The results for the restricted electric and magnetic fields are consistent with the expectation (though, of course, they do not prove that the model of confinement we presented in the first part of this paper is correct: they may also be consistent with other models of the vacuum).

\section{Conclusions}\label{sec:5}
We have investigated whether the confining quark potential in quenched QCD is caused by CDG gauge-invariant \monsp. Our main novelty is to construct the CDG colour field $n^j = \theta \lambda^j \theta^\dagger$ from the eigenvectors of the Wilson loop, which allows us to access the full symmetry group of the \mons and permits a theoretical discussion. From this analysis, we have shown that the Wilson Loop for the actual QCD gauge field is identical to the Wilson Loop constructed from the restricted Abelian field, and this allows us to express the Wilson Loop in terms of an area integral over its surface. We then find that the construction of $\theta$ allows for certain topological defects to appear in the restricted field strength. In the component of the restricted field strength in the plane of the Wilson Loop, these manifest themselves as point-like objects; in the other restricted electric and magnetic fields they manifest themselves as strings parallel to one of the axes of the Wilson Loop. The net
result is that the restricted field strength is dominated by a large structure which is constructed from these one dimensional strings. The full QCD field strength contains the restricted field,
\begin{gather}
F_{\mu\nu}[A] = \hat{F}_{\mu\nu}[\hat{A}] + F_{\mu\nu}[X] + ig([\hat{A}_\mu,X_\nu] + [X_\mu,\hat{A}_\nu]),
\end{gather}
so it may be expected that these global one dimensional structure may extend to the full QCD field strength (albeit with O(4) symmetry suggesting that $F_{\mu\nu}[X]$ will also contain strings in other directions). If these strings are sufficiently dense, this picture is consistent with lattice simulations~\cite{Horvath:2003yj,Horvath:2005rv}. In particular, these objects are not magnetic monopoles; at least for this particular choice of the colour field. The peaks in the $\hat{E}_x$ field lead to an area law for the Wilson Loop and thus confinement.

We have also given an argument within the context of this decomposition which might explain string breaking at large spatial distances.

Our numerical results focus on searching for these strings and peaks in the restricted field strength, and we have uncovered firm evidence that the electric field in the direction of the Wilson Loop, which is responsible for the confining potential, is dominated by objects which are just present on a single lattice spacing. The remaining components of the electric and magnetic fields contain one dimensional objects which extend in the directions suggested by our theoretical analysis. While this does not prove that our analysis is correct, they are certainly consistent with it; and whatever these peaks are, they are responsible for the confining potential. Our numerical analysis is consistent with the vacuum related to the restricted field being dominated by a combination of magnetic monopoles and the $\pi_1$ topological objects we have labelled \monsp. It is also possible that the binding of  monopoles and anti-monopoles into pairs (similar to the Cooper pairs in superconductivity) might lead to a monopole condensate with a different pattern in the field strength than one would expect from a naive picture of isolated monopoles and anti-monopoles, and we cannot rule out this possibility.  We also show that in an initial test, making an important approximation, the string tension can be accounted for by the second, `magnetic', term within the restricted field strength. The results of the full calculation (without any approximation) are less clear: all we have shown so far is that we will need to make a careful study of the continuum limit. The objects responsible for confinement are embedded in the colour field, and only indirectly in the gauge field.

 We believe that this evidence is enough to merit further, more detailed, investigation of whether confinement can be explained within the context of this particular formulation of the CDG Abelian decomposition.

We have not yet investigated the effects of changing the representation of the group, of incorporating quark loops into our lattice study, or studied the de-confinement transition, and hope to investigate at least some of these questions in future work.
\section*{Acknowledgments}

umerical calculations used servers at Seoul National University
funded by the BK21 program of the NRF (MEST), Republic of
Korea. NC was supported in part through the BK21 program of the NRF (MEST), Republic of Korea. This research was supported by Basic Science Research Program through the National Research Foundation of Korea(NRF) funded by the Ministry of Education(2013057640). W. Lee is supported by the Creative Research Initiatives
Program (2013-003454) of the NRF grant, and acknowledges the support
from the KISTI supercomputing center through the strategic support
program for the supercomputing application research
(KSC-2012-G2-01).

 YMC is supported in part by NRF grant (2012-002-134)
funded by MEST.
\appendix
\section{The Abelian Decomposition}\label{app:A}

\subsection{Proof of the existence of $\theta$}\label{app:A.2}
In the text, we defined $\theta$ as the field which diagonalises the gauge link $U_\sigma$ along the Wilson Line. We here demonstrate that this field exists, using the following three lemmas (and the well known properties of the eigenvectors of normal and unitary matrices, including that a unitary matrix can always be diagonalised).
\begin{lemma}
Suppose that there is a $U(N_C)$ field $\theta$ which diagonalises the gauge links $U_\sigma$ which contribute to the closed Wilson Line $W[C_s]$ of length $L$ starting and ending at a position $s$, so $\theta^\dagger_\sigma U_\sigma \theta_{\sigma + \delta \sigma} = e^{i \delta \sigma u^j_\sigma \lambda^j}$ for diagonal Gell-Mann matrices $\lambda^j$. Then $\theta_s$ contains the eigenvectors of $W[C_s]$, with eigenvalues $e^{i\sum_\sigma \delta \sigma u^j_\sigma \lambda^j}$, i.e. (given that $W$ is unitary and thus normal) $W[C_s] \theta_s = \theta_s e^{i\sum_\sigma \delta \sigma u^j_\sigma \lambda^j}$.\label{lem:a1}
\end{lemma}
\begin{proof}
$W[C_s]$ is defined as
\begin{gather}
W[C_s] = U_s U_{s+\delta \sigma}\ldots U_{s+L-2\delta\sigma}U_{s+L-\delta\sigma}.
\end{gather}
and, since $\theta_s \equiv \theta_{s+L}$ as $s$ and $s+L$ refer to the same location in space,
\begin{align}
W[C_s]\theta_s =& U_s U_{s+\delta \sigma}\ldots U_{s+L-2\delta\sigma}U_{s+L-\delta\sigma}\theta_s\nonumber\\
=& U_sU_s U_{s+\delta \sigma}\ldots U_{s+L-2\delta\sigma} \theta_{s+L-\delta\sigma} e^{i u^j_{s+L-\delta\sigma} \lambda^j \delta\sigma}\nonumber\\
=& U_sU_s U_{s+\delta \sigma}\ldots  \theta_{s+L-2\delta\sigma} e^{i ( u^j_{s+L-\delta\sigma} + u^j_{s+L-2\delta\sigma}) \lambda^j \delta\sigma}\nonumber\\
=& \theta_s e^{i \sum_\sigma u^j_\sigma \delta\sigma \lambda^j}.
\end{align}
where the last line follows by repeatedly evaluating $U \theta  = \theta e^{iu^j\lambda^j \delta\sigma}$, and we have noted that the diagonal Gell-Mann matrices commute with each other. But, since $\sum_\sigma u^j_\sigma \delta\sigma \lambda^j$ is diagonal, this is just an eigenvalue equation. Hence we have proved the result.
\end{proof}
\begin{lemma}
The eigenvalues of the Wilson Line along a closed curve $C$ are independent of which location $s$ along the curve is used as the start and end of the line.
\end{lemma}
\begin{proof}
We note that
\begin{align}
W[C_s] =& U_s U_{s+\delta \sigma}\ldots U_{s+L-2\delta\sigma}U_{s+L-\delta\sigma}\nonumber\\
W[C_{s+\delta \sigma}] =&  U_{s+\delta \sigma}\ldots U_{s+L-2\delta\sigma}U_{s+L-\delta\sigma}U_s\nonumber\\
W[C_{s+\delta \sigma}] =& U_s^\dagger W[C_s] U_s.
\end{align}
But $U_s$ is just a unitary matrix, which means that if $\theta_{s+\delta \sigma}$ diagonalises $W[C_{s+\delta \sigma}]$ then $U_s \theta_{s+\delta \sigma}$ diagonalises $W[C_s]$ with the same eigenvalues. Thus, with $\theta_s\equiv U_s \theta_{s+\delta s}$,
\begin{gather}
\theta^\dagger_sW[C_s] \theta_s = \theta^\dagger_{s+\delta \sigma} W[C_{s+\delta \sigma}] \theta_{s+\delta \sigma} = e^{i \sum_\sigma u^j_\sigma \delta\sigma \lambda^j}.
\end{gather}
The result follows by induction.
\end{proof}
\begin{lemma}
Suppose that there is a field $\theta_\sigma$ which diagonalises $W[C_\sigma]$ at each location $\sigma$ along the curve, with non-degenerate eigenvalues given by the diagonal elements of $e^{i \sum_\sigma u^j_\sigma \delta\sigma \lambda^j}$ then $\theta^\dagger_\sigma U_\sigma \theta_{\sigma + \delta \sigma}$ is also diagonal.\label{lem:a3}
\end{lemma}
\begin{proof}
We have, using the results of the previous lemma,
\begin{align}
e^{i \sum_\sigma u^j_\sigma \delta\sigma \lambda^j} = &\theta^\dagger_{s+\delta \sigma} W[C_{s+\delta \sigma}] \theta_{s+\delta \sigma}\nonumber\\
=&\theta^\dagger_{s+\delta \sigma}U_s^\dagger \theta_s \theta_s^\dagger W[C_{s}]\theta_s \theta_s^\dagger U_s \theta_{s+\delta \sigma}\nonumber\\
=&\theta^\dagger_{s+\delta \sigma}U_s^\dagger \theta_s e^{i \sum_\sigma u^j_\sigma \delta\sigma \lambda^j} \theta_s^\dagger U_s \theta_{s+\delta \sigma}
\end{align}
which means that
\begin{gather}
[\theta^\dagger_{s}U_s \theta_{s+\delta\sigma},e^{i \sum_\sigma u^j_\sigma \delta\sigma \lambda^j}] = 0,
\end{gather}
which, unless $W[C_s]$ has degenerate eigenvalues, is only possible if $\theta^\dagger_{s}U_s \theta_{s+\delta\sigma}$ is diagonal, as required by the lemma.
\end{proof}
There always exists a $\theta$ field (unique baring the $U(1)^{N_C}$ phase factors and ordering of the eigenvectors) which diagonalises a $W[C_s] \in SU(N_C)$ as long as $W$ does not have degenerate eigenvalues. From lemmas \ref{lem:a1} and \ref{lem:a3}, we see that this field is also a unique (up to the same caveats) solution for a field that diagonalises $U$. Thus we have shown that there is a $SU(N_C)/U(1)^{N_C-1}$ field $\theta$ which diagonalises the gauge links along the closed Wilson Line.
\subsection{The Abelian Decomposition in the continuum}
Here we construct an explicit form for $\hatA$. The argument here is based on that presented in~\cite{Kondo:2008su}.
\begin{lemma}
For any Field $Y$ in the adjoint representation of a SU($N_C$) lie algebra
\begin{gather}
Y = \sum_j\left( \frac{1}{2}\lambda^j \tr(\lambda^jY) + \frac{1}{4}[\lambda^j,[\lambda^j,Y]]\right),\label{eq:YH}
\end{gather}
where $\lambda^j$ are the diagonal elements of the lie algebra normalised so that $\tr(\lambda^j \lambda^k) = 2\delta^{jk}$.
\end{lemma}
\begin{proof}
$Y$ can be decomposed as $Y = \lambda^a Y^a$, where $\lambda^a/2$ are the generators of the gauge group. $\lambda^a$ are Hermitian, traceless matrices which satisfy $\tr \lambda^a \lambda^b = 2\delta^{ab}$. We adopt the convention that the diagonal Gell-Mann matrices are referred to with the index $j,k,\ldots$, while $\lambda^a, \lambda^b,\ldots$ may refer to any of the Gell-Mann matrices (including the diagonal components). The diagonal generators commute with each other, $[\lambda^j,\lambda^k] = 0$, and can be parametrised as $\text{diag}(1,-1,0,0,\ldots)$, $\frac{1}{\sqrt{3}}\text{diag}(1,1,-2,0,\ldots)$ $\ldots$ $\frac{\sqrt{2}}{\sqrt{N_C(N_C-1)}}\text{diag}(1,1,1,1,\ldots,-(N_C - 1))$. The other elements can be expressed in terms of the basis $E_{\pm \alpha}$, where $E_{+|\alpha|}$ consists of one element of $1$ above the diagonal and its symmetric counterpart below it, while $E_{-|\alpha|}$ has $-i$ above the diagonal and $i$ below it. For example,
\begin{align}
E_1 =& \lambda^1 = \left(\begin{array}{c c c} 0&1&\cdots\\1&0&\cdots\\\vdots&\vdots&\ddots\end{array}\right) & E_{-1} =& \lambda^2 = \left(\begin{array}{c c c} 0&-i&\cdots\\i&0&\cdots\\\vdots&\vdots&\ddots\end{array}\right)
\end{align}

Now, as $\lambda^j$ is diagonal, $[\lambda^j,E_\alpha]$ must have the same non-zero components as $E_\alpha$, and thus be a linear combination of $E_\alpha$ and $E_{-\alpha}$. Given that $\tr E_\alpha [\lambda^j,E_\alpha] = 0$, we can write $[\lambda^j,E_\alpha] = a^j_\alpha E_{-\alpha}$ for some coefficient $a$ which needs to be determined ($a$ could be $0$). Similarly, $[E_{\alpha},E_{-\alpha}]$ is diagonal and traceless, so $[E_{\alpha},E_{-\alpha}] = \sum_k b^k_\alpha \lambda^k$, for some coefficients $b$. The following results then proceed immediately using the trace theorem $\tr (A[B,C]) = \tr(C[A,B])$:
\begin{align}
[\lambda^j,[\lambda^j,E_\alpha]] =& a^j_\alpha a^j_{-\alpha} E_\alpha\nonumber\\
\tr (E_\alpha [\lambda^j,[\lambda^j,E_\alpha]]) = & 2a^j_\alpha a^j_{-\alpha} = \tr([\lambda^j,E_\alpha][E_\alpha,\lambda^j]) = - 2(a^j_\alpha)^2\nonumber\\
a^j_\alpha =& - a^j_{-\alpha}\nonumber\\
b^k_\alpha =& \frac{1}{2}\tr(\lambda^k [E_{\alpha},E_{-\alpha}]) = \frac{1}{2}\tr(E_{-\alpha}[\lambda^k,E_\alpha]) = a^k_\alpha\nonumber\\
\tr([E_{\alpha},E_{-\alpha}]^2) =& -8 = \tr(\lambda^j b^j_\alpha \lambda^k b^k_\alpha) = 2\sum_k (b^k_\alpha)^2 = 2\sum_k (a^k_\alpha)^2.
\end{align}
The first equality of the last equation follows from the specific form of $E_\alpha$ and $E_{-\alpha}$: $[E_{\alpha},E_{-\alpha}]$ has two non-zero terms on the diagonal, one $2i$ and the other $-2i$.
We now have the necessary ingredients to prove equation (\ref{eq:YH}). Expanding $Y = c^j \lambda^j + d_\alpha E_\alpha$, the right hand side reads
\begin{align}
\sum_j\left(\frac{1}{2} \lambda^j \tr(\lambda^j (c^k \lambda^k + d_\alpha E_\alpha)) + \frac{1}{4}[\lambda^j,[\lambda^j,c^k \lambda^k + d_\alpha E_\alpha]]\right) = & \sum_j\left( \lambda^j c^j - \frac{1}{4} (a_\alpha^j)^2 d_\alpha E_\alpha\right)\nonumber\\
=&c^j\lambda^j + d_\alpha E_\alpha = Y
\end{align}
\end{proof}

\begin{remark}
In the Lemma above, we used a natural representation of $\lambda^j$ and $E_\alpha$ in  terms of the Gell-Mann matrices. However, it is easy to see that the same commutation relations between the generators apply if we rotate the basis $\lambda^j \rightarrow \theta \lambda^j \theta^\dagger$, $E_\alpha \rightarrow \theta E_\alpha \theta^\dagger$, with $\theta \in SU(N_C)$, while the traces of powers of the operators are also untouched. As the proof of (\ref{eq:YH}) just depends on the commutation relations and the traces, it is clear that a form of this equation is also valid in the rotated basis. In particular, this means that we can replace $\lambda^j$ in equation (\ref{eq:YH}) with the rotated basis $n^j$, leaving us with
\begin{gather}
Y = \sum_j\left(\frac{1}{2} n^j \tr(n^jY) + \frac{1}{4}[n^j,[n^j,Y]]\right).\label{eq:YH2}
\end{gather}
\end{remark}
Using equation (\ref{eq:YH2}), we can use the defining equations to construct $\hatA$. From the second equation (\ref{eq:contde2}), we immediately have
\begin{gather}
\tr(n^j \hatA) = \tr(n^j A).
\end{gather}
From equation (\ref{eq:contde1}),
\begin{align}
0=& \partial_\mu n^j - ig [\hat{A}_\mu,n^j]\nonumber\\
0=& \sum_j\left( [n^j, \partial_\mu n^j] - i g [n^j,[\hatA_\mu,n^j]]\right)\nonumber\\
=& \sum_j [n^j, \partial_\mu n^j] + 2ig \left(2\hatA_\mu - \sum_j n^j \tr n^j \hatA_\mu\right)\nonumber\\
\hatA_\mu = & \sum_j\left(\frac{1}{2}n^j \tr (n^j A_\mu) + \frac{i}{4g} [n^j, \partial_\mu n^j] \right).\label{eq:67865}
\end{align}
The second field $X_\mu$ may be constructed from $X_\mu = A_\mu - \hatA_\mu$.

\begin{theorem}
The definitions of the restricted gauge field given in equations (\ref{eq:hatA1}) and (\ref{eq:hatA}), i.e.
\begin{align}
g\hat{A}_\mu =& - \hat{u}^j_\mu n^j  + i \theta\partial_\mu \theta^\dagger\label{eq:9910}\\
\hatA_\mu = & \sum_j\left(\frac{1}{2}n^j \tr (n^j A_\mu) + \frac{i}{4g} [n^j, \partial_\mu n^j] \right)\label{eq:678651}
\end{align}
 are equivalent, where $\hat{u}^j_\mu$ is a function of $\theta$  and $A_\mu$ which needs to be determined.
\end{theorem}
\begin{proof}
We start from equation (\ref{eq:678651}), and write, with the help of equation (\ref{eq:YH2}) and $\partial_\mu n^j = [n^j,\theta\partial_\mu \theta^\dagger]$
\begin{align}
\hatA_\mu = & \sum_j\left(\frac{1}{2}n^j \tr (n^j A_\mu) + \frac{i}{4g} [n^j, \partial_\mu n^j] \right)\nonumber\\
=&\frac{1}{2} n^j \tr n^j A_\mu + \frac{i}{4g} [n^j,[n^j,\theta \partial_\mu \theta^\dagger]]\nonumber\\
=&\frac{1}{2} n^j \tr (n^j(A_\mu - \frac{i}{g}\theta \partial_\mu \theta^\dagger)) + \frac{i}{g}\theta \partial_\mu \theta^\dagger
\end{align}
and if
\begin{gather}
\hat{u}^j_\mu \equiv -\frac{1}{2}\tr(n^j(gA_\mu - i\theta \partial_\mu \theta^\dagger)),\label{eq:reallydull}
\end{gather}
 we recover equation (\ref{eq:9910}).
\end{proof}
\begin{theorem}
The restricted field strength can be written in the form given in equation (\ref{eq:26})
\begin{gather}
F_{\mu\nu}[\hatA] = \frac{1}{2}\left(n^j \partial_\mu \tr(n^j \hat{A}_\nu) -  n^j \partial_\nu \tr(n^j \hat{A}_\mu) \right)+ \frac{i}{8g}n^j \tr(n^j [\partial_\mu n^k,\partial_\nu n^k])\label{eq:f2}
\end{gather}
\end{theorem}
\begin{proof}
We first of all construct $F_{\mu\nu}$. For any field $X$, we have
\begin{gather}
D_\mu[\hatA] D_\nu[\hatA] X  = \partial_\mu \partial_\nu X - i g \partial_\mu[\hatA_\nu,X] - ig [\hatA_\mu,\partial_\nu X] - g^{2}[\hatA_\mu,[\hatA_\nu,X]]
\end{gather}
Anti-symmetrising this in the indices $\mu$ and $\nu$ gives
\begin{align}
[D_\mu[\hatA],D_\nu[\hatA]] X =& -ig[\partial_\mu \hatA_\nu - \partial_\nu \hatA_\mu,X] - g^2([A_\mu,[A_\nu,X]] - [A_\nu,[A_\mu,X]])\nonumber\\
=&-ig[\partial_\mu\hatA_\nu - \partial_\nu \hatA_\mu - i g[\hatA_\mu,\hatA_\nu],X],
\end{align}
and
\begin{gather}
F_{\mu\nu}[\hatA] = \partial_\mu\hatA_\nu - \partial_\nu \hatA_\mu - i g[\hatA_\mu,\hatA_\nu]
\end{gather}
where we have made use of the Jacobi Identity $[A,[B,C]] + [B,[C,A]] + [C,[A,B]] = 0$ and we have used the definition of $F_{\mu\nu}$ given in equation (\ref{eq:f7}).

We may now use equation (\ref{eq:67865}) to obtain an expression for $F_{\mu\nu}[\hatA]$:
\begin{align}
F_{\mu\nu}[\hatA] = & \frac{1}{2}\partial_\mu (n^j \tr (n^j A_\nu)) -\frac{1}{2}\partial_\nu(n^j \tr (n^j A_\mu)) + \nonumber\\
&\partial_\mu \frac{i}{4g} [n^j, \partial_\nu n^j] - \partial_\nu \frac{i}{4g} [n^j, \partial_\mu n^j] +\nonumber\\
& \frac{1}{8}[n^j \tr (n^j A_\mu),[n^k, \partial_\nu n^k]]-\frac{1}{8}[n^j \tr (n^j A_\nu),[n^k, \partial_\mu n^k]] +\nonumber\\
& \frac{i}{16g}[ [n^j, \partial_\mu n^j],[n^k, \partial_\nu n^k]]\nonumber\\
=&\frac{1}{2}n^j (\partial_\mu \tr n^j \hatA_\nu - \partial_\nu \tr n^j \hatA_\mu) + \nonumber\\
&\frac{1}{2}\tr (n^j A_\nu) (\partial_\mu n^j - \frac{1}{4} [n^j,[n^k,\partial_\mu n^k]]) - \nonumber\\
&\frac{1}{2}\tr (n^j A_\mu) (\partial_\nu n^j - \frac{1}{4} [n^j,[n^k,\partial_\nu n^k]]) + \nonumber\\
&\frac{i}{16g}(8[\partial_\mu n^j,\partial_\nu n^j] + [[n^j,\partial_\mu n^j],[n^k,\partial_\nu n^k]]).
\end{align}
For the term proportional to $\tr n^j A_\nu$, we can use the Jacobi identity and equation (\ref{eq:YH2}) to write
\begin{align}
\partial_\mu n^j - \frac{1}{4} [n^j,[n^k,\partial_\mu n^k]] = & \partial_\mu n^j +\frac{1}{4} [n^k,[\partial_\mu n^k,n^j]]\nonumber\\
=&\partial_\mu n^j -\frac{1}{4} [n^k,[n^k,\partial_\mu n^j]]\nonumber\\
=& \partial_\mu n^j - \partial_\mu n^j + \frac{1}{2}n^k \tr(n^k \partial_\mu n^j) = \frac{1}{2} n^k\tr(n^k[n^j,\theta\partial_\mu \theta^\dagger])=0
\end{align}
For the term with commutators between $\partial_\nu n$ and $\partial_\mu n$, we use $\partial_\mu n = [n,\theta \partial_\mu \theta^\dagger]$ and expand
\begin{align}
8[\partial_\mu n^j,\partial_\nu n^j] +& [[n^j,\partial_\mu n^j],[n^k,\partial_\nu n^k]] =  8[\partial_\mu n^j,\partial_\nu n^j] +[[n^j,[n^j,\theta\partial_\mu \theta^\dagger]],[n^k,[n^k,\theta\partial_\nu \theta^\dagger]]]\nonumber\\
=& 4(\partial_\mu[n^j,\partial_\nu n^j] - \partial_\nu [n^j,\partial_\mu n^j]) +\nonumber\\
&4 [n^j \tr(n^j \theta\partial_\mu \theta^\dagger) - 2\theta\partial_\mu \theta^\dagger,n^k \tr(n^k \theta\partial_\nu \theta^\dagger) - 2\theta\partial_\nu \theta^\dagger]\nonumber\\
=& 4(\partial_\mu[n^j,[n^j,\theta\partial_\nu \theta^\dagger]] - \partial_\nu [n^j,[n^j,\theta\partial_\mu \theta^\dagger]) +\nonumber\\
&4 [n^j \tr(n^j \theta\partial_\mu \theta^\dagger) - 2\theta\partial_\mu \theta^\dagger,n^k \tr(n^k \theta\partial_\nu \theta^\dagger) - 2\theta\partial_\nu \theta^\dagger]\nonumber\\
=& -8\big(\partial_\mu(n^j \tr(n^j \theta\partial_\nu \theta^\dagger)) -\partial_\nu(n^j \tr(n^j \theta\partial_\mu \theta^\dagger)) + 2\partial_\nu \theta\partial_\mu \theta^\dagger - 2\partial_\mu \theta\partial_\nu \theta^\dagger\big)-\nonumber\\
&8 \big(\tr(n^j \theta\partial_\mu \theta^\dagger) \partial_\nu n^j - \tr(n^j \theta\partial_\nu \theta^\dagger) \partial_\mu n^j - 2[\theta\partial_\mu \theta^\dagger,\theta\partial_\nu \theta^\dagger]\big)\nonumber\\
=&-8 n^j \tr (-\partial_\nu n^j \theta\partial_\mu \theta^\dagger + \partial_\mu n^j \theta\partial_\nu \theta^\dagger - n^j \partial_\nu \theta\partial_\mu \theta^\dagger + n^j \partial_\mu \theta\partial_\nu \theta^\dagger)\nonumber\\
=&-8n^j\tr([n^j,\theta\partial_\mu \theta^\dagger]\theta\partial_\nu\theta^\dagger - [n^j,\theta\partial_\nu \theta^\dagger]\theta\partial_\mu\theta^\dagger + n^j[\theta\partial_\nu\theta^\dagger,\theta\partial_\mu\theta^\dagger])\nonumber\\
=& -8 n^j \tr(n^j[\theta\partial_\mu \theta^\dagger,\theta\partial_\nu \theta^\dagger]).
\end{align}
This gives
\begin{gather}
F_{\mu\nu}[\hatA] = \frac{1}{2}n^j (\partial_\mu \tr n^j A_\nu - \partial_\nu \tr n^j A_\mu) - \frac{i}{2g} n^j\tr(n^j[\theta\partial_\mu \theta^\dagger,\theta\partial_\nu \theta^\dagger]).\label{eq:F1}
\end{gather}
Finally, we use $\tr(n^k [n^j,X])=0$ for any $X$, which follows from $[n^k,n^j] = 0$, to show that
\begin{align}
\tr n^k [\partial_\mu n^j,\partial_\nu n^j] = & \tr(n^k[[n^j,\theta\partial_\mu \theta^\dagger],[n^j,\theta\partial_\nu \theta^\dagger]]\nonumber\\
=&-\tr(n^k[n^j,[\theta\partial_\nu \theta^\dagger,[n^j,\theta\partial_\mu \theta^\dagger]]]) - \tr(n^k[\theta\partial_\nu \theta^\dagger,[[n^j,\theta\partial_\mu \theta^\dagger],n^j]])\nonumber\\
=&-2\tr (n^k [\theta\partial_\nu \theta^\dagger, n^j \tr (n^j \theta\partial_\mu \theta^\dagger)]) + 4 \tr (n^k [\theta\partial_\nu \theta^\dagger,\theta\partial_\mu \theta^\dagger])\nonumber\\
=& 4 \tr (n^k [\theta\partial_\nu \theta^\dagger,\theta\partial_\mu \theta^\dagger]),
\end{align}
which gives the alternative form for the restricted field strength tensor
\begin{gather}
F_{\mu\nu}[\hatA] = \frac{1}{2}n^j (\partial_\mu \tr n^j A_\nu - \partial_\nu \tr n^j A_\mu) + \frac{i}{8g} n^j\tr(n^j[\partial_\mu n^k,\partial_\nu n^k]),\label{eq:F2}
\end{gather}
in agreement with equation (\ref{eq:f2}).
The factors of $1/2$ and $1/8$ depend on our normalisation convention for the generators of the gauge group.
\end{proof}
\subsection{Exponents of the Lie algebra}
In this section, we give a few results which were needed in section \ref{sec:3}.
\begin{lemma}
For non-commuting objects $X$ and $Y$ in some Lie algebra,
\begin{gather}
e^X Y e^{-X} = Y + [X,Y] + \frac{1}{2!} [X,[X,Y]] + \ldots\label{eq:l1b}
\end{gather}
\end{lemma}
\begin{proof}
The result is standard; it follows by expanding $e^X$ in a Taylor series and commuting $X$ term by term through $Y$.
\end{proof}
\begin{lemma}
For non-commuting $X$ and $Y$ in some Lie algebra,
\begin{gather}
e^{X+Y} = e^X e^Y e^{-\frac{1}{2}[X,Y]} e^{\frac{1}{3!}([X,[X,Y]] -2[Y,[X,Y]]) } e^{-\frac{1}{4!}([X,[X,[X,Y]]] +3[[[X,Y],X,Y]] + 3[[[X,Y],Y],Y] )}\ldots\label{eq:l2b}
\end{gather}
\end{lemma}
\begin{proof}
See, for example,~\cite{Magnus:1954}.
\end{proof}
\begin{lemma}
For non-commuting objects $X$ and $Y$ in some Lie algebra, where $X$ is differentiable with respect to a variable $s$ and we can neglect terms of O($\delta s^2$)
\begin{gather}
e^{X(s)} e^{\delta s Y} e^{-X(s+\delta s)} = e^{\delta s Y - \delta s \partial_s X + \delta s \sum_{n>0} \frac{1}{n!}[X^n (Y-\frac{1}{n+1}\partial_s X) ]}\label{eq:l3}
\end{gather}
where we use the notation
\begin{gather}
[X^n Y] \equiv \underbrace{[X,[X,\ldots [X}_{n\text{ terms}},Y]]].
\end{gather}
\end{lemma}
\begin{proof}
Using a Taylor expansion, and the results of equations (\ref{eq:l1b}) and (\ref{eq:l2b}), we have
\begin{align}
e^{X(s)} e^{\delta s Y}& e^{-X(s+\delta s)} = e^{X(s)} (1+\delta s Y) e^{-X(s)-\delta s \partial_s X(s)} + O(\delta s^2)\nonumber\\
=&(1+ \delta s Y +\sum_{n>0} \frac{1}{n!}[X^n \delta s Y])e^{-\delta s \partial_s X(s)} e^{-\frac{1}{2}\delta s [X,\partial_s X]} e^{-\frac{1}{3!}\delta s [X^2,\partial_s X]} \ldots + O(\delta s^2)\nonumber\\
=& e^{\delta s Y +\sum_{n>0} \frac{1}{n!}[X^n \delta s Y] - \delta s \partial_s X - \sum_{n>1} \frac{1}{n!}[X^{n-1} \delta s \partial_s X]} + O(\delta s^2)
\end{align}
where $X \equiv X(s)$. Equation (\ref{eq:l3}) follows immediately.
\end{proof}
\begin{lemma}
For $\phi,X \in su(2)$,
\begin{gather}
X + \sum_{n>0} \frac{1}{n!}[(i\phi)^n X] = \cos \sqrt{2 \tr \phi^2}\left(X - \phi \frac{\tr \phi X}{\tr \phi^2}\right) + \frac{\sin \sqrt{2\tr\phi^2}}{\sqrt{2 \tr \phi^2}}[\phi,X] + \phi \frac{\tr \phi X}{\tr \phi^2}\label{eq:lemma:su2}
\end{gather}
\end{lemma}
\begin{proof}
We may always parametrise $\phi = a \chi^\dagger \lambda^3 \chi$ for $\chi \in SU(2)/U(1)$ and real $a$. This means that equation (\ref{eq:YH2}) is applicable, and we have
\begin{gather}
[i\phi,[i\phi,X]] = 2 \phi \tr (\phi X) - 2 \tr \phi^2 X.
\end{gather}
Therefore, for integer $n>0$,
\begin{align}
[(i\phi)^{2n}X] = & (-2\tr \phi^2)^n X + (-2 \tr \phi^2)^{n-1} 2\phi \tr (\phi X)\nonumber\\
[(i\phi)^{2n+1}X]=& i(-2\tr \phi^2)^n[\phi,X].
\end{align}
Therefore
\begin{multline}
X + [i\phi,X] + \sum_{n>0}\left[\frac{1}{(2n)!}[(i\phi)^{2n}X] + \frac{1}{(2n+1)!}[(i\phi)^{2n+1}X] \right] = \\
\cos \sqrt{2 \tr \phi^2} \left(X - \phi \frac{\tr (\phi X)}{{ \tr \phi^2}}\right) + i\frac{\sin\sqrt{2 \tr \phi^2}}{\sqrt{2 \tr \phi^2}} [\phi,X] + \phi \frac{\tr (\phi X)}{{ \tr \phi^2}},
\end{multline}
which is the result we wanted to prove.
\end{proof}

\section{Area Law and the String Tension}\label{app:B}
We expect that the average number of the \mons contained within a loop will be proportional to the area of the loop. Given that different configurations have statistically independent numbers of \monsp, we may model the distribution of the number, $n_j$, of \mons of a type $j$ (i.e. \mons constructed from $\lambda^j$), within the Wilson Loop across the configurations using a Poisson distribution with a mean value $\mu_j$ which is the density of \mons multiplied by the loop area. Each \mon (or \mon-anti-\mon pair) will contribute a certain value $\chi^j \lambda^j$ to the Wilson Loop, and this will be distributed according to some distribution $P(\chi^j)$ which is an even function of $\chi^j$ as positive winding number and negative winding number \mons are equally likely. Neglecting any perimeter contribution from non-\mon effects, the expectation value of the Wilson line will be
\begin{align}
 \tr \left(\langle\theta^\dagger W \theta\right)\rangle = &\tr\prod_j \sum_{n_j} \frac{\mu_j^{n_j}}{n_j!} e^{-\mu_j} \left( \int d\chi^k P(\chi^3,\chi^8,\ldots) e^{i\lambda^k \chi^k}\right)^{n_j}\nonumber\\
=&\tr \prod_j e^{-\mu_j \left(1- \int d\chi^k P(\chi^3,\chi^8,\ldots) e^{i \lambda^k \chi^k}\right)}.
\end{align}
We may expand the exponential $e^{i\lambda^j \chi^j}$ in terms of trigonometric functions of $\chi^j$. For example, in SU(2),
\begin{gather}
e^{i\lambda^3 \chi^3} = \cos \chi^3 + i \sin \chi^3 \lambda^3,
\end{gather}
and in SU(3),
\begin{multline}
e^{i\lambda^3 \chi^3 + i \lambda^8 \chi^8} = \frac{1}{3}\left(2 \cos \frac{\chi^8}{\sqrt{3}} \cos \chi^3 + \cos \frac{2\chi^8}{\sqrt{3}} + 2 i \sin \frac{\chi^8}{\sqrt{3}} \cos\chi^3 -i \sin \frac{2 \chi^8}{\sqrt{3}}\right)\\
+\lambda^3 \left(-\sin\frac{\chi^8}{\sqrt{3}} \sin \chi^3 + i \cos \frac{\chi^8}{\sqrt{3}} \sin \chi^3\right) + \\
\frac{\lambda^8}{\sqrt{3}} \left(\cos \frac{\chi^8}{\sqrt{3}} \cos \chi^3 - \cos \frac{2\chi^8}{\sqrt{3}} +  i \sin \frac{\chi^8}{\sqrt{3}} \cos\chi^3 +i \sin \frac{2 \chi^8}{\sqrt{3}}\right) .
\end{multline}
In each case, the imaginary terms are eliminated by the integral over $\chi$, so as long as the distiribution $P$ is not too strongly peaked at $\chi = 0$, we are left with a result which resembles $e^{-\mu \alpha_0 - \mu \alpha^j \lambda^j}$, where $\mu$ is proportional to the area of the loop (and depends on the number of structures) while $\alpha$, which represents the quantity obtained after the integral over $\chi$, is independent of the area (and depends a single structure); each diagonal term exponentially decreases with the area of the loop and the trace will be dominated by whichever term decreases the slowest. This gives an area law scaling for the Wilson Loop.


\section{SU(3) \Mon algebra}\label{app:su3}
 The following commutation relations are useful in evaluating the product of $\theta$ matrices in SU(3). $\phi_1$, $\phi_2$, $\ldots$ are defined in equation (\ref{eq:39z})
\begin{align}
[\phi_1,\lambda^3] =& 2 i \bar{\phi_1} & [\phi_1,[\phi_1,\lambda^3]] = & 4\lambda^3\nonumber\\
[\phi_1,\sqrt{3} \lambda^8] =& 0  & [\phi_1,[\phi_1,\sqrt{3}\lambda^8]] = &0 \nonumber\\
[\phi_1,\phi_2] =& -i\cos(c_1+c_3-c_2)\bar{\phi}_3 - i \sin(c_1+c_3-c_2)\phi_3  & [\phi_1,[\phi_1,\phi_2]] = &\phi_2 \nonumber\\
[\phi_1,\phi_3] =& -i\cos(c_1+c_3-c_2)\bar{\phi}_2+i\sin(c_1+c_3-c_2)\phi_2  & [\phi_1,[\phi_1,\phi_3]] = & \phi_3\nonumber\\
[\phi_1,\bar{\phi}_1] =& -2i\lambda^3  & [\phi_1,[\phi_1,\bar{\phi}_1]] = &4 \bar{\phi}_1 \nonumber\\
[\phi_1,\bar{\phi}_2] =& i\cos(c_1+c_3-c_2)\phi_3 - i \sin(c_1+c_3-c_2)\bar{\phi}_3   & [\phi_1,[\phi_1,\bar{\phi}_2]] = &\bar{\phi}_2 \nonumber\\
[\phi_1,\bar{\phi}_3] =&i\cos(c_1+c_3-c_2)\phi_2 +i \sin(c_1+c_3-c_2)\bar{\phi}_2   & [\phi_1,[\phi_1,\bar{\phi}_3]] = & \bar{\phi_3}\nonumber\\
[\phi_2,\lambda^3] =& i\bar{\phi}_2&[\phi_2,[\phi_2,\lambda^3]] = &2 \frac{\sqrt{3} \lambda^8 + \lambda^3}{2} \nonumber\\
[\phi_2,\sqrt{3}\lambda^8] =& 3i\bar{\phi}_2  & [\phi_2,[\phi_2,\sqrt{3}\lambda^8]] = & 6 \frac{\sqrt{3} \lambda^8 + \lambda^3}{2}\nonumber\\
[\phi_2,\phi_1] =& i \cos(c_1+c_3-c_2)\bar{\phi}_3 + i \sin(c_1+c_3-c_2){\phi}_3 & [\phi_2,[\phi_2,\phi_1]] = &\phi_1 \nonumber\\
[\phi_2,\phi_3] =& -i\cos(c_3+c_1-c_2)\bar{\phi}_1 - i \sin(c_1+c_3-c_2)\phi_1  & [\phi_2,[\phi_2,\phi_3]] = & \phi_3\nonumber\\
[\phi_2,\bar{\phi}_1] =& i\cos(c_1+c_3-c_2)\phi_3 + i \sin(c_2 - c_1-c_3)\bar{\phi_3}  & [\phi_2,[\phi_2,\bar{\phi}_1]] = & \bar{\phi}_1\nonumber\\
[\phi_2,\bar{\phi}_2] =& -2i \frac{\lambda^3 + \sqrt{3}\lambda^8}{2} & [\phi_2,[\phi_2,\bar{\phi}_2]] = & 4\bar{\phi}_2\nonumber\\
[\phi_2,\bar{\phi}_3] =&  -i\cos(c_1+c_3-c_2)\phi_1 + i \sin(c_1+c_3-c_2)\bar{\phi}_1  & [\phi_2,[\phi_2,\bar{\phi}_3]] = &\bar{\phi}_3 \nonumber\\
[\phi_3,\lambda^3] =& - i \bar{\phi_3}  & [\phi_3,[\phi_3,\lambda^3]] = & 2 \frac{\lambda^3 - \sqrt{3}\lambda^8}{2}\nonumber\\
[\phi_3,\sqrt{3}\lambda^8] =& 3i\bar{\phi_3}  & [\phi_3,[\phi_3,\sqrt{3} \lambda^8]] = & -6 \frac{\lambda^3 - \sqrt{3}\lambda^8}{2}\nonumber\\ 
[\phi_3,\phi_1] =& i \cos(c_1+c_3-c_2)\bar{\phi}_2 - i \sin(c_1+c_3-c_2)\phi_2  & [\phi_3,[\phi_3,\phi_1]] = & \phi_1\nonumber\\
[\phi_3,\phi_2] =& i\cos(c_1+c_3-c_2)\bar{\phi}_1 + i \sin(c_1+c_3-c_2 )\phi_1 & [\phi_3,[\phi_3,\phi_2]] = & \phi_2\nonumber\\
[\phi_3,\bar{\phi}_1] =& -i \cos(c_3+c_1-c_2){\phi}_2 - i \sin(c_3+c_1-c_2)\bar{\phi}_2  & [\phi_3,[\phi_3,\bar{\phi}_1]] = &\bar{\phi}_1 \nonumber\\
[\phi_3,\bar{\phi}_2] =& -i\cos(c_1+c_3-c_2)\phi_1 + i \sin(c_1+c_3-c_2)\bar{\phi}_1  & [\phi_3,[\phi_3,\bar{\phi}_2]] = &\bar{\phi}_2 \nonumber\\
[\phi_3,\bar{\phi}_3] =& 2i \frac{\lambda^3 - \sqrt{3}\lambda^8}{2}  & [\phi_3,[\phi_3,\bar{\phi}_3]] = &4 \bar{\phi}_3\label{eq:31}.
\end{align}
We obtain, using equation (\ref{eq:l1b}), the definition of $\theta$ given in equation (\ref{eq:l1}) and the commutation relations listed in equation (\ref{eq:31}),
\begin{align}
n^3 =& e^{i\phi_3 a_3} e^{i\phi_2 a_2} e^{i \phi_1 a_1} \lambda^3 e^{-i \phi_1 a_1}e^{-i\phi_2 a_2}e^{-i\phi_3 a_3}\nonumber\\
=&e^{i\phi_3 a_3} e^{i\phi_2 a_2}\bigg[\lambda^3 \cos(2a_1) - \bar{\phi}_1\sin(2a _1)\bigg]e^{-i\phi_2 a_2}e^{-i\phi_3 a_3}\nonumber\\
=&e^{i\phi_3 a_3}\bigg[\cos(2a_1)\left(-\frac{\sin(2a_2)}{2} \bar{\phi}_2 + \frac{\cos (2a_2)-1}{4}\sqrt{3}\lambda^8 + \frac{\cos(2a_2) + 3}{4} \lambda^3\right) - \nonumber\\
&\sin (2a_1)(\cos(a_2) \bar{\phi}_1 - \sin a_2 (\phi_3 \cos(c_2 - c_1 - c_3) + \sin(c_2-c_1-c_3) \bar{\phi}_3)) \bigg]e^{-i\phi_3 a_3}\nonumber\\
=&-\cos(2 a _1) \frac{\sin 2a_2}{2}\sin a_3(\phi_1 \cos(c_2-c_3-c_1) - \bar{\phi}_1\sin(c_1+c_3-c_2))- \nonumber\\
&\cos(2a_1) \frac{\sin 2 a_2}{2} \cos a_3 \bar{\phi}_2+\nonumber\\
&\cos(2a_1) \frac{\cos 2a_2-1}{4} \left( - 3\frac{\sin(2a_3)}{2} \bar{\phi}_3 + \sqrt{3} \lambda^8 \frac{1+3\cos(2a_3)}{4} + \frac{3}{4}\lambda^3 (1-\cos 2a_3)\right)+\nonumber\\
&\cos(2a_1) \frac{\cos 2a_2+3}{4} \left(\lambda^3 \frac{\cos(2a_3) +3}{4}+ \sqrt{3}\lambda^8 \frac{1-\cos(2a_3)}{4} + \frac{\sin(2a_3)}{2} \bar{\phi}_3 \right)-\nonumber\\
&\sin (2a_1)\cos a_2 (\bar{\phi}_2 \sin a_3 \sin(c_3+c_1-c_2) + \phi_2 \sin a_3 \cos(c_3+c_1-c_2) + \bar{\phi}_1 \cos a_3)+\nonumber\\
&\sin(2a_1) \sin a_2 \cos(c_2 - c_1-c_3) \phi_3 + \nonumber\\
&\sin (2a_1) \sin a_2 \sin(c_2-c_1-c_3)\left(\cos(2a_3) \bar{\phi}_3 - \frac{\sin (2a_3)}{2}(\lambda^3 - \sqrt{3} \lambda^8)\right).
\end{align}
We also obtain,
\begin{align}
\sqrt{3} n^8 =& e^{ia_3\phi_3}e^{ia_2\phi_2} \sqrt{3}\lambda^8 e^{-ia_2\phi_2}e^{-i \phi_3 a_3} \nonumber\\
=&e^{ia_3 \phi_3}\bigg[\lambda^3 \frac{3 \cos (2a_2) - 3}{4} + \sqrt{3}\lambda^8 \frac{3 \cos(2a_2)+1}{4}- 3\frac{\sin(2a_2)}{2} \bar{\phi}_2 \bigg] e^{-ia_3 \phi_3}\nonumber\\
=&-\frac{3}{2} \sin(2a_2) \sin(a_3)(\phi_1 \cos(c_2-c_3-c_1) - \bar{\phi}_1 \sin(c_1+c_3-c_2)) -\frac{3}{2}\sin(2a_2) \bar{\phi}_2 \cos(a_3)+\nonumber\\
&\frac{3 \cos(2a_2) - 3}{4} (\lambda^3 \frac{\cos(2a_3)+3}{4} + \sqrt{3} \lambda^8 \frac{1-\cos(2a_3)}{4} + \frac{\sin (2a_3)}{2} \bar{\phi}_3) + \nonumber\\
&\frac{3 \cos(2a_2) + 1}{4} (-\frac{3}{2} \sin(2a_3)\bar{\phi}_3 + \sqrt{3} \lambda_8 \frac{1+3\cos (2a_3)}{4} + \frac{3}{4}\lambda_3(1-\cos (2a_3))).
\end{align}
A similar calculation gives
\begin{align}
\theta^\dagger \partial_\mu \theta = & S_{1\mu}+S_{2\mu}+S_{3\mu}+S_{4\mu}, \label{eq:D1}
\end{align}
where
\begin{align}
S_{1\mu} =& e^{-id_3 \lambda^3 -i d_8 \lambda^8} \partial_\mu e^{id_3 \lambda^3 +i d_8 \lambda^8}\nonumber\\
=&i \lambda^3 \partial_\mu d_3 + i \lambda^8 \partial_\mu d_8,
\end{align}
\begin{align}
S_{2\mu} =& e^{-id_3 \lambda^3 -i d_8 \lambda^8}e^{-ia_1\phi_1} \partial_\mu (e^{ia_1\phi_1})e^{id_3 \lambda^3 +i d_8 \lambda^8}\nonumber\\
=&e^{-id_3 \lambda^3 -i d_8 \lambda^8}\bigg[i\phi_1 \partial_\mu a_1 + i \partial_\mu c_1(\cos a_1 \sin a_1 \bar{\phi_1} - \sin^2 a_1 \lambda^3) \bigg] e^{id_3 \lambda^3 +i d_8 \lambda^8},
\end{align}
\begin{align}
S_{3\mu} =& e^{-id_3 \lambda^3 -i d_8 \lambda^8}e^{-ia_1\phi_1} e^{-ia_2\phi_2} \partial_\mu (e^{ia_2\phi_2}) e^{ia_1 \phi_1}e^{id_3 \lambda^3 +i d_8 \lambda^8}\nonumber\\
=&e^{-id_3 \lambda^3 -i d_8 \lambda^8}e^{-i\phi_1 a_1}\nonumber\\
&\phantom{space}\bigg[\partial_\mu a_2 i \phi_2 + i \partial_\mu c_2 (\cos(a_2)\sin(a_2) \bar{\phi}_2 - \frac{\sin^2 a_2}{2} (\lambda^3 + \sqrt{3} \lambda^8))  \bigg]e^{i\phi_1 a_1}e^{id_3 \lambda^3 +i d_8 \lambda^8}\nonumber\\
=&e^{-id_3 \lambda^3 -i d_8 \lambda^8}\bigg[-i \partial_\mu c_2 \frac{\sin^2 a_2}{2}(\sqrt{3}\lambda^8 - \cos(2a_1)\lambda^3 - \sin(2a_1)\bar{\phi}_1)+\nonumber\\
&i \partial_\mu a_2(\cos a_1 \phi_2 - \sin a_1 (\cos(c_2 - c_1-c_3) \bar{\phi}_3 - \sin(c_2 - c_1 - c_3) \phi_3))+\nonumber\\
&i \partial_\mu c_2 \frac{\sin 2a_2}{2} (\cos a_1 \bar{\phi}_2 + \sin a_1(\cos(c_2 - c_1 - c_3)\phi_3 + \sin(c_2 - c_1 - c_3)\bar{\phi}_3))\bigg]e^{id_3 \lambda^3 +i d_8 \lambda^8},
\end{align}
and
\begin{align}
S_{4\mu} =& e^{-id_3 \lambda^3 -i d_8 \lambda^8}e^{-ia_1\phi_1} e^{-ia_2\phi_2}e^{-ia_3\phi_3} \partial_\mu (e^{ia_3\phi_3})e^{ia_2\phi_2} e^{ia_1 \phi_1}e^{id_3 \lambda^3 +i d_8 \lambda^8}\nonumber\\
=&e^{-id_3 \lambda^3 -i d_8 \lambda^8}e^{-ia_1\phi_1} e^{-ia_2\phi_2}\bigg[ \partial_\mu a_3 i \phi_3 + i \partial_\mu c_3 \left(\cos a_3 \sin a_3 \bar{\phi}_3 - \frac{\sin^2a_3}{2} (\lambda^3 - \sqrt{3}\lambda^8)\right)\bigg]\nonumber\\
& e^{ia_2\phi_2} e^{ia_1 \phi_1}e^{id_3 \lambda^3 +i d_8 \lambda^8}\nonumber\\
=&e^{-id_3 \lambda^3 -i d_8 \lambda^8}e^{-ia_1\phi_1}\nonumber\\
&\bigg[i\partial_\mu a_3\left(\phi_3 \cos(a_2) - \sin(a_2)(\bar{\phi}_1 \cos(c_2 - c_3 - c_1) + \phi_1 \sin(c_1+c_3-c_2))\right) + \nonumber\\
&i \partial_\mu c_3\cos a_3 \sin a_3\left( \bar{\phi}_3 \cos a_2 - \sin a_2 (\phi_1 \cos(c_1+c_3-c_2) - \bar{\phi}_1 \sin(c_1 + c_3 - c_2) \right)-\nonumber\\
&i \partial_\mu c_3 \frac{\sin^2 a_3}{2} \left(\lambda^3 - \sqrt{3} \lambda^8 - (\cos(2a_2) - 1) (\sqrt{3} \lambda^8 + \lambda^3) - \sin(2a_2) \bar{\phi}_2\right)
\bigg] \nonumber\\
&e^{ia_1 \phi_1}e^{id_3 \lambda^3 +i d_8 \lambda^8}\nonumber\\
=&e^{-id_3 \lambda^3 -i d_8 \lambda^8}\nonumber\\
&\bigg[i\partial_\mu a_3 \cos a_2 (\cos a_1 \phi_3 + \sin a_1(-\bar{\phi}_2 \cos(c_1 - c_2 + c_3) + {\phi}_2 \sin(c_1 - c_2 + c_3)))-\nonumber\\
&i \partial_\mu a_3 \sin a_2 \cos(c_2 - c_3 - c_1)(\bar{\phi}_1 \cos(2a_1) - \lambda^3 \sin(2a_1))-\nonumber\\
&i \partial_\mu a_3 \sin(a_2) \sin(c_1+c_3-c_2) \phi_1+\nonumber\\
&i\partial_\mu c_3 \frac{\sin 2a_3}{2} \cos a_2 (\cos a_1 \bar{\phi}_3 + \sin a_1(\phi_2 \cos(c_1 - c_2 + c_3) + \bar{\phi}_2 \sin(c_1 - c_2 + c_3)))-\nonumber\\
&i\partial_\mu c_3 \frac{\sin(2a_3)}{2}\sin a_2 \phi_1 \cos(c_2 - c_3 - c_1)+\nonumber\\
&i\partial_\mu c_3 \frac{\sin (2a_3)}{2} \sin a_2 \sin(c_1 + c_3 - c_2)(\cos (2a_1) \bar{\phi}_1 - \sin 2a_1 \lambda^3) - \nonumber\\
&i\partial_\mu c_3 \frac{\sin^2 a_3}{2}\left( (2-\cos(2a_2))(\cos(2a_1)\lambda^3 + \sin 2a_1 \bar{\phi}_1) - \cos(2a_2)\sqrt{3} \lambda^8 \right) -\nonumber\\
 &i \partial_\mu c_3 \frac{\sin^2 a_3}{2} \sin(2a_2) \left(\cos(a_1) \bar{\phi}_2 + \sin a_1 \left(\phi_3\cos(c_1 + c_3 - c_2) - \bar{\phi}_3(c_1 + c_3 - c_2) \right)\right) \bigg]\nonumber\\
& e^{id_3 \lambda^3 +i d_8 \lambda^8}.\label{eq:D2}
\end{align}

\section{Numerical Methods}\label{app:C}

\subsection{Algorithm for fixing $\theta$}
It is simplest to find $\theta$ by solving for the eigenvectors of the Wilson Loop; however we employed a different technique (the method discussed below has the advantage that it can be easily adapted to other rules for choosing $\theta$ discussed in the literature than just the one we use in this paper and to other problems such as gauge fixing and solution of the defining equations. It is also reasonably quick and efficient.).

We wish to find the $\theta$ that solves
\begin{gather}
[\theta^\dagger_x U_{x,\mu} \theta_{x+\hat{\mu}} ,\lambda^j] = 0,
\end{gather}
for each diagonal $\lambda^j$ and gauge link along the Wilson Loops. This is equivalent to solving
\begin{align}
0 = - \tr[(U_{x,\mu} n^j_{x+\hat{\mu}} - n^j_x U_{x,\mu})(U^\dagger_{x,\mu} n^j_{x} - n^j_{x+\hat{\mu}} U^\dagger_{x,\mu})],
\end{align}
which is solved by the $\theta$ which minimises
\begin{gather}
F[\theta] =\sum_j\left[ 4V -  \text{Re}\;\tr(U_{x,\mu} n^j_{x+\hat{\mu}} U^\dagger_{x,\mu} n^j_{x})\right],
\end{gather}
 with
\begin{gather}
n^j_x = \theta \lambda^j \theta^\dagger
\end{gather}
and $\lambda^j = \lambda^3$ or $\lambda^8$ and $\theta \in SU(3)$ (in fact, the solution we want is at $F[\theta] = 0$). The trace is
over both lattice sites and SU(3) indices and $V$ is the lattice volume.

We solve this minimisation problem by combining two different algorithms; one based on molecular dynamics, which is an extension of the steepest descent method, and the second a Newton-Raphson algorithm. The goal is to have the molecular dynamics solve the system to a sufficient accuracy that the Newton-Raphson algorithm will converge, and then Newton-Raphson rapidly finds the solution to a high precision.
 The combination of these two methods, while it may not be the most efficient algorithm, worked well on all the lattice volumes we used. We also employed similar algorithms to gauge fix and solve the defining equations of the Abelian decomposition.

Having found $\theta$, one must subsequently order the columns of $\theta$ and fix the U(1) phases according to the chosen fixing condition. We ordered the columns according to decreasing real part of eigenvalue, and either fixed $\theta$ using $d=0$ or by having the same value for each $\hat{u}$ around the Wilson Loop, depending on the observable.
\subsubsection{Molecular dynamics}
For the molecular dynamics algorithm, we introduce a momentum field
$\pi(x)$ conjugate to $\theta$ and define a fictitious energy
\begin{gather}
E = \frac{1}{2}\tr \pi^2 + F[\theta].
\end{gather}
$\pi$ here represents a traceless Hermitian matrix on each lattice site.

We evolve $\theta$ in a fictitious computer time $\tau$ using
\begin{gather}
\frac{d}{d\tau}\theta_x = i  \theta_x \pi_x,
\end{gather}
or, as $\delta \tau \rightarrow 0$,
\begin{gather}
\theta_x(\tau + \delta \tau) =  \theta_x(\tau) e^{i\delta \tau \pi_x}.
\end{gather}
The second equation of motion is taken from the conservation of energy.
\begin{align}
0 =& \tr \pi \frac{d}{d\tau} \pi + \frac{d}{d\tau} F\nonumber\\
\frac{d}{d\tau}F = & \text{Re} \tr (\pi(x) G(x))\nonumber\\
G'_x =& i \left[ \theta_x^\dagger U_{x,\mu} \theta_{x+\hat{\mu}} \lambda^j \theta^\dagger_{x+\hat{\mu}} U^\dagger_{x,\mu} \theta_x \lambda^j +
\theta_x^\dagger U^\dagger_{x-\hat{\mu},\mu} \theta_{x-\hat{\mu}} \lambda^j \theta^\dagger_{x-\hat{\mu}} U^\dagger_{x-\hat{\mu},\mu} \theta_x \lambda^j\right]-\nonumber\\
&i\left[\lambda^j \theta^\dagger_x U^\dagger_{x-\hat{\mu},\mu} \theta_{x-\hat{\mu}}\lambda^j \theta^\dagger_{x-\hat{\mu}} U_{{x-\hat{\mu}},\mu} \theta_x + \lambda^j \theta^\dagger_x U_{x,\mu} \theta_{x+\hat{\mu}}\lambda^j \theta^\dagger_{x+\hat{\mu}} U^\dagger_{{x},\mu} \theta_x \right].
\end{align}
$G$ is taken to be the hermitian, traceless part of $G'$, and in our particular problem (since the $\lambda^3$ and $\lambda^8$ components of $\theta$ are arbitrary) we also set the diagonal components of $G$ to be zero.
We then evolve $\theta$ and $\pi$ using the Omelyan integrator~\cite{Takaishi:2005tz,Omelyan} for a time step $\delta \tau$
\begin{align}
\pi_x \rightarrow& \pi_x - \alpha \delta \tau G_x\nonumber\\
\theta_x \rightarrow& \theta_x e^{i \frac{\delta \tau}{2} \pi_x}\nonumber\\
\pi_x \rightarrow& \pi_x - (1-2\alpha) \delta \tau G_x\nonumber\\
\theta_x \rightarrow& \theta_x e^{i \frac{\delta \tau}{2} \pi_x}\nonumber\\
\pi_x \rightarrow& \pi_x - \alpha \delta \tau G_x.
\end{align}
This conserves energy up to O($\delta \tau^2$). $\alpha$ should be tuned to have the best conservation of energy for a given $\delta \tau$, while $\delta \tau$ should be the maximum value where the algorithm remains stable (i.e. energy is conserved). We repeat this integrator a number of times to form a trajectory, ending the trajectory when $F[\theta]$ stops decreasing. We select the $\theta$ which had the lowest $F[\theta]$ over the course of the trajectory, and repeat until $F[\theta]$ is below a target threshold.

The trick is to start each trajectory with $\pi = 0$. This means, at least for the first few steps, $\tr \pi^2 $ must increase, and since energy is conserved, $F$ must decrease. While this method cannot be used to find a precise solution for the minimisation of $F[\theta]$, as the force $G$ becomes zero close to the minimum and the method slows down, it rapidly converges towards a solution, and gets close enough to allow the use of more powerful methods. The initial step is equivalent to the steepest descent method; but as the trajectory progresses we incorporate information from the forces at larger distances from the lattice site.

The exponential of the anti-hermitian $3\times 3$ matrix is calculated using
the Caley-Hamilton theorem, following~\cite{Morningstar:2003gk}.
\subsubsection{Newton-Raphson}
Once we have found the minimum to a sufficient accuracy, we employ a Newton
Raphson algorithm to polish it to an arbitrary accuracy.

We define a force according to
\begin{align}
G^a_x =& \frac{\partial F[\theta]}{\partial \theta_x} \theta_x\lambda^a
\nonumber\\
=& - \text{Re}\;\tr i\left[U_{x-\hat{\mu},\mu} \theta_x [\lambda^a,\lambda^j]\theta^\dagger_x U^\dagger_{x-\hat{\mu},\mu} \theta_{x-\hat{\mu}} \lambda^j\theta^\dagger_{x-\hat{\mu}}\right]\nonumber\\
&- \text{Re}\;\tr i\left[U_{x,\mu} \theta_{x+\hat{\mu}} \lambda^j\theta^\dagger_{x+\hat{\mu}} U^\dagger_{x,\mu} \theta_{x}[\lambda^a, \lambda^j]\theta^\dagger_{x}\right].
\end{align}
This force is related to the molecular dynamics force above by $G^a_x = \frac{1}{2} \tr(\lambda^a G_x)$, and is essentially a vector representation of the force rather than a matrix representation.
The force gradient is given by differentiating $F$ twice
\begin{align}
H^{a,b}_{x,y} =& -\delta_{x,y}\text{Re}\;\tr i\left[U_{x-\hat{\mu},\mu} \theta_x [\lambda_b,[\lambda^a,\lambda^j]]\theta^\dagger_x U^\dagger_{x-\hat{\mu},\mu} \theta_{x-\hat{\mu}} \lambda^j\theta^\dagger_{x-\hat{\mu}}\right]\nonumber\\
& -\delta_{x,y}\text{Re}\;\tr i\left[U_{x,\mu} \theta_{x+\hat{\mu}} \lambda^j\theta^\dagger_{x+\hat{\mu}} U^\dagger_{x,\mu} \theta_{x} [\lambda_b,[\lambda^a,\lambda^j]]\theta^\dagger_{x}\right]\nonumber\\
&-\delta_{x-\hat{\mu},y}\text{Re}\;\tr i\left[U_{x-\hat{\mu},\mu} \theta_x [\lambda^a,\lambda^j]\theta^\dagger_x U^\dagger_{x-\hat{\mu},\mu} \theta_{x-\hat{\mu}}[\lambda^b, \lambda^j]\theta^\dagger_{x-\hat{\mu}}\right]\nonumber\\
&-\delta_{x+\hat{\mu},y}\text{Re}\;\tr i\left[U_{x,\mu} \theta_{x+\hat{\mu}} [\lambda^b,\lambda^j]\theta^\dagger_{x+\hat{\mu}} U^\dagger_{x,\mu} \theta_{x}[\lambda^a, \lambda^j]\theta^\dagger_{x}\right].
\end{align}
We then update $\theta$ according to
\begin{gather}
\theta_x \rightarrow \theta_x e^{-i\lambda^a (H^{a,b}_{x,y})^{-1} G^b_y}.
\end{gather}
In this particular problem, the sum over $a$ and $b$ excludes the indices related to the diagonal Gell Mann matrices (where $G$ is zero, and this removes the zero eigenvalues from $H$). $H^{a,b}_{x,y}$ can be easily inverted using a CGNE algorithm with odd/even preconditioning.

This is equivalent to the usual Newton-Raphson minimisation algorithm, and has the usual advantages and disadvantages associated with Newton-Raphson. It has superlinear convergence as long as $H^{a,b}_{x,y}$ contains no small eigenvalues and we start close enough to the desired solution. Restricting the sum over $\lambda^a$ to the off-diagonal elements ensures that the inversion of $H^{a,b}_{x,y}$ is smooth enough, and we did not encounter any difficulties: $F[\theta]$ was reduced to zero within machine precision in two or three iterations.

Once we have found a $\theta$ for which $F=0$, it is straightforward to fix the ordering of the eigenvalues of the Wilson Loop and the $U(N_C-1)$ phase according to the choice of fixing condition. We sorted the eigenvectors according to their real part, and fixed the phase so that $\theta^\dagger U \theta$ is equal for all the links along the Wilson Line (so under a gauge transformation $\theta \rightarrow \Lambda \theta$ without any need for any additional gauge fixing).

\subsection{Numerical solution of the defining equations}
Given $n_j$, we wish to find the solution to the lattice defining equations
\begin{align}
\hat{U}_{x,\mu} n^j_{x+\mu} \hat{U}^\dagger_{x,\mu} - n^j_x = & 0\nonumber\\
i\tr (n^j (\hat{X}_\mu - (\hat{X}_\mu)^\dagger)) = & 0,
\end{align}
for $\hat{U}_{x,\mu} = U_{x,\mu} \hat{X}^\dagger_{x,\mu}$. The solution of these equations is given by the $\hat{X} \in SU(N_C)$ which minimises the functional
\begin{align}
g[\hat{X}_\mu] =& g_1[\hat{X}_\mu] + g_2[\hat{X}_\mu]\nonumber\\
g_2[\hat{X}_\mu]=&\sum_j\left[-(\tr(n^j_x (\hat{X}_{x,\mu} -
\hat{X}_{x,\mu}^\dagger)))(\tr(n^j_x (\hat{X}_{x,\mu} - \hat{X}^\dagger_{x,\mu})))\right]
\nonumber\\
g_1[\hat{X}_\mu] = & \sum_j \tr\left[ (U_{x-\mu,\mu}\hat{X}^\dagger_{x,\mu} n_x^j \hat{X}_{x,\mu}
U^\dagger_{x-\mu,\mu} -
n_{x-\mu}^j)^2\right]
\end{align}
for each direction $\mu$. Each $\hat{X}_\mu$ can be solved independently, although in practice we set up our algorithm to solve them all simultaneously. Clearly, $g[\hat{X}] \ge 0$, with the equality achieved
when the defining equations are satisfied. If there is an exact solution to the defining equations, this procedure will find it. If not, we at least achieve something close to the lattice solution. In practice, by varying the initial guess for $\hat{X}$ we found that there were usually three or four solutions for each $\hat{X}_\mu$. We choose the solution with the largest  $\tr \hat{X}$, a condition which is both gauge invariant and satisfied on those links along the Wilson Loops used to define $\theta$, where $U = \hat{U}$ and $\hat{X} = 1$.

The same minimisation routine is used as in the previous section (solving for $\hat{X}$), combining molecular dynamics with Newton-Raphson minimisation, and while we make no claims that this is the most efficient algorithm, it proved to be efficient enough on our configurations. There were no issues concerning algorithmic stability or convergence.

 The Newton-Raphson force and force Gradient are
\begin{align}
G^a_x = & - 2i (\tr(n^j_x (\hat{X}_{x,\mu} -\hat{X}_{x,\mu}^\dagger)))
(\tr(n^j_x (\hat{X}_{x,\mu}\lambda^a +\lambda^a \hat{X}_{x,\mu}^\dagger)))+\nonumber\\
&2i\tr\left[\left(U_{x-\mu,\mu}\hat{X}^\dagger_{x,\mu} n_x^j \hat{X}_{x,\mu}
U^\dagger_{x-\mu,\mu} -
n_{x-\mu}^j\right)\right.\nonumber\\
&\phantom{space}\left.\left(U_{x-\mu,\mu}\hat{X}^\dagger_{x,\mu} n_x^j \hat{X}_{x,\mu}\lambda^a
U^\dagger_{x-\mu,\mu} - U_{x-\mu,\mu}\lambda^a\hat{X}^\dagger_{x,\mu} n_x^j \hat{X}_{x,\mu}
U^\dagger_{x-\mu,\mu} \right)\right]\nonumber\\
H^{a,b}_{x,x} = & 2 (\tr(n^j_x (\hat{X}_{x,\mu}\lambda^b +\lambda^b\hat{X}_{x,\mu}^\dagger)))
(\tr(n^j_x (\hat{X}_{x,\mu}\lambda^a +\lambda^a \hat{X}_{x,\mu}^\dagger))) + \nonumber\\
 &2 (\tr(n^j_x (\hat{X}_{x,\mu} -\hat{X}_{x,\mu}^\dagger)))
(\tr(n^j_x (\hat{X}_{x,\mu}\lambda^b\lambda^a -\lambda^a\lambda^b \hat{X}_{x,\mu}^\dagger))) - \nonumber\\
&2\tr\left[\left(U_{x-\mu,\mu}\hat{X}^\dagger_{x,\mu} n_x^j \hat{X}_{x,\mu}\lambda^b
U^\dagger_{x-\mu,\mu} - U_{x-\mu,\mu}\lambda^b\hat{X}^\dagger_{x,\mu} n_x^j \hat{X}_{x,\mu}
U^\dagger_{x-\mu,\mu}\right)\right.\nonumber\\
&\phantom{space}\left.\left(U_{x-\mu,\mu}\hat{X}^\dagger_{x,\mu} n_x^j \hat{X}_{x,\mu}\lambda^a
U^\dagger_{x-\mu,\mu} - U_{x-\mu,\mu}\lambda^a\hat{X}^\dagger_{x,\mu} n_x^j \hat{X}_{x,\mu}
U^\dagger_{x-\mu,\mu} \right)\right]
-\nonumber\\
&2\tr\left[\left(U_{x-\mu,\mu}\hat{X}^\dagger_{x,\mu} n_x^j \hat{X}_{x,\mu}
U^\dagger_{x-\mu,\mu} -
n_{x-\mu}^j\right)\right.\nonumber\\
&\phantom{space}\left.\left(-U_{x-\mu,\mu}\lambda^b\hat{X}^\dagger_{x,\mu} n_x^j \hat{X}_{x,\mu}\lambda^a
U^\dagger_{x-\mu,\mu} + U_{x-\mu,\mu}\lambda^a\lambda^b\hat{X}^\dagger_{x,\mu} n_x^j \hat{X}_{x,\mu}
U^\dagger_{x-\mu,\mu} \right)\right] - \nonumber\\
&2\tr\left[\left(U_{x-\mu,\mu}\hat{X}^\dagger_{x,\mu} n_x^j \hat{X}_{x,\mu}
U^\dagger_{x-\mu,\mu} -
n_{x-\mu}^j\right)\right.\nonumber\\
&\phantom{space}\left.\left(U_{x-\mu,\mu}\hat{X}^\dagger_{x,\mu} n_x^j \hat{X}_{x,\mu}\lambda^b\lambda^a
U^\dagger_{x-\mu,\mu} - U_{x-\mu,\mu}\lambda^a\hat{X}^\dagger_{x,\mu} n_x^j \hat{X}_{x,\mu}\lambda^b
U^\dagger_{x-\mu,\mu} \right)\right],
\end{align}
and the molecular dynamics force can be constructed from the Newton-Raphson force using $G_x = \lambda^a G^a_x$.

\bibliographystyle{elsart-num.bst}
\bibliography{weyl}

\begin{thebibliography}{10}
\expandafter\ifx\csname url\endcsname\relax
  \def\url#1{\texttt{#1}}\fi
\expandafter\ifx\csname urlprefix\endcsname\relax\def\urlprefix{URL }\fi

\bibitem{Vortices}
G.~{'t Hooft}, Nucl. Phys. B190 (1981) 455.

\bibitem{Monopoles1}
G.~{'t Hooft}, Nucl. Phys. B79 (1974) 276.

\bibitem{Monopoles2}
A.~M. Polyakov, JETP Lett. 20 (1974) 194.

\bibitem{Mandelstam:1976}
S.~Mandelstam, Phys. Reports 23C (1976) 245.

\bibitem{thooft:1976}
G.~{'t Hooft}, in: A.~Zichichi (Ed.), {High Energy Physics}, Editrice
  Comprostrini, Bologna, 1976.

\bibitem{Cho:1980}
Y.~M. Cho, Phys. Rev. D 21 (1980) 1080.

\bibitem{Cho:1981}
Y.~M. Cho, Phys. Rev. D 23 (1981) 2415.

\bibitem{F-N:98}
L.~Faddeev, A.~Niemi, Phys. Rev. Lett. 82 (1999) 1624.

\bibitem{Shabanov:1999}
S.~Shabanov, Phys. Lett. B 458 (1999) 322.

\bibitem{Duan:1979}
Y.~Duan, M.~Ge, Sci. Sinica 11 (1979) 1072.

\bibitem{Kondo:2008su}
K.-I. Kondo, {Wilson loop and magnetic monopole through a non-Abelian Stokes
  theorem}, Phys. Rev. D77 (2008) 085029.

\bibitem{Shibata:2009af}
A.~Shibata, K.-I. Kondo, T.~Shinohara, {The Exact decomposition of gauge
  variables in lattice Yang-Mills theory}, Phys. Lett. B691 (2010) 91--98.

\bibitem{Kondo:2010pt}
K.-I. Kondo, A.~Shibata, T.~Shinohara, S.~Kato, {Non-Abelian Dual
  Superconductor Picture for Quark Confinement}, Phys. Rev. D83 (2011) 114016.

\bibitem{Kondo:2005eq}
K.-I. Kondo, T.~Murakami, T.~Shinohara, {Yang-Mills theory constructed from
  Cho-Faddeev-Niemi decomposition}, Prog.Theor.Phys. 115 (2006) 201--216.

\bibitem{Shibata:2007pi}
A.~Shibata, et~al., {Toward gauge independent study of confinement in SU(3)
  Yang-Mills theory}, POS LATTICE-2007 (2007) 331.

\bibitem{Cundy:2012ee}
N.~Cundy, W.~Lee, J.~Leem, {Y. M. Cho}, {String tension from gauge invariant
  Magnetic Monopoles}, PoS LATTICE2012 (2012) 213.

\bibitem{Cundy:2013pfa}
N.~Cundy, {Y. M. Cho}, W.~Lee, {Confinement From The Gauge Invariant Abelian
  Decomposition}, PoS LATTICE2013 (2013) 471.

\bibitem{Cundy:2013xsa}
N.~Cundy, {Y. M. Cho}, W.~Lee, J.~Leem, {The static quark potential from the
  gauge invariant Abelian decomposition}, Phys.Lett. B729 (2014) 192--198.

\bibitem{Wilson:1974}
K.~G. Wilson, Phys. Rev. D10 (1974) 2445.

\bibitem{Bali:2005fu}
G.~S. Bali, H.~Neff, T.~Duessel, T.~Lippert, K.~Schilling, {Observation of
  string breaking in QCD}, Phys.Rev. D71 (2005) 114513.

\bibitem{DP}
D.~Diakonov, V.~Petrov, Phys. Lett. B 224 (1989) 131.

\bibitem{PhysRevD.62.074009}
Y.~M. Cho, {Abelian dominance in Wilson loops}, Phys. Rev. D 62 (2000) 074009.
\newline\urlprefix\url{http://link.aps.org/doi/10.1103/PhysRevD.62.074009}

\bibitem{PhysRevD.62.025019}
M.~Faber, A.~N. Ivanov, N.~I. Troitskaya, M.~Zach, {Path integral
  representation for Wilson loops and the non-Abelian Stokes theorem}, Phys.
  Rev. D 62 (2000) 025019.
\newline\urlprefix\url{http://link.aps.org/doi/10.1103/PhysRevD.62.025019}

\bibitem{Cundy:2007df}
N.~Cundy, {Small Wilson Dirac operator eigenvector mixing in dynamical overlap
  hybrid Monte-Carlo}, Comput. Phys. Commun. 180 (2009) 180--191.

\bibitem{YangWu}
T.~T. Wu, C.~N. Yang, in: H.~Mark, S.~Fernbach (Eds.), {Properties of Matter
  under Unusual Conditions}, Interscience, New York, 1969.

\bibitem{tHooft1976}
G.~{'t Hooft}, Nucl. Phys. B79 (1976) 276.

\bibitem{Polyakov:1974}
A.~M. Polyakov, Pis'ma Zh. Eksp. Teor. Fiz. 20 (1974) 430.

\bibitem{TILW}
M.~L{\"u}scher, P.~Weisz, {}, Commun Math Phys 97 (1985) 59.

\bibitem{TILW2}
G.~P. Lepage, P.~B. Mackenzie, {Viability of lattice perturbation theory},
  Phys. Rev. D 48 (1993) 2250--2264.
\newline\urlprefix\url{http://link.aps.org/doi/10.1103/PhysRevD.48.2250}

\bibitem{TILW3}
M.~L{\"u}scher, P.~Weisz, {}, Phys. Lett. B. B158 (1985) 250.

\bibitem{TILW4}
J.~Snippe, {}, Nucl. Phys B498 (1997) 347.

\bibitem{TILW5}
G.~Curci, P.~Menotti, G.~Paffuti, {}, Phys. Lett. B B130 (1983) 205.

\bibitem{HMC}
S.~Duane, A.~Kennedy, B.~Pendleton, D.~Roweth, Phys. Lett. B195 (1987) 216.

\bibitem{Morningstar:2003gk}
C.~Morningstar, M.~J. Peardon, Phys. Rev. D69 (2004) 054501.

\bibitem{Moran:2008ra}
P.~J. Moran, D.~B. Leinweber, Phys. Rev. D77 (2008) 094501.

\bibitem{Luscher:2009eq}
M.~Luscher, {Trivializing maps, the Wilson flow and the HMC algorithm},
  Commun.Math.Phys. 293 (2010) 899--919.

\bibitem{Cundy:2014nna}
N.~Cundy, {Y. M. Cho}, W.~Lee, {Confinement, the Abelian Decomposition, and the
  Contribution of Topology to the Static Quark Potential}, PoS LATTICE2014
  (2014) 335.

\bibitem{Horvath:2003yj}
I.~Horvath, et~al., {Low-dimensional long-range topological charge structure in
  the QCD vacuum}, Phys. Rev. D68 (2003) 114505.

\bibitem{Horvath:2005rv}
I.~Horvath, A.~Alexandru, J.~Zhang, Y.~Chen, S.~Dong, et~al., {Inherently
  global nature of topological charge fluctuations in QCD}, Phys.Lett. B612
  (2005) 21--28.

\bibitem{Magnus:1954}
W.~Magnus, {On the exponential solution of differential equations for a linear
  operator.}, Comm. Pure Appl. Math. 7 (1954) 649--673.

\bibitem{Takaishi:2005tz}
T.~Takaishi, P.~de~Forcrand, {Testing and tuning new symplectic integrators for
  hybrid Monte Carlo algorithm in lattice QCD}, Phys. Rev. E73 (2006) 036706.

\bibitem{Omelyan}
I.~P. Omelyan, I.~M. Mrygloda, R.~Folk, {Symplectic analytically integrable
  decomposition algorithms: classification, derivation, and application to
  molecular dynamics, quantum and celestial mechanics simulations}, Computer
  Physics Communications 151 (2003) 272.

\end{thebibliography}

\end{document}